\newif\ifanonym
\anonymtrue
\anonymfalse 

\documentclass[11pt]{article}
\newif\ifprocs
\procstrue
\procsfalse 

\newif\ifarxiv
\arxivtrue

\newif\ifcomments
\commentstrue
\commentsfalse 

\usepackage{amsmath,amssymb,bbm,amsthm}
\usepackage{fullpage}
\usepackage{thm-restate,color,xspace}
\usepackage{comment}
\usepackage{thmtools}
\usepackage{graphicx}
\usepackage{tabularx}
\ifarxiv\else
\graphicspath{{Figures/}}
\fi

\usepackage{xcolor}
\usepackage{nameref}
\definecolor{ForestGreen}{rgb}{0.1333,0.5451,0.1333}
\definecolor{DarkRed}{rgb}{0.65,0,0}
\definecolor{Red}{rgb}{1,0,0}
\usepackage[linktocpage=true,
pagebackref=true,colorlinks,
linkcolor=DarkRed,citecolor=ForestGreen,
bookmarks,bookmarksopen,bookmarksnumbered]
{hyperref}
\usepackage{cleveref}

\usepackage{xspace}

\usepackage{caption}
\usepackage[ruled,vlined,linesnumbered]{algorithm2e}
\usepackage[noend]{algpseudocode}

\declaretheorem[numberwithin=section]{theorem}
\declaretheorem[numberlike=theorem]{lemma}
\declaretheorem[numberlike=theorem,name=Lemma]{lem}

\declaretheorem[numberlike=theorem]{proposition}

\declaretheorem[numberlike=theorem]{corollary}

\declaretheorem[numberlike=theorem]{claim}

\declaretheorem[numberlike=theorem,style=definition]{definition}

\declaretheorem[numberlike=theorem,name=Problem]{problem}

\newcommand{\ghtree}{{\sc GHtree}\xspace}
\newcommand{\apsp}{{\sc APSP}\xspace}
\newcommand{\apmf}{{\sc APMF}\xspace}
\newcommand{\tO}{\tilde{O}}
\newcommand{\polylog}{\mathrm{polylog}}
\def\eps{\varepsilon}
\DeclareMathOperator{\EX}{\mathbb{E}}

\global\long\def\Htil{\tilde{H}}

\global\long\def\P{{\cal P}}
\global\long\def\eps{\epsilon}
\global\long\def\opt{\mathrm{OPT}}

\global\long\def\Otil{\tilde{O}}
\global\long\def\HH{{\cal H}}

\global\long\def\pset{{\cal P}}
\global\long\def\val{\mathrm{val}}
\global\long\def\optval{\mathrm{pack}}

\global\long\def\sscv{\textsc{SSMC}\xspace}

\def\BAL#1\EAL{\begin{align*}#1\end{align*}}
\def\BALN#1\EALN{\begin{align}#1\end{align}}
\def\BG#1\EG{\begin{gather}#1\end{gather}}

\ifcomments

\newcommand{\alert}[1]{\textcolor{red}{#1}}

\def\jason#1{\marginpar{$\leftarrow$\fbox{JL}}\footnote{$\Rightarrow$~{\sf\textcolor{blue}{#1 --Jason}}}}

\def\thatchaphol#1{\marginpar{$\leftarrow$\fbox{TS}}\footnote{$\Rightarrow$~{\sf\textcolor{purple}{#1 --Thatchaphol}}}}

\def\debmalya#1{\marginpar{$\leftarrow$\fbox{DP}}\footnote{$\Rightarrow$~{\sf\textcolor{orange}{#1 --Debmalya}}}}

\def\amir#1{\marginpar{$\leftarrow$\fbox{AA}}\footnote{$\Rightarrow$~{\sf\textcolor{teal}{#1 --Amir}}}}

\def\ohad#1{\marginpar{$\leftarrow$\fbox{OT}}\footnote{$\Rightarrow$~{\sf\textcolor{red}{#1 --Ohad}}}}

\def\robi#1{\marginpar{$\leftarrow$\fbox{RK}}\footnote{$\Rightarrow$~{\sf\textcolor{violet}{#1 --Robi}}}}

\newcommand{\colnote}[3]{\textcolor{#1}{$\ll$\textsf{#2}$\gg$\marginpar{\tiny\bf #3}}}

 \newcommand{\onote}[1]{\colnote{red}{#1--Ohad}{OT}}
 \newcommand{\anote}[1]{\colnote{olive}{#1--Amir}{AA}}
 \newcommand{\rnote}[1]{\colnote{blue}{#1--Robi}{RK}}

\else
\newcommand{\rnote}[1]{}
\newcommand{\anote}[1]{}
\newcommand{\onote}[1]{}
\newcommand{\alert}[1]{}
\newcommand{\ohad}[1]{}
\newcommand{\robi}[1]{}
\newcommand{\thatchaphol}[1]{}
\newcommand{\amir}[1]{}
\newcommand{\debmalya}[1]{}
\newcommand{\jason}[1]{}
\fi 

\newcommand{\eat}[1]{}

\newcommand{\update}{\textsc{Update}}

\makeatletter
\newcounter{algocounter}
\@ifpackageloaded{hyperref}%
  {\newcommand{\mylabel}[2]
    {\refstepcounter{algocounter}\protected@write\@auxout{}{\string\newlabel{#1}{{\textcolor{black}{\textup{#2}}}{\thepage}%
      {\@currentlabelname}{\@currentHref}{}}}}}%
\makeatother

\newcommand{\lar}{\textup{large}}

\begin{document}

\title{Breaking the Cubic Barrier for All-Pairs Max-Flow:\\ Gomory-Hu Tree in Nearly Quadratic Time}

\ifanonym
\author{Anonymous Authors}
\else
\author{
Amir Abboud\thanks{Weizmann Institute of Science. Supported by an Alon scholarship and a research grant from the Center for New Scientists at the Weizmann Institute of Science. Email: \texttt{amir.abboud@weizmann.ac.il}}
\and Robert Krauthgamer%
\thanks{Weizmann Institute of Science. 
    Work partially supported by ONR Award N00014-18-1-2364,
    the Israel Science Foundation grant \#1086/18,
    the Weizmann Data Science Research Center,
    and a Minerva Foundation grant.
    Email: \texttt{robert.krauthgamer@weizmann.ac.il}
}
\and Jason Li\thanks{Simons Institute for the Theory of Computing, University of California, Berkeley. Email: \texttt{jmli@cs.cmu.edu}}
\and Debmalya Panigrahi\thanks{Duke University. Supported in part by NSF Awards CCF-1750140 (CAREER) and CCF-1955703, and ARO Award W911NF2110230. Email: \texttt{debmalya@cs.duke.edu}}
\and Thatchaphol Saranurak\thanks{University of Michigan, Ann Arbor. Email: \texttt{thsa@umich.edu}}
\and Ohad Trabelsi\thanks{University of Michigan, Ann Arbor.
	Work partially supported by the NSF Grant CCF-1815316, and by the NWO VICI grant 639.023.812. Email: \texttt{ohadt@umich.edu}}
}
\fi


\maketitle


\begin{abstract}

In 1961, Gomory and Hu showed that the All-Pairs Max-Flow problem of computing the max-flow between all ${n\choose 2}$ pairs of vertices in an undirected graph can be solved using only $n-1$ calls to any (single-pair) max-flow algorithm.
Even assuming a linear-time max-flow algorithm, this yields a running time of $O(mn)$, which is $O(n^3)$ when $m = \Theta(n^2)$. While subsequent work has improved this bound for various special graph classes, no subcubic-time algorithm has been obtained in the last 60 years for general graphs. We break this longstanding barrier by giving an $\tilde{O}(n^{2})$-time algorithm on general, weighted graphs.
Combined with a popular complexity assumption, we establish a counter-intuitive separation: all-pairs max-flows are strictly \emph{easier} to compute than all-pairs shortest-paths.

Our algorithm produces a cut-equivalent tree, known as the Gomory-Hu tree, from which the max-flow value for any pair can be retrieved in near-constant time.
For unweighted graphs, we refine our techniques further to produce a Gomory-Hu tree in the time of a poly-logarithmic number of calls to any max-flow algorithm.
This shows an equivalence between the all-pairs and single-pair max-flow problems, and is optimal up to poly-logarithmic factors.
Using the recently announced $m^{1+o(1)}$-time max-flow algorithm (Chen {\em et al.}, March 2022), our Gomory-Hu tree algorithm for unweighted graphs also runs in $m^{1+o(1)}$-time.



\medskip
\paragraph{Historical note:} 
The first version of this paper (\href{https://arxiv.org/abs/2111.04958v1}{arXiv:2111.04958}) 
titled ``Gomory-Hu Tree in Subcubic Time'' (Nov.~9, 2021) broke the cubic barrier 
but only claimed a time bound of $\tilde{O}(n^{2.875})$. 
The second version (Nov.~30, 2021) optimized one of the ingredients 
(Section~\ref{sec:singlesource}) and gave the $\tilde{O}(n^2)$ time bound. 
The latter optimization was discovered independently by Zhang~\cite{Zhang21b}.

\end{abstract}

\pagenumbering{gobble}

\clearpage
\tableofcontents
\pagenumbering{gobble}
\clearpage
\pagenumbering{arabic}
\section{Introduction}

The {\em edge connectivity} of a pair of vertices $s, t$ in an undirected graph is defined as the minimum weight of edges whose removal disconnects $s$ and $t$ in the graph. Such a set of edges is called an $(s,t)$ mincut, and by duality, its value is equal to that of an $(s,t)$ max-flow. Consequently, the edge connectivity of a vertex pair is obtained by running a max-flow algorithm, and by extension, the edge connectivity for {\em all} vertex pairs can be obtained by ${n\choose 2} = \Theta(n^2)$ calls to a max-flow algorithm. 
(Throughout, $n$ and $m$ denote the number of vertices and edges in the input graph $G = (V, E, w)$, where $w: E\rightarrow \mathbb{Z}^0_+$ maps edges to non-negative integer weights. We denote the maximum edge weight by $W$.)

\begin{definition}[The All-Pairs Max-Flow (\apmf) Problem]
Given an undirected edge-weighted graph, return the edge connectivity of all pairs of vertices.
\end{definition}

Remarkably, Gomory and Hu~\cite{GH61} showed in a seminal work in 1961 that one can do a lot better than this na\"ive algorithm. In particular, they introduced the notion of a {\em cut tree} (later called {\em Gomory-Hu tree}, which we abbreviate as \ghtree) 
to show that $n-1$ max-flow calls suffice for finding the edge connectivity of all vertex pairs. 

\begin{theorem}[Gomory-Hu (1961)]
\label{thm:GH}
For any undirected edge-weighted graph $G=(V,E)$, there is a cut tree (or \ghtree), which is defined as a tree $\mathcal{T}$ on the same set of vertices $V$ such that for all pairs of vertices $s,t\in V$, the $(s,t)$ mincut in $\mathcal{T}$ is also an $(s,t)$ mincut in $G$ and has the same cut value. 
Moreover, such a tree can be computed using $n-1$ max-flow calls.%
\footnote{These max-flow calls are on graphs that are contractions of $G$, and thus no larger than $G$.}
\end{theorem}

Since their work, substantial effort has gone into obtaining better \ghtree algorithms, and faster algorithms are now known for many restricted graph classes, including unweighted graphs~\cite{BHKP07, KL15, AKT20}, simple graphs~\cite{AKT21_stoc,AKT21_focs,LPS21,Zhang22,AKT22_soda}, planar graphs~\cite{BSW15}, surface-embedded graphs~\cite{BENW16}, bounded treewidth graphs~\cite{ACZ98, AKT20_b}, and so on (see Table~\ref{table:algs} and the survey~\cite{Panigrahi16}). Indeed, \ghtree algorithms are part of standard textbooks in combinatorial optimization (e.g.,~\cite{AhujaMO93, CookCPS97, Schrijver03}) and have numerous applications in diverse areas such as networks~\cite{Hu74}, image processing~\cite{WL93}, and optimization~\cite{PR82}. They have also inspired entire research directions as the first example of a {\em sparse representation} of graph cuts, the first non-trivial {\em global min-cut} algorithm, the first use of {\em submodular minimization} in graphs, and so forth. 

In spite of this attention, Gomory and Hu's 60-year-old algorithm has remained the state of the art for constructing a \ghtree in general, weighted graphs (or equivalently for \apmf, due to known reductions~\cite{AKT20_b,LPS21} showing that any APMF algorithm must essentially construct a \ghtree). Even if we assume an optimal $O(m)$-time max-flow algorithm, the Gomory-Hu algorithm takes $O(mn)$ time, which is $O(n^3)$ when $m = \Theta(n^2)$. 
Breaking this cubic barrier for the \ghtree problem  has been one of the outstanding open questions in the graph algorithms literature. 

In this paper, we break this longstanding barrier
by giving a \ghtree algorithm that runs in $\tO(n^2)$-time for general, weighted graphs.

\begin{theorem}
\label{thm:main}
There is a randomized Monte Carlo algorithm for the \ghtree (and \apmf) problems
that runs in $\tilde{O}(n^{2})$ time in general, weighted graphs.
\end{theorem}

\paragraph{Remarks:}
1. As noted earlier (and similar to state-of-the-art max-flow algorithms), we assume throughout the paper that edge weights are integers in the range $\{1, 2, \ldots, W\}$. Throughout, the notation $\tO(\cdot)$ hides poly-logarithmic factors in $n$ and $W$.\\ 
2. Our result is unconditional, i.e., it does not need to assume a (near/almost) linear-time max-flow algorithm. We note that concurrent to our work, an almost-linear time max-flow algorithm has been announced~\cite{ChenKLPGS22}. Our improvement of the running time of \ghtree/\apmf is independent of this result: even with this result,
the best \ghtree/\apmf bound was $m^{1+o(1)} n$ which is between $n^{2+o(1)}$ and $n^{3+o(1)}$ depending on the value of $m$, and we improve it to $\tO(n^2)$. Moreover, we stress that we do not need any recent advancement in max-flow algorithms for breaking the cubic barrier: even using the classic Goldberg-Rao max-flow algorithm~\cite{GR98} 
in our (combinatorial) algorithm solves \ghtree/\apmf in subcubic time. 

\medskip

Our techniques also improve the bounds known for the \ghtree problem in {\em unweighted graphs}, and even for {\em simple graphs}.
For unweighted graphs, the best previous results were $\tO(mn)$ obtained by Bhalgat {\em et al.}~\cite{BHKP07} and by Karger and Levine~\cite{KL15}, and an incomparable result that reduces the \ghtree problem to $O(\sqrt{m})$ max-flow calls~\cite{AKT20}.
There has recently been
much interest and progress on \ghtree in simple graphs as well~\cite{AKT21_stoc,AKT21_focs,LPS21,Zhang22,AKT22_soda}, with the current best running time being $(m + n^{1.9})^{1+o(1)}$.

We give a reduction of the \ghtree problem in unweighted graphs to $\polylog(n)$ calls of any max-flow algorithm. Note that this reduction is nearly optimal (i.e., up to the poly-log factor) since the all-pairs max-flow problem is at least as hard as finding a single-pair max-flow. Using the recent $m^{1+o(1)}$-time max-flow algorithm~\cite{ChenKLPGS22}, this yields a running time of $m^{1+o(1)}$ for the \ghtree problem in unweighted graphs. 
\begin{theorem}
\label{thm:unweighted}
There is a randomized Monte Carlo algorithm for the \ghtree problem that runs in $m^{1+o(1)}$ time in unweighted graphs.
\end{theorem}

\paragraph{APMF vs APSP.}
Our results deliver a surprising message to a primordial question in graph algorithms:  \emph{What is easier to compute, {shortest paths} or {max-flows}?} 
Ignoring $n^{o(1)}$ factors, the single-pair versions are both solvable in linear-time and therefore equally easy; albeit, the shortest path algorithm \cite{dijkstra1959note} is classical, elementary, and fits on a single page, whereas the max-flow algorithm \cite{ChenKLPGS22} is very recent, highly non-elementary, and requires more than a hundred pages to describe and analyze.
This and nearly all other evidence had supported the consensus that max-flows are at least as hard as (if not strictly harder than) shortest paths, and perhaps this can be established by looking at the more general \emph{all-pairs} versions: \apmf and \apsp (All-Pairs Shortest-Paths).
Much effort had gone into proving this belief (\apmf $\geq$ \apsp) using the tools of fine-grained complexity~\cite{AVY15,KT18,A+18,AKT20,AKT20_b}, with limited success: it was shown that (under popular assumptions) \apmf is strictly harder than \apsp in \emph{directed} graphs, but the more natural undirected setting remained open.
The first doubts against the consensus were raised in the aforementioned $n^{2+o(1)}$ algorithms for \apmf in simple (\emph{unweighted}) graphs that go below the $n^{\omega}$ bound of \apsp \cite{seidel1995all} (where $2\leq \omega < 2.37286$ \cite{AlmanW20} denotes the fast matrix multiplication exponent).
But if, as many experts believe, $\omega=2+o(1)$ then the only conclusion is that \apmf and \apsp are equally easy in simple graphs.
In general (\emph{weighted}) graphs, however, one of the central conjectures of fine-grained complexity states that the cubic bound for \apsp cannot be broken (even if $\omega=2$).
Under this ``\apsp Conjecture'', Theorem~\ref{thm:main} proves that \apmf is \emph{strictly easier} than \apsp!
Alternatively, if one still believes that \apmf $\geq$ \apsp, then our paper provides strong evidence against the \apsp Conjecture and against the validity of the dozens of lower bounds that are based upon it (e.g., \cite{RZ04,VW18,AbboudW14,AbboudGW15,Saha15,AbboudD16,bringmann2020tree}) or upon stronger forms of it (e.g., \cite{backurs2017improving,BackursDT16,cygan2019problems,AbboudCK20,GawrychowskiMW21}).


\subsection{Related Work}
\label{sec:related}

\paragraph*{Algorithms.}
Before this work, the time complexity of constructing a Gomory-Hu tree in general graphs has improved over the years only due to improvements in max-flow algorithms.
An alternative algorithm for the problem was discovered by Gusfield~\cite{Gusfield90}, where the $n-1$ max-flow queries are made on the original graph (instead of on contracted graphs).
This algorithm has the same worst-case time complexity as Gomory-Hu's, but may perform better in practice \cite{GT01}.
Many faster algorithms are known for special graph classes or when allowing a $(1+\eps)$-approximation,
see Table~\ref{table:algs} for a summary.
Moreover, a few heuristic ideas for getting a subcubic complexity in social networks and web graphs have been investigated~\cite{akiba2016cut}. 

\begin{table}[t]
\centering

\begin{footnotesize}
\begin{tabular}{l c l}
\hline
\hline
Restriction & Running time & Reference  \\ [0.5ex]
\hline
\textbf{General} & $(n-1)\cdot T(n,m)$ & \textbf{Gomory and Hu}~\cite{GH61}\\[1ex]
Bounded Treewidth* & $\tO(n)$ &Arikati, Chaudhuri, and Zaroliagis~\cite{ACZ98} \\
Unweighted & $\tO(mn)$ & Karger and Levine~\cite{KL15} \\
Unweighted & $\tO(mn)$ & {Bhalgat, Hariharan, Kavitha, and Panigrahi~}\cite{BHKP07} \\
Planar & $\tO(n)$ & Borradaile, Sankowski, and Wulff-Nilsen~\cite{BSW15} \\
Bounded Genus & $ \tO(n)$ & {Borradaile, Eppstein, Nayyeri, and Wulff-Nilsen}~\cite{BENW16} \\ 
Unweighted & $\tilde{O}(\sqrt{m}) \cdot T(n,m)$ 
& Abboud, Krauthgamer, and Trabelsi~\cite{AKT20} \\
$(1+\eps)$-Approx* & $\tO(n^{2})$ & Abboud, Krauthgamer, and Trabelsi~\cite{AKT20_b} \\
Bounded Treewidth & $\tO(n)$ & Abboud, Krauthgamer, and Trabelsi~\cite{AKT20_b} \\
Simple & $\tO(n^{2.5})$ & Abboud, Krauthgamer, and Trabelsi~\cite{AKT21_stoc} \\
$(1+\eps)$-Approx & $\polylog(n) \cdot T(n,m)$ & Li and Panigrahi~\cite{LP21} \\
Simple & $\tO(n^2)$ & $\dagger$ Abboud, Krauthgamer, and Trabelsi~\cite{AKT21_focs} \\
Simple & $\tO(n^2)$ & $\dagger$ Li, Panigrahi, and Saranurak~\cite{LPS21} \\
Simple & $\tO(n^{2 \frac{1}{8}})$ & $\dagger$ Zhang~\cite{Zhang22} \\
Simple & $\tO(m+n^{1.9})$ & Abboud, Krauthgamer, and Trabelsi~\cite{AKT22_soda} \\[1ex]
\textbf{General} & $\tilde{O}(n^{2})$ & \textbf{\Cref{thm:main}} \\[1ex]
\textbf{Unweighted} & $m^{1+o(1)} + \polylog(n)\cdot T(n, m)$ & \textbf{\Cref{thm:unweighted}} \\[1ex]
\hline
\end{tabular}
\caption{Algorithms for constructing a data-structure that answers $s,t$-min-cut queries in $\tilde{O}(1)$ time (listed chronologically). 
Except for those marked with *, they all also produce a Gomory-Hu tree.
$T(n,m)$ denotes the time to compute $s,t$-max-flow in an undirected graph, which using Chen~{\em et al.}~\cite{ChenKLPGS22} is $m^{1+o(1)}$.
The three results marked with $\dagger$ were obtained concurrently and independently of each other.
}
\label{table:algs}

\end{footnotesize}
\end{table}

\paragraph*{Hardness Results.}
The attempts at proving conditional lower bounds for All-Pairs Max-Flow have only succeeded in the harder settings of directed
graphs~\cite{KT18, A+18} or undirected graphs with vertex weights~\cite{AKT20}, where Gomory–Hu trees cannot even exist~\cite{May62,Jel63,HL07}. In particular, SETH gives an $n^{3-o(1)}$ lower bound for weighted sparse directed graphs~\cite{KT18} and the $4$-Clique conjecture gives an $n^{\omega+1-o(1)}$ lower bound for unweighted dense directed graphs~\cite{A+18}.

\paragraph*{Applications.}
Gomory-Hu trees have appeared in many application domains. We mention a few examples:
in mathematical optimization for the $b$-matching problem~\cite{PR82} (and that have been used in a breakthrough NC algorithm for perfect matching in planar graphs \cite{AnariV20});
in computer vision~\cite{WL93}, leading to the \emph{graph cuts} paradigm;
in telecommunications~\cite{Hu74} where there is interest in characterizing which graphs have a Gomory-Hu tree that is a \emph{subgraph} \cite{korte2012combinatorial,Naves18}.
The question of how the Gomory-Hu tree changes with the graph has arisen in applications such as energy and finance and has also been investigated, e.g. \cite{elmaghraby1964sensitivity,picard1980structure,barth2006revisiting,hartmann2013dynamic,BGK20esa}.

\subsection{Overview of Techniques}
\label{sec:overview}

We now introduce the main technical ingredients used in our algorithm, and explain how to put them together to prove \Cref{thm:main} and \Cref{thm:unweighted}.

\paragraph{Notation.}
In this paper, a \emph{graph} $G$ is an undirected graph $G=(V,E, w)$ with edge weights 
$w(e) \in \{1, 2, \ldots, W\}$ for all $e\in E$. 
If $w(e)=1$ for all $e\in E$, we say that $G$ is unweighted. 
The total weight of an edge set $E'\subseteq E$ is defined as $w(E')=\sum_{e\in E'}w(e)$.
For a cut $(S,V\setminus S)$, we also refer to a side $S$ of this cut as a cut. 
The \emph{value of cut $S$} is denoted $\delta(S)=w(E(S,V\setminus S))$. 
For any two vertices $a,b$, we say that $S$ is an $(a,b)$-cut 
if $|S\cap \{a, b\}| = 1$.
An \emph{$(a,b)$-mincut} is an $(a,b)$-cut of minimum value,
and we denote its value by $\lambda(a,b)$.


\paragraph{Reduction to Single-Source Minimum Cuts.}
The classic Gomory-Hu approach to solving APMF is to recursively solve
$(s,t)$ mincut problems on graphs obtained by contracting portions of the
input graph. This leads to $n-1$ max-flow calls on graphs that
cumulatively have $O(mn)$ edges. Recent work~\cite{AKT20_b} has
shown that replacing $(s,t)$ mincuts by a more powerful gadget of 
single-source mincuts reduces the cumulative size of the contracted
graphs to only $\tO(m)$. But, how do we solve the single-source mincuts
problem? Prior to our work, a subcubic algorithm was only known for 
simple graphs~\cite{AKT21_stoc,AKT21_focs,LPS21,Zhang22,AKT22_soda}. 
Unfortunately, if applied to non-simple graphs,
these algorithms become incorrect, and 
not just inefficient. 

Conceptually, our main 
contribution is to give an $\tO(n^2)$-time algorithm for the single 
source mincuts problem in general weighted graphs.
For technical reasons, however, we will further restrict this 
problem in two ways: 
(1) the algorithm (for the single-source problem) only needs to 
return the values $\lambda(s,t)$ for some terminals $t\in U\setminus \{s\}$, 
and (2) the mincut values $\lambda(s,t)$ for the terminals $t\in U\setminus \{s\}$ are guaranteed
to be within a $1.1$-factor of each other.\footnote{The value $1.1$ is arbitrary and can be replaced by any suitably small constant greater than $1$.}

We now state a reduction from \ghtree to this restricted single-source problem.
Let $U \subseteq V$ be a set of terminal vertices. The \emph{$U$-Steiner connectivity/mincut} is $\lambda(U)=\min_{a,b\in U}\lambda(a,b)$. 
The restricted single-source problem is defined below.

\begin{problem}[Single-Source Terminal Mincuts]\label{problem:ssmc}
The input is a graph $G=(V,E,w)$, a terminal set $U \subseteq V$ and a source terminal $s\in U$ 
with the promise that 
for all $t\in U\setminus\{s\}$, we have $\lambda(U) \le \lambda(s,t) \le 1.1\lambda(U)$. 
The goal is to determine the value of $\lambda(s,t)$ for each terminal $t\in U\setminus\{s\}$.
\end{problem}

The reduction has two high-level steps. First, we reduce the single-source terminal mincuts problem \emph{without the promise that $\lambda(s,t) \in [\lambda(U), 1.1 \lambda(U)]$} (we define this as \Cref{problem:ssmc-no-promise} in \Cref{appendix:reduction}) to the corresponding problem with the promise (i.e., \Cref{problem:ssmc}) by calling an approximate single-source mincuts algorithm of Li and Panigrahi~\cite{LP21}. 
Then, we use a reduction from Gomory-Hu tree to the single-source terminal mincuts without the promise (i.e., \Cref{problem:ssmc-no-promise}) that was presented by Li~\cite{Li21thesis}.\footnote{The actual reduction is slightly stronger in the sense that it only requires a ``verification'' version of single-source terminal mincuts, but we omit that detail for simplicity.} We present both steps of the reduction in Appendix~\ref{appendix:reduction}.


\begin{restatable}[Reduction to Single-Source Terminal Mincuts]{lemma}{Reduction}\label{lem:reduction2}
There is a randomized algorithm that computes a \ghtree of an input graph by making calls to max-flow and single-source terminal mincuts (with the promise, i.e., \Cref{problem:ssmc}) on graphs with a total of $\tilde{O}(n)$ vertices and $\tilde{O}(m)$ edges, and runs for $\tilde{O}(m)$ time outside of these calls.
\end{restatable}

\paragraph{Guide Trees.} The main challenge, thus, is to solve single-source terminal mincuts (\Cref{problem:ssmc}) faster than $n-1$ max-flow calls.
%
Let us step back and think of a simpler problem: the {\em global mincut} problem. In a beautiful paper, Karger~\cite{Karger00} gave a two-step recipe for solving this problem by using the duality between cuts and tree packings.
First, by packing a maximum set of edge-disjoint \emph{spanning} trees in a graph and sampling one of them uniformly at random, the algorithm obtains a spanning tree that, with high probability, \emph{$2$-respects} the global mincut, meaning that only two edges from the tree cross the cut.
Second, a simple linear-time dynamic program computes the minimum value cut that $2$-respects the tree.
Can we use this approach? 


Clearly, we cannot hope to pack $\lambda(U)$ disjoint spanning trees since the global mincut value could be much less than $\lambda(U)$. 
But what about Steiner trees? A tree $T$ is called a \emph{$U$-Steiner tree} if it spans $U$, i.e., $U \subseteq V(T)$.
When $U$ is clear from the context, we write Steiner instead of $U$-Steiner.

First, we define the $k$-respecting property for Steiner trees.

\begin{definition}[$k$-respecting]\label{defn:k-respect}
Let $A\subseteq V$ be a cut in $G = (V, E, w)$. Let $T$ be a tree on (some subset of) vertices in $V$. We say that the tree $T$ \emph{$k$-respects} the cut $A$ (and vice versa) if $T$ contains at most $k$ edges with exactly one endpoint in $A$.
\end{definition}

Using this notion of $k$-respecting Steiner trees, we can now define a collection of guide trees that is analogous to a packing of spanning trees.

\begin{definition}[Guide Trees]
For a graph $G$ and set of terminals $U \subseteq V$ with a source $s \in U$,
a collection of $U$-Steiner trees $T_1,\dots,T_h$ is called a \emph{$k$-respecting set of guide trees}, or in short \emph{guide trees},
if for every $t \in U \setminus \{s\}$,
at least one tree $T_i$ $k$-respects some $(s,t)$-mincut in $G$. 
\end{definition}

Two questions immediately arise:
\begin{enumerate}
\item Can we actually obtain such $k$-respecting guide trees, for a small $k$ (and $h$)?
\item Can guide trees be used to speed up the single-source mincuts algorithm?
\end{enumerate}

The first question can be solved in a way that is conceptually (but not technically) similar to Karger's algorithm for global mincut.
We first prove, using classical tools in graph theory (namely, Mader's splitting-off theorem~\cite{Mader78}, and Nash-Williams~\cite{NW61} and Tutte's~\cite{Tut61} tree packing) that there exists a packing with $\lambda(U)/2$ edge-disjoint Steiner trees.
Then, we use the width-independent Multiplicative Weights Update (MWU) framework  \cite{GargK07,Fleischer00,AroraHK12} to pack a near-optimal number of Steiner trees using $\tilde{O}(m)$ calls to an (approximation) algorithm for the \emph{minimum Steiner tree} problem.
For the latter, we use Mehlhorn's $2$-approximation algorithm \cite{Mehlhorn88} that runs in $\tilde{O}(m)$ time, giving a packing of $\lambda(U)/4$ Steiner trees in $\tilde{O}(m^2)$ time. To speed this up,
we compute the packing in a $(1+\eps)$-cut-sparsifier of $G$~(e.g., \cite{BeK15}), which effectively reduces $m$ to $\tO(n)$ for this step.
Overall, this gives an $\tilde{O}(n^2)$-time algorithm for constructing $4$-respecting guide trees.

We note that our improved running time for unweighted graphs comes from replacing this algorithm for constructing guide trees by a more complicated algorithm.
Specifically, we show that \emph{all} of the $\tilde{O}(m)$ calls to (approximate) minimum Steiner tree during the MWU algorithm can be handled in a total of $m^{1+o(1)}$ time using a novel dynamic data structure that relies on (1) a non-trivial adaptation of Mehlhorn's reduction from minimum Steiner tree to Single-Source Shortest Paths and (2) a recent dynamic algorithm for the latter problem \cite{BGS21_arxiv}. 
This achieves running time $m^{1+o(1)}$ compared with $\tO(n^2)$ for unweighted graphs.

We summarize the construction of guide trees in the next theorem, which we prove in \Cref{sec:packing}. (The new dynamic data structure that is used in the improvement for unweighted graphs is given in \Cref{sec:unweighted}.)

\begin{theorem}[Constructing Guide Trees]\label{lem:guide}
There is a randomized algorithm that, given a graph $G=(V,E,w)$, a terminal set $U \subseteq V$ and a source terminal $s\in U$, with the guarantee that for all $t\in U\setminus\{s\}$, $\lambda(U) \le \lambda(s,t) \le 1.1\lambda(U)$, computes a $4$-respecting set of $O(\log{n})$ guide trees.
The algorithm takes $\Otil(n^2)$ time on weighted graphs (i.e., when $w(e)\in \{1, 2, \ldots, W\}$ for all $e\in E$) and $m^{1+o(1)}$ time on unweighted graphs (i.e., when $w(e)=1$ for all $e\in E$).
\end{theorem}

But, how do guide trees help? 
In the case of global mincuts, the tree is spanning, 
hence every $k$ tree edges define a partition of $V$, and also a cut in $G$.
Therefore, once the $k$-respecting property has been achieved, finding the best $k$-respecting cut is a search over at most $n^k$ cuts for any given tree, and can be done using dynamic programming for small $k$~\cite{Karger00}. 
In contrast, specifying the $k$ tree-edges that are cut 
leaves an exponential number of possibilities when $T$ is a Steiner tree
based on which side of the cut the vertices not in $T$ appear on.
In fact, in the extreme case where the Steiner tree is a single edge between two terminals $s$ and $t$, 
computing the $1$-respecting mincut is as hard as computing $(s,t)$-mincut. 

We devise a recursive strategy to solve the problem of obtaining $k$-respecting $(s,t)$-mincuts. First, we root the tree $T$ at a centroid, and recurse in each subtree (containing at most half as many vertices). We show that this preserves the $k$-respecting property for $(s,t)$-mincuts. However, in general, this is too expensive since the entire graph $G$ is being used in each recursive call, and there can be many subtrees (and a correspondingly large number of recursive calls). Nevertheless, we show that this strategy can be made efficient {\em when all the cut edges are in the same subtree} by an application of the Isolating Cuts Lemma from~\cite{LP20,AKT21_stoc}.

This leaves us with the case that the cut edges are spread across multiple subtrees. Here, we use a different recursive strategy. We use random sampling of the subtrees to reduce the number of cut edges, and then make recursive calls with smaller values of $k$. Note that this effectively turns our challenge in working with Steiner trees vis-\`a-vis spanning trees into an advantage; if we were working on spanning trees, sampling and removing subtrees would have violated the spanning property. This strategy works directly when there exists at least one cut edge in a subtree other than those containing $s$ and $t$; then, with constant probability, we remove this subtree but not the ones containing $s, t$ to reduce $k$ by at least $1$. The more tricky situation is if the cut edges are only in the subtrees of $s$ and $t$; this requires a more intricate procedure involving a careful relabeling of the source vertex $s$ using a Cut Threshold Lemma from \cite{LP21}.

The algorithm is presented in detail in \Cref{sec:singlesource}, and we state here its guarantees.



\begin{restatable}[Single-Source Mincuts given a Guide Tree]{theorem}{ssmc}\label{thm:ssmc}
Let $G=(V,E,w)$ be a weighted graph, let $T$ be a tree defined on (some subset of) vertices in $V$, and let $s$ be a vertex in $T$. 
For any fixed integer $k\ge 2$, there is a Monte-Carlo algorithm that finds, for each vertex $t \not= s$ in $T$, a value $\tilde\lambda(t) \ge \lambda(s,t)$ such that $\tilde\lambda(t)=\lambda(s,t)$ if $T$ is $k$-respecting an $(s,t)$-mincut. 
The algorithm takes $m^{1+o(1)}$ time.
\end{restatable}

\paragraph{Remarks:} The algorithm in \Cref{thm:ssmc} calls max-flow on instances of maximum number of $m$ edges and $n$ vertices and total number of $\tilde{O}(m)$ edges and $\tilde{O}(n)$ vertices, and spends $\tilde{O}(m)$ time outside these calls. The number of logarithmic factors hidden in the $\tilde O(\cdot)$ depends on $k$. Note that the running time of the algorithm is $m^{1+o(1)}$ even when $G$ is a weighted graph.

\paragraph{Putting it all together: Proof of \Cref{thm:main} and \Cref{thm:unweighted}}
The three ingredients above suffice to prove our main theorems.
By Lemma~\ref{lem:reduction2}, it suffices to solve the single-source mincut problem (\Cref{problem:ssmc}).
Given an instance of \Cref{problem:ssmc} on a graph $G$ with terminal set $U$, we use \Cref{lem:guide} to obtain a $4$-respecting set of $O(\log{n})$ guide trees. 
We call the algorithm in \Cref{thm:ssmc} for each of the $O(\log{n})$ trees separately and keep, for each $t \in U \setminus \{s\}$, 
the minimum $\tilde\lambda(s,t)$ found over all the $O(\log n)$ trees. 

The running time of the final algorithm equals that of max-flow calls on graphs with at most $O(m)$ edges and $O(n)$ vertices each, and total number of $\tilde{O}(m)$ edges and $\tilde{O}(n)$ vertices. In addition, the algorithm takes $\tilde{O}(n^2)$ time outside of these calls (in \Cref{lem:guide}); in unweighted graphs, the additional time is only $m^{1+o(1)}$.

\eat{

We summarize this in the theorem below.
\begin{theorem}
There is a Monte Carlo algorithm for constructing a Gomory-Hu tree of an undirected, weighted graph $G$ with $m$ edges and $n$ vertices
that runs in $\tilde{O}(n^2)$ time 
plus a set of max-flow calls whose cumulative running time is at most that of $\polylog(n)$ max-flow calls on graphs with $n$ vertices and $m$ edges each.
If the input $G$ is unweighted, the algorithm's running time outside these calls is only $m^{1+o(1)}$.
\end{theorem}
\debmalya{Do we need this last summarization theorem? I feel we can remove this.}
\amir{yes, something needs to be removed. this page repeats the ``max-flow calls on graphs with at most $O(m)$ edges and $O(n)$ vertices each'' point three times.}

}




\section{Single-Source Mincuts Given a Guide Tree}
\label{sec:singlesource}

In this section, we present our single-source mincut algorithm (\sscv) given a guide tree, which proves \Cref{thm:ssmc}. 



Before describing the algorithm, we state two tools we will need. 
The first is the \emph{Isolating-Cuts} procedure introduced by Li and Panigrahi \cite{LP20} and independently by Abboud, Krauthgamer, and Trabelsi \cite{AKT21_stoc}. (Within a short time span, this has found several interesting applications \cite{LP21,CQ21,MN21,LNPSS21,AKT21_focs,LPS21,Zhang22,AKT22_soda,CenLP22}.)

Recall that for a vertex set $S\subseteq V$, $\delta(S)$ denotes the total weight of edges with exactly one endpoint in $S$
(i.e., the value of the cut $(S, V\setminus S)$). For any two disjoint vertex sets $A,B\subseteq V$, we say that $S$ is an \emph{$(A,B)$-cut}
if $A\subseteq S$ and $B\cap S=\emptyset$ or $B\subseteq S$ and
$A\cap S=\emptyset$. In other words, the cut $S$ ``separates'' the vertex sets $A$ and $B$. 
We say that $S$ is an \emph{$(A,B)$-mincut} if it is
an $(A,B)$-cut of  minimum value, and let $\lambda(A,B)$ denote the
value of an $(A,B)$-mincut. As described earlier, 
if $A$ and $B$ are singleton sets, say $A = \{a\}$ and $B = \{b\}$,
then we use the shortcut $(a,b)$-mincut to denote an $(A,B)$-mincut, and use 
$\lambda(a,b)$ to denote the value of an $(a, b)$-mincut. 


We now state the isolating cuts lemma from \cite{LP20,AKT21_stoc}:
\begin{lem}[Isolating Cuts Lemma: Theorem 2.2 in \cite{LP20}, also follows from Lemma 3.4 in \cite{AKT21_stoc}]\label{lem:iso cut}
There is an algorithm that, given a graph $G=(V,E,w)$ and a collection $\cal{U}$ of disjoint terminal sets $U_1,\dots,U_h \subseteq V$, computes a $(U_i,\cup_{j\neq i} U_j)$-mincut for every $U_i \in\cal{U}$. The algorithm calls max-flow on graphs that cumulatively contain $O(m \log h)$ edges and $O(n \log h)$ vertices, and spends $\tO(m)$ time outside these calls. 
\end{lem}

\paragraph{Remark:} The isolating cuts lemma stated above slightly generalizes the corresponding statement from \cite{LP20,AKT21_stoc}. In the previous versions, each of the sets $U_1, U_2, \ldots, U_h$ is a distinct singleton vertex in $V$. The generalization to disjoint sets of vertices is trivial because we can contract each set $U_i$ for $i\in [h]$ and then apply the original isolating cuts lemma to this contracted graph to obtain \Cref{lem:iso cut}.

We call each $(U_i,\cup_{j\neq i} U_j)$-mincut $S_i$ a  \emph{minimum isolating cut} because it ``isolates'' $U_i$ from the rest of the terminal sets, using a cut of minimum size.
The advantage of this lemma is that it essentially only costs $O(\log{h})$ max-flow calls, which is an exponential improvement over the na\"ive strategy of running $h$ max-flow calls, one for each $U_i$.

The next tool is the \emph{Cut-Threshold} procedure of Li and Panigrahi, which has been used earlier in the approximate Gomory-Hu tree problem~\cite{LP21} and in edge connectivity augmentation and splitting off algorithms~\cite{CenLP22}.

\begin{lemma}[Cut-Threshold Lemma: Theorem 1.6 in \cite{LP21}]\label{lem:ct}
There is a randomized, Monte-Carlo algorithm that, given a graph $G=(V,E,w)$, a vertex $s\in V$, and a threshold $\overline\lambda$, computes all vertices $v\in V$ with $\lambda(s,v)\le\overline\lambda$ (recall that $\lambda(s,v)$ is the size of an $(s, v)$-mincut). The algorithm calls max-flow on graphs that cumulatively contain $\tilde{O}(m)$ edges and $\tilde{O}(n)$ vertices, and spends $\tilde{O}(m)$ time outside these calls.
\end{lemma}

We use the Cut-Threshold lemma to obtain the following lemma, which is an important component of our final algorithm.
\begin{lemma}\label{lem:cut-thr-step}
For any subset $U\subseteq V$ of vertices and a vertex $s\notin U$, there is a randomized, Monte-Carlo  algorithm that computes $\lambda_{\max}=\max\{ \lambda(s,t):t\in U\}$ as well as all vertices $t\in U$ attaining this maximum, i.e., the vertex set $\arg\max_{t\in U}\{\lambda(s, t)\}$. The algorithm calls max-flow on graphs that cumulatively contain $\tilde{O}(m)$ edges and $\tilde{O}(n)$ vertices, and spends $\tilde{O}(m)$ time outside these calls.
\end{lemma}
\begin{proof}
We binary search for the value of $\lambda_{\max}$. For a given estimate $\overline\lambda$, we call the Cut-Threshold Lemma (\Cref{lem:ct}) with this value of $\overline\lambda$; if the procedure returns a set containing all vertices in $U$, then we know $\lambda_{\max} \le \overline\lambda$; otherwise, we have $\lambda_{\max}>\overline\lambda$. A simple binary search recovers the exact value of $\lambda_{\max}$ in $O(\log nW)$ iterations since edge weights are integers in $\{1,2,\ldots,W\}$. Finally, we call the Cut-Threshold Lemma with $\overline\lambda=\lambda_{\max}-1$; we remove the vertices returned by this procedure from $U$ to obtain all vertices $t\in U$ satisfying $\lambda(s, t) = \lambda_{\max}$. For the running time bound, note that by \Cref{lem:ct}, each iteration of the binary search calls max-flow on graphs that cumulatively contain $\tO(n)$ vertices and $\tO(m)$ edges, and uses $\tO(m)$ time outside these calls.
\end{proof}

\paragraph{The \sscv Algorithm.}
Having introduced the main tools, we are now ready to present our \sscv algorithm (see Figure~\ref{Figs:Proof_single_alg}). 
The input to the algorithm is a graph $G=(V,E,w)$ containing a specified vertex $s$, a (guide) tree $T$ containing $s$, and a positive integer $k$. 
The algorithm is a recursive algorithm, and although the guide tree initially only contains vertices in $V$, there will be additional vertices (not in $V$) that are introduced into the guide tree in subsequent recursive calls.
To distinguish between these two types of vertices, we define $R(T)$ as the subset of vertices of $T$ that are in $V$, and call these {\em real} vertices. We call the vertices of $T$ that are not in $V$ {\em fake} vertices. 

We extend the definition of $k$-respecting (i.e., \Cref{defn:k-respect}) to fake vertices as follows: 
\begin{definition}[Generalized $k$-respecting]\label{defn-k-respect-fake}
Let $A\subseteq V$ be a cut in $G = (V, E, w)$. Let $T$ be a tree on (some subset of) vertices in $V$ as well as additional vertices not in $V$. We say that $T$ \emph{$k$-respects} cut $A$ (and vice versa) if there exists a set $F_A$ of fake vertices such that $T$ contains at most $k$ edges with exactly one endpoint in $A\cup F_A$; we say that such edges are \emph{cut} by $A\cup F_A$. 
\end{definition}
We also note that even if all the vertices in $T$ are real vertices, $T$ may not be a subgraph of $G$.

\bigskip

Recall that our goal is to obtain a value $\tilde\lambda( t) \ge \lambda(s, t)$ for every terminal $t\in U\setminus \{s\}$ such that if an $(s, t)$-mincut $k$-respects $T$, then $\tilde\lambda(t) = \lambda(s, t)$. We will actually compute $\tilde\lambda(t)$ for every real vertex $t\in R(T)\setminus \{s\}$; clearly, this suffices since the input Steiner tree (i.e., at the top level of the recursion) spans all the vertices in $U$.

The algorithm maintains estimates $\tilde\lambda(t)$ of the mincut values $\lambda(s,t)$ for all $t\in R(T) \setminus \{s\}$. The values $\tilde\lambda(t)$ are initialized to $\infty$, and whenever we compute an $(s,t)$-cut in the graph, we ``update'' $\tilde\lambda(t)$ by replacing $\tilde\lambda(t)$ with the value of the $(s,t)$-cut if it is lower. Formally, we define $\update(t, x): \tilde\lambda( t) \gets \min(\tilde\lambda( t), x)$.

We describe the algorithm below. The reader should use the illustration in \Cref{Figs:Proof_single_alg} as a visual description of each step of the algorithm.

\begin{enumerate}

\item
First, we describe a base case. If $|R(T)|$ is less than some fixed constant, then we simply compute the $(s,t)$-mincut in $G$ separately for each $t\in R(T)\setminus \{ s \}$ using $|R(T)|-1 = O(1)$ max-flow calls, and run $\update(t,\lambda(s,t))$. 
\label{step:1}

\item[]
From now on, assume that $|R(T)|$ is larger than some (large enough) constant.\footnote{For example, the constant $10$ is more than enough.} 



\item Let $c$ be a \emph{centroid} of the tree $T$, defined in the following manner: $c$ is a (possibly fake) vertex in $T$ such that if we root $T$ at $c$, then each subtree rooted at a child of $c$ has at most $|R(T)|/2$ 
real vertices.\footnote{A centroid always exists by the following simple argument: take the (real or fake) vertex of $T$ of maximum depth whose subtree rooted at $T$ has at least $|R(T)|/2$ real vertices. By construction, this vertex is a centroid of $T$, and it can be found in time linear in the number of vertices in the tree using a simple dynamic program.} \label{step:2}
\item[]If $c\in R(T)$ and $s\not= c$, then compute an $(s,c$)-mincut in $G$ (whose value is denoted $\lambda$) using a max-flow call and run $\update(c, \lambda(s,c))$.

\item Root $T$ at $c$ and let $u_1,\ldots,u_\ell$ be the children of $c$. For each $i\in[\ell]$, let $T_i$ be the subtree rooted at $u_i$. Recall that $R(T_i)$ denotes the set of real vertices in the respective subtrees $T_i$ for $i\in [\ell]$. (For technical reasons, we ignore subtrees $T_i$ that do not contain any real vertex.)  Use \Cref{lem:iso cut} to compute minimum isolating cuts in $G$ with the following terminal sets: (1) $U_i = R(T_i)$ for $i\in[\ell]$. (2) If $c\in R(T)$, then we add an additional set $U_{\ell+1} = \{c\}$. Note that $\cup_i U_i = R(T)$ irrespective of whether $c\in R(T)$ or not. 
\item[] Let $S_i\subseteq V$ be the $(U_i,R(T)\setminus U_i)$-mincut in $G$ obtained from \Cref{lem:iso cut}. We ignore $S_{\ell+1}$ (if it exists) and proceed with the remaining sets $S_i$ for $i\in [\ell]$ in the next step.\label{step:3}
\item For each $i\in [\ell]$, define $G_i$ as the graph $G$ with $V\setminus S_i$ contracted to a single vertex. Now, there are two cases. In the first case, we have $s\in  V\setminus S_i$. Then, the contracted vertex for $V\setminus S_i$ is labeled the new $s$ in graph $G_i$. Correspondingly, define $T'_i$ as the tree $T_i$ with an added edge $(s,u_i)$ (recall that $u_i$ is the root of $T_i$). In the second case, we have $s\in S_i$. Then, assign a new label $c_i$ to the contracted vertex for $V\setminus S_i$ in $G_i$. In this case, define $T'_i$ as the tree $T_i$ with an added edge $(c_i,u_i)$, and keep the identity of vertex $s$ unchanged since it is in $T_i$. (Note that if $s=c$, the only difference is that the second case does not happen for any $i\in [\ell]$.)
\item[]In both cases above, make recursive calls $(G_i,T'_i,k)$ for all $i\in[\ell]$. Call $\update(t,\lambda'(s,t))$ for all $t\in R(T_i)\setminus \{ s \}$ where the recursive call returns the value $\lambda'(s, t)$ for the variable $\tilde\lambda(t)$. Furthermore, if $s\in S_i$, call $\update(t,\lambda'(s,c_i))$ for all $t\in R(T)\setminus R(T_i)$ where the recursive call returns the value $\lambda'(s,c_i)$ for the variable $\tilde\lambda(c_i)$. 
\item[] If $k=1$, then we terminate the algorithm at this point, so from now on, assume that $k>1$.\label{step:4}
\item Sample each subtree $T_i$ independently with probability $1/2$ except the subtree containing $s$ (if it exists), which is sampled with probability $1$. (If $c=s$, then there is no subtree containing $s$, and all subtrees are sampled with probability $1/2$.) Let $T^{(\ref{step:5})}$ be the tree $T$ with all (vertices of) non-sampled subtrees deleted. Recursively call $(G,T^{(\ref{step:5})},k-1)$ and update $\tilde\lambda(t)$ for all $t\in R(T^{(\ref{step:5})})$. (Note that $R(T^{(\ref{step:5})})$ denotes the set of real vertices in tree $T^{(\ref{step:5})}$. Moreover, by the sampling procedure, $s$ is always in $R(T^{(\ref{step:5})})$ and hence, the recursion is valid.) Repeat this step for $O(\log n)$ independent sampling trials.\label{step:5}
\item Execute this step only if $s\ne c$, and let $T_s$ be the subtree from step (\ref{step:3}) containing $s$. Using \Cref{lem:cut-thr-step}, compute the value $\lambda_{\max}=\max\{\lambda(s,t):t\in R(T)\setminus R(T_s)\}$, as well as all vertices $t\in R(T)\setminus R(T_s)$ attaining this maximum. Update $\tilde\lambda(t)=\lambda_{\max}$ for all such $t$, and arbitrarily select one such $t$ to be labeled $s'$. Let $T^{(\ref{step:6})}$ be the tree $T$ with (the vertices of) subtree $T_s$ removed. Recursively call $(G,T^{(\ref{step:6})},k-1)$ where $s'$ is treated as the new $s$, and update $\tilde\lambda(t)$ for all $t\in R(T^{(\ref{step:6})})$.\label{step:6}

\end{enumerate}

\begin{figure}[h]
  \begin{center}
    \ifarxiv
    \includegraphics[width=6.6in]{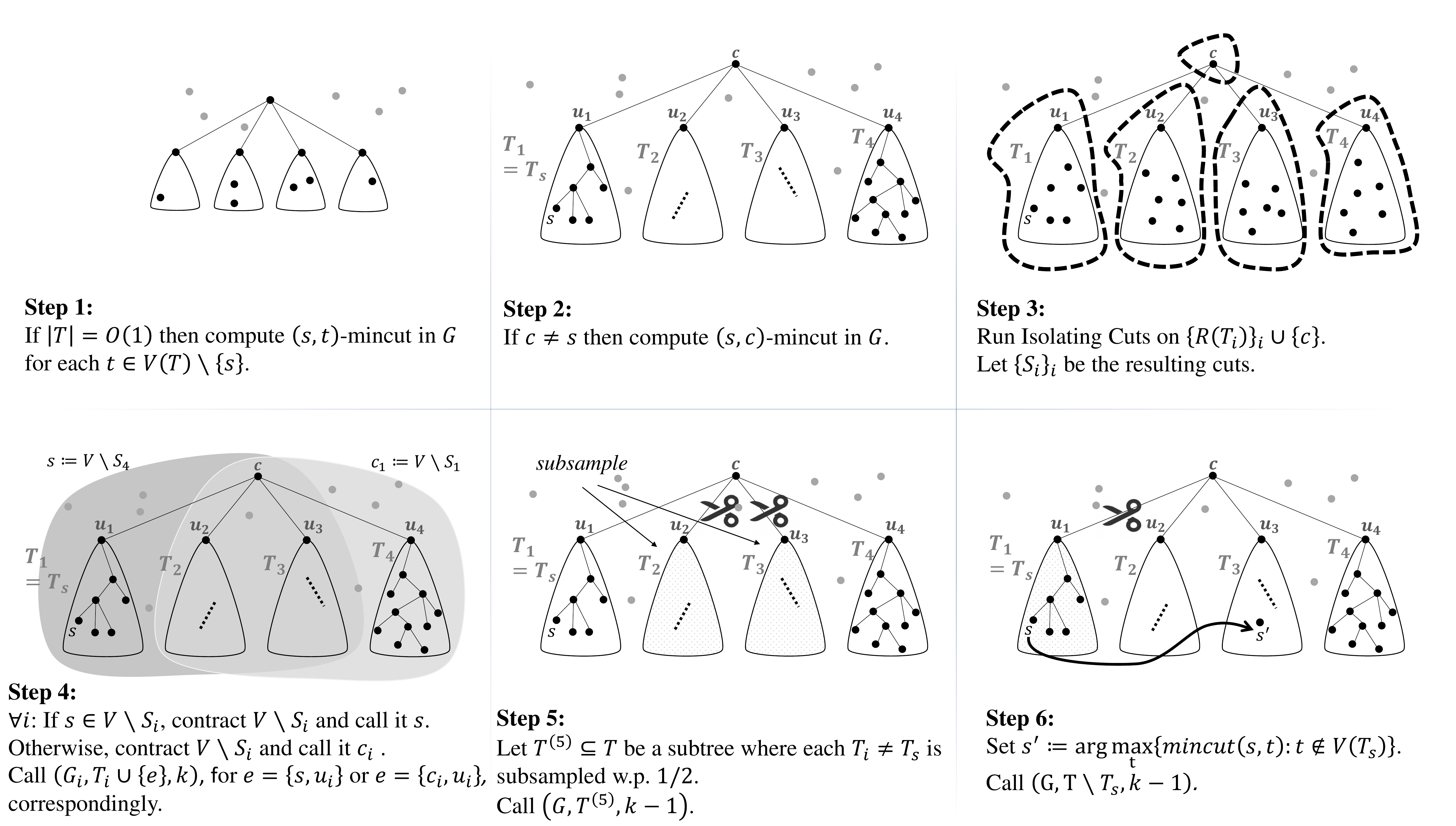}
    \else
    \includegraphics[width=6.6in]{Proof_single_alg}
    \fi
    \end{center}
    \caption{An illustration of the steps inside a recursive iteration of the \sscv algorithm.
    We assume that the centroid $c$ has four children in $T$ and that all tree vertices are real;  in particular $c \in R(T)$, which simplifies some of the steps. Graph vertices that are not spanned by $T$ are represented by gray dots. The gray areas in Step 4 refer to contracted subsets, and the scissors symbol in Steps 5 and 6 means we remove the subtree.}
    \label{Figs:Proof_single_alg}
\end{figure}

\subsection{Correctness}

First, we show a standard (uncrossing) property of mincuts.
\begin{lemma}\label{lem:uncross}
Let $G=(V,E,w)$ be a weighted, undirected graph with vertex subset $U\subseteq V$. For any subsets $\emptyset\subsetneq X\subseteq X'\subsetneq U$ and an $(X',U\setminus X')$-mincut $A'\subseteq V$ of $G$, there is an $(X,U\setminus X)$-mincut $A\subseteq V$ of $G$ satisfying $A\subseteq A'$.
\end{lemma}
\begin{proof}
Consider any $(X,U\setminus X)$-mincut $A\subsetneq A'$. We claim that $A\cap A'$ is also an $(X,U\setminus X)$-mincut of $G$. First, note that $(A\cup A')\cap U=X'$, so $A\cup A'$ is an $(X',U\setminus X')$-cut. Since $A'$ is an $(X', U\setminus X')$-mincut, we have
$$\delta(A\cup A')\ge\delta(A').$$ 
$$\delta(A)+\delta(A')\ge\delta(A\cup A')+\delta(A\cap A').$$ Combining the two inequalities gives $\delta(A\cap A')\le\delta(A)$.
Now, since $X \subseteq X'\subseteq A'$, we have $X\cap (A\setminus A') = \emptyset$. Since $X\subseteq A$, it must be  that $X\subseteq A\cap A'$. So, $A\cap A'$ is an $(X,U\setminus X)$-cut. Since $\delta(A\cap A')\le\delta(A)$, 
it follows that $A\cap A'$ is an $(X,U\setminus X)$-mincut, which completes the proof.
\end{proof}

Now, we proceed to establish correctness of the \sscv algorithm. Note that $\tilde\lambda(t)$ starts with the value $\infty$, and every time we run $\update(t,x)$, we have that $x$ is the value of some $(s,t)$-cut in $G$. Naturally, this would suggest that our estimate $\tilde\lambda(t)$ is always an upper bound on the true value $\lambda(s,t)$. 
However, this is not immediately clear because the vertex $s$ may be relabeled in a recursive call from step~(\ref{step:6}). The lemma below shows that this relabeling is not an issue.
\begin{lemma}[Upper bound]
For any instance $(G=(V,E,w),T,k)$ and a vertex $t\in R(T)$, the output value $\tilde\lambda(t)$ is at least $\lambda(s,t)$.
\end{lemma}
\begin{proof}
If $\tilde\lambda(t)$ is updated on Step~(\ref{step:1}) or Step~(\ref{step:2}), then we have $\tilde\lambda(t)\ge\lambda(s,t)$ because the updated value corresponds to a valid $(s,t)$-cut. Suppose now that $\tilde\lambda(t)$ is updated on Step~(\ref{step:4}), and let $T_i$ be the subtree containing $t$. By construction of $G_i$, we either contract a set containing $s$ (namely, $V\setminus S_i$) into a vertex labeled the new $s$, or we contract a set not containing $s$ (namely, $V\setminus S_i$ again) into a vertex (not labeled the new $s$). In both cases, any $(s,t)$-cut of graph $G_i$, with the contraction ``reversed'', is a valid $(s,t)$-cut in the original graph $G$. It follows that the $(s,t)$-mincut value $\lambda_{G_i}(s,t)$\footnote{$\lambda_{G_i}(s, t)$ is the value of an $(s,t)$-mincut in graph $G_i$.} in $G_i$ is at least the value $\lambda_G(s,t)$ in $G$. By induction, the output of the recursive call $(G_i,T_i',k)$ is at least $\lambda_{G_i}(s,t) \ge \lambda_G(s,t)$, as promised.

If $\tilde\lambda(s,t)$ is updated on Step~(\ref{step:5}), then since the graph $G$ remains unchanged, the value $\lambda(s,t)$ is also unchanged, and we have $\tilde\lambda(t)\ge\lambda(s,t)$ by induction. The most interesting case is when $\tilde\lambda(s,t)$ is updated on Step~(\ref{step:6}). Here, by the choice of $s'$, we have $\lambda(s,s') \ge \lambda(s,t)$. Next, observe that $\lambda(s',t) \ge \min\{\lambda(s,s'), \lambda(s,t)\}$ holds because the $(s',t)$-mincut is either an $(s,t)$-cut or an $(s',s)$-cut depending on whether $s$ is on the side of $s'$ or the side of $t$. Combining the two previous inequalities gives $\lambda(s',t) \ge \lambda(s,t)$, and by induction, the output of the recursive call $(G,T^{(\ref{step:6})},k-1)$ is at least $\lambda(s',t) \ge \lambda(s,t)$, as promised.
\end{proof}

The lemma above establishes the condition $\tilde\lambda(t)\ge\lambda(s,t)$ of \Cref{thm:ssmc}. It remains to show equality when $T$ is $k$-respecting an $(s,t)$-mincut, which we prove below.


\begin{lemma}[Equality]\label{lem:correctness}
Consider an instance $(G=(V,E,w),T,k)$ and a vertex $t\in R(T)$ such that there is an $(s,t)$-mincut in $G$ that $k$-respects $T$. Then, the value $\tilde\lambda(t)$ computed by the algorithm equals $\lambda(s,t)$ w.h.p. 
\end{lemma}
\begin{proof}
Consider an $(s,t)$-mincut $C$ in $G$ that $k$-respects $T$. First, if the centroid $c$ is the vertex $t$, then the mincut computation in Step~(\ref{step:2}) correctly recovers $\lambda(s,t)$. Otherwise, let $T_t$ be the subtree containing $t$. We have a few cases based on the locations of the edges in $T$ that cross the cut $C$, which we call the \emph{cut} edges.
Note that there is at least one cut edge along the $(s,t)$ path in $T$, and it is incident to (the vertices of) either $T_t$ or the subtree $T_s$ containing $s$. (If $c=s$ and there is no subtree $T_s$ containing $s$, then at least one cut edge must be incident on some vertex in $T_t$.) 

\begin{figure}[htbp]
  \begin{center}
    \ifarxiv
    \includegraphics[width=6.8in]{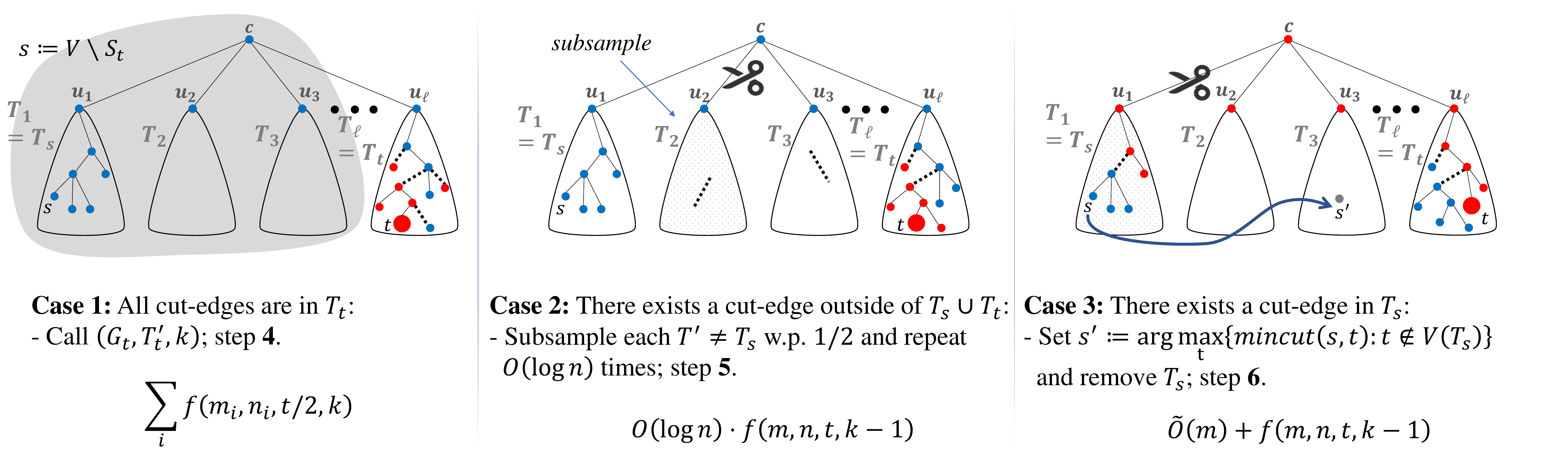}
    \else
    \includegraphics[width=6.8in]{Proof_single}
    \fi
    \end{center}
    \caption{An illustration of the different cases, which part of our algorithm deals with them, and the corresponding running time. Here, vertices in the side of $s$ are depicted by blue dots, vertices in the side of $t$ by red dots, and cut edges by dashed lines. The gray area refers to a contracted subset, and the scissors symbol means we remove the subtree. Observe that whenever the latter happens, we get rid of at least one cut edge.} 
    \label{Figs:Proof_single}
\end{figure}

The first case (Case 1 in Figure~\ref{Figs:Proof_single}) is that all the cut edges are incident to the vertices of a single subtree $T_j$, which must be either $T_t$ or $T_s$ (if the latter exists). Then, there is a side $A\in \{C, V\setminus C\}$ of the $(s,t)$-mincut $C$ whose vertices in $R(T)$ are all in $R(T_j)$; in other words, $A\cap R(T)=A\cap R(T_j)$. Note that $A$ is an $(A\cap R(T_j),R(T)\setminus(A\cap R(T_j)))$-mincut since if there were a smaller such cut, then that would also be a smaller $(s,t)$-cut, which contradicts that $A$ is an $(s,t)$-mincut. Also, by construction, $S_j$ is a $(R(T_j),R(T)\setminus R(T_j))$-mincut. We now apply \Cref{lem:uncross} on parameters $U=R(T)$, $A=A$, $X=A\cap R(T_j)$, $A'=S_j$, and $X'=R(T_j)$. The lemma implies that there is an $(A\cap R(T_j),R(T)\setminus(A\cap R(T_j)))$-mincut $\tilde A\subseteq S_j$, and this cut survives in the contracted graph $G_j$. Since $\tilde A$ is an $(s,t)$-cut of the same value as $A$, we conclude that $\tilde A$ is also an $(s,t)$-mincut. Finally, we argue that $\tilde A$ also $k$-respects the tree $T'_j$ in the recursive instance. By definition, since $A$ $k$-respects $T$, there exists a set $F_A$ of fake vertices such that $T$ contains at most $k$ edges cut by $A\cup F_A$. Since $A$ and $\tilde A$ agree on vertices in $R(T_j)$, tree $T$ also contains at most $k$ edges cut by $\tilde A\cup F_A$ (it is the exact same set of edges). Define $F_{\tilde A}=F_A\cap V(T_j)$, and from $\tilde A\cap R(T)\subseteq R(T_j)$, we observe that tree $T_j$ contains at most $k$ edges cut by $\tilde A\cap F_{\tilde A}$ (it is all edges from before, restricted to tree $T_j$). Furthermore, the new edge $(s,u_j)$ or $(c_j,u_j)$ added to $T'_j$ is cut by $\tilde A\cup F_{\tilde A}$ if and only if the edge $(c,u_j)$ of $T$ is cut by $\tilde A\cup F_A$. It follows that at most $k$ edges of $T'_j$ are cut by $\tilde A\cup F_{\tilde A}$. Thus, the lemma statement is satisfied on recursive call $(G_j,T'_j,k)$ of Step~(\ref{step:4}), and the algorithm recovers $\lambda(s,t)$ w.h.p.

In the rest of the proof, we assume that the edges of $T$ cut by $A\cup F_A$ are incident to (the vertices of) at least two subtrees. Suppose first (Case 2 in Figure~\ref{Figs:Proof_single}) that a cut edge is incident to some subtree $T_j$ that is not $T_t$ or $T_s$ (or only $T_t$, if $s=c$ and $T_s$ does not exist). In each independent trial of Step~(\ref{step:5}), we sample $T_t$ but not $T_j$ with constant probability. In this case, since $T_j$ is discarded in the construction of $T^{(\ref{step:5})}$, the $(s,t)$-mincut $C$ $(k-1)$-respects the resulting tree $T^{(\ref{step:5})}$. Over $O(\log n)$ independent trials, this happens w.h.p., and the algorithm correctly recovers $\lambda(s,t)$ w.h.p.

We are left with the case (Case 3 in \Cref{Figs:Proof_single}) that all edges of $T$ cut by $A\cup F_A$ are incident to subtrees $T_t$ and $T_s$. Note that $T_s$ must exist since if $s=c$ and Case 2 does not happen, we would be in Case 1. Furthermore, $T_s\ne T_t$, because otherwise, we would either be in Case 1 (if all cut edges are incident on $T_t = T_s$) or in Case 2 (if there is at least one cut edge incident on some $T_j \not= T_t = T_s$). 

Since $T_s\ne T_t$, we have $t\notin R(T_s)$, i.e., $t\in R(T)\setminus R(T_s)$. If $\lambda(s,t)=\lambda_{\max}$ (where $\lambda_{\max}$ is as defined in Step~(\ref{step:6})), then Step~(\ref{step:6}) sets $\lambda(s,t)=\lambda_{\max}$ correctly. Otherwise, we must have $\lambda(s,t)<\lambda_{\max}$. In this case, we claim that the vertex $s'$ (that has the property $\lambda(s,s')=\lambda_{\max}$ in Step~(\ref{step:6}) of the algorithm) satisfies $\lambda(s',t)=\lambda(s,t)$. To prove this claim, we first observe that $s'$ must appear on the $s$-side of the $(s,t)$-mincut $C$. Otherwise, if $s'$ is on the $t$-side, then $C$ is an $(s,s')$-cut of value $\lambda(s,t)<\lambda_{\max}$, contradicting the guarantee $\lambda(s,s')=\lambda_{\max}$. It follows that $\lambda(s',t)\le\lambda(s,t)$. Next, observe that $s$ must appear on the $s'$-side of the $(s',t)$-mincut $C'$. Otherwise, if $s$ is on the $t$-side, then $C'$ is an $(s,s')$-cut of value $\lambda(s',t)\le\lambda(s,t)<\lambda_{\max}$, contradicting the guarantee $\lambda(s,s')=\lambda_{\max}$. It follows that $\lambda(s,t)\le\lambda(s',t)$, which proves the claim $\lambda(s,t)=\lambda(s',t)$.

Consider again the $(s,t)$-mincut $C$. Since $s'$ is on the $s$-side of the $(s,t)$-mincut $C$, if we swap the locations of $s$ and $s'$ in $T$, then $C$ still $k$-respects the modified tree, and the edges of the tree that cross the cut are the same (except that $s$ and $s'$ swap places on the edges). In particular, the subtree $T_s$ with $s$ replaced by $s'$ has at least one cut edge. By removing this modified subtree $T_s$, we arrive at the tree $T^{(\ref{step:6})}$ in Step~(\ref{step:6}), and the $(s,t)$-mincut $C$ must $(k-1$)-respect $T^{(\ref{step:6})}$. So, the recursive call $(G,T^{(\ref{step:6})},k-1)$ recovers $\lambda(s',t)$ w.h.p., which equals $\lambda(s,t)$ by the claim above.

This concludes all cases, and hence the proof of \Cref{lem:correctness}.
\end{proof}

\subsection{Running Time}
\begin{lemma}[Running time]\label{lem:runtime}
For any fixed integer $k\ge1$, the algorithm calls max-flow on instances of at most $n$ vertices and $m$ edges each, and a total of $\tilde{O}(n)$ vertices and $\tilde{O}(m)$ edges. Moreover, these max-flow calls dominate the running time.
\end{lemma}
We first bound the total number of vertices across all recursive instances, then use the same technique to also bound the total number of edges.

We use the following notation for any recursive call: $r=|R(T)|$ and $n$ represents the number of vertices in $G$ including contracted vertices, i.e., vertices resulting from the contraction on Step~(\ref{step:4}) of any previous instance. (Since the original instance has no contracted vertex, the initial value of $n$ is just the number of vertices in the input graph.) The function $f(n,r,k)$ represents an upper bound on the total number of vertices among all max-flow calls that occur, starting from a single input instance with parameters $n,r,k$ (and including max-flows called inside recursive instances).

Fix an instance with parameters $n,r,k$. For each $i$, let $n_i$ represent the number of vertices in $G_i$, and let $r_i=|R(T_i)|$. Now observe that 
 \begin{enumerate}
 \item $\sum_{i=1}^\ell (n_i-1)=n-1$ since the sets $S_i$ are disjoint by the guarantee of \Cref{lem:iso cut}, and
 \item $r_i\le r/2$ for each $i\in[\ell]$ by the fact that $c$ is a centroid.
 \end{enumerate}
We now consider the individual steps of the recursive \sscv algorithm.
 \begin{enumerate}
 \item The algorithm calls a single max-flow in step~(\ref{step:2}), and then in step~(\ref{step:3}),  it calls \Cref{lem:iso cut}, which in turn calls max-flows on a total of $O(n\log \ell)$ vertices. In total, this is $O(n \log \ell)$ vertices among the max-flow calls. 
 \item In step~(\ref{step:4}), the algorithm makes recursive calls on trees $T'_i$ containing $r_i+1$ real vertices each, and graphs $G_i$ containing $n_i$ vertices each, so the total number of vertices in the max-flow calls in the recursion is at most $\sum_{i\in[\ell]}f(n_i,r_i+1,k)$.
 \item In step~(\ref{step:5}), the algorithm makes $O(\log n)$ independent calls to an instance where $k$ decreases by $1$. So, this step contributes at most $O(\log n)\cdot f(n,r,k-1)$.
 \item In step~(\ref{step:6}), the algorithm calls \Cref{lem:cut-thr-step}, which in turn calls max-flows on a total of $\tilde O(n)$ vertices, followed by a recursive call on an instance where $k$ decreases by $1$. In total, this step contributes at most $\tilde{O}(n)+f(n,r,k-1)$.
 \end{enumerate}
We may assume that $f(n,r,k)$ is monotone non-decreasing in all three parameters, which gives us the recursive formula
\BAL
f(n,r,k) \le \underbrace{O(n\log \ell)}_{\text{steps~(\ref{step:2}),(\ref{step:3})}}&+ \underbrace{\sum_{i\in[\ell]}f(n_i,r_i+1,k)}_{\text{step (\ref{step:4})}} + \underbrace{O(\log n)\cdot f(n,r,k-1)}_{\text{step (\ref{step:5}), only for }k>1} + \underbrace{\tilde{O}(n)+f(n,r,k-1)}_{\text{step (\ref{step:6}), only for }k>1} .
\EAL
We now claim that $f(n,r,k)$ solves to $\tilde{O}(n)$ for any constant $k$, where the number of polylog terms depends on $k$.
For $k=1$, the recursive formula $f(n,t,1)$ solves to $O(n\log^2t)$. This is because $r_i+1 \le r/2+1\leq \frac23 r$ for all $i\in [\ell]$ limits the recursive depth to $O(\log t)$.\footnote{Here, we have used the assumption that $r$ is larger than some constant, e.g.\ 10.} And, since $\sum_{i=1}^\ell (n_i-1)=n-1$, the sum of $f(\cdot)$ in any recursive level is $O(n\log t)$. For larger $k$, note that if we assume that $f(n,r,k-1)\le\tilde{O}(n)$, then we also obtain $f(n,r,k)\le\tilde{O}(n)$, where the $\tilde{O}(\cdot)$ hides more logarithmic factors. The claim then follows by induction on $k$. (Note that the polylogarithmic dependency on $k$ is $f(n,r,k)=n(\log n)^{O(k)}$.)

We now bound the total number of edges. We use the following notation in any recursive call: as earlier, $r=|R(T)|$ and $n$ represents the number of vertices in $G$ including contracted vertices. In addition, $m'$ represents $n$ plus the number of edges in $G$ \emph{not incident to a contracted vertex}. (Since the original instance has no contracted vertex, the initial value of $m'$ is just the number of vertices plus the number of edges in the input graph.) The function $g(m',n,r,k)$ represents an upper bound on $f(n,r,k)$ plus the total number of edges not incident to contracted vertices among all max-flow calls that occur, starting from a single input instance with parameters $m',n,r,k$ (and including max-flows called inside recursive instances). This then implies a bound on the total number of edges over all max-flow calls, including those incident to contracted vertices, by the following argument. Each recursive instance has at most $O(\log|R(T)|)$ contracted vertices, since each contraction in Step~(\ref{step:4}) decreases $|R(T)|$ by a constant factor. So the total number of edges incident to contracted vertices is at most the total number of vertices times $O(\log|R(T)|)$, which is at most $f(n,r,k)\cdot O(\log n) \le g(m',n,r,k)\cdot O(\log n)$. So from now on, we only focus on edges not incident to contracted vertices.

Fix an instance with parameters $m',n,r,k$. For each $i$, let $n_i$ represent the number of vertices in $G_i$, let $m'_i$ represent the number of edges in $G_i$ not incident to a contracted vertex, and let $r_i=|R(T_i)|$. Once again, observe that $\sum_{i=1}^\ell (n_i-1)=n-1$ and $r_i\le r/2$ for each $i\in[\ell]$. This time, we also have $\sum_{i=1}^\ell m'_i\le m'$ by the following explanation: \Cref{lem:iso cut} guarantees that the vertex sets $S_i$ are disjoint, and the edges of each $G_i$ not incident to a contracted vertex have both endpoints in $S_i$, and are therefore disjoint over all $i$. We may assume that $g(m,n,r,k)$ is monotone non-decreasing in all four parameters, which gives us the recursive formula
\BAL
g(m',n,r,k) \le \underbrace{O((m+n)\log \ell)}_{\text{steps~(\ref{step:2}),(\ref{step:3})}} + \underbrace{\sum_{i\in[\ell]}g(m_i,n_i,r_i,k)}_{\text{step (\ref{step:4})}} &+ \underbrace{O(\log n)\cdot g(m,n,r,k-1)}_{\text{step (\ref{step:5}), only for }k>1} \\&+ \underbrace{\tilde{O}(m)+g(m,n,r,k-1)}_{\text{step (\ref{step:6}), only for }k>1}  .
\EAL


Similar to the solution for $f(n,r,k)$, we now have that $g(m,n,r,k)$ solves to $\tilde{O}(m+n)$ for any constant $k$ by the same inductive argument. (Once again, the polylogarithmic dependency on $k$ is $f(n,r,k)=(m+n)(\log n)^{O(k)}$.)


Since the graph never increases in size throughout the recursion, each max-flow call is on a graph with at most as many vertices and edges as the original input graph. Finally, we claim that the max-flow calls dominate the running time of the algorithm. In particular, finding the centroid of $T$ on step~(\ref{step:2}) can be done in time in the size of the tree (see the footnote at step~(\ref{step:2})), which is dominated by the single max-flow call on the same step. This finishes the proof of \Cref{lem:runtime}.

\section{Constructing Guide Trees}
\label{sec:packing}

In this section, we show how to obtain guide trees that prove \Cref{lem:guide}. 
Our algorithm is based on the notion of a \emph{Steiner subgraph packing},
as described next.

\begin{definition}
Let $G = (V, E,w)$ be an undirected edge-weighted graph
with a set of terminals $U\subseteq V$.
A subgraph $H$ of $G$ is said to be
a \emph{$U$-Steiner subgraph} (or simply a Steiner subgraph
if the terminal set $U$ is unambiguous from the context) 
if all the terminals are connected in $H$. In this case,
we also call $H$ a {\em terminal-spanning} subgraph of $G$. 
\end{definition}

\begin{definition}
A \emph{$U$-Steiner-subgraph packing} $\pset$ is a collection
of $U$-Steiner subgraphs $H_{1},\ldots,H_{k}$,
where each subgraph $H_{i}$ is assigned a value $\val(H_{i})>0$.
If all $\val(H_{i})$ are integral,
we say that $\pset$ is an \emph{integral packing}.
Throughout, a packing is assumed to be \emph{fractional}
(which means that it does not have to be integral), unless specified otherwise. 
The \emph{value} of
the packing $\pset$ is the total value of all its Steiner subgraphs,
denoted $\val(\pset)=\sum_{H\in\pset}\val(H)$.
We say that $\pset$ is \emph{feasible} if
\[
  \forall e\in E, \quad
  \sum_{H\in\pset : e\in H} \val(H)\le w(e).
\]
\end{definition}

To understand this definition, think of $w(e)$ as the ``capacity'' of $e$; 
then, this condition means that the total value of all Steiner subgraphs
$H\in{\pset}$ that ``use'' edge $e$ does not exceed its capacity $w(e)$. 
A \emph{Steiner-tree packing} is a packing $\pset$
where each subgraph $H\in\pset$ is a tree.

Denote by $\optval(U)$ the maximum value of a feasible $U$-Steiner-subgraph
packing in $G$.
The next two lemmas show a close connection between
Steiner-subgraph packing $\optval(U)$ and $U$-Steiner mincut $\lambda(U)$,
and that the former problem admits a $(2+\eps)$-approximation algorithm. 

\begin{lem}
\label{lem:packing vs mincut}
For every graph $G$ with terminal set $U$,
we have $\lambda(U)/2\le\optval(U)\le\lambda(U)$.
\end{lem}

\begin{lem}
\label{lem:packing MWU}
There is a deterministic algorithm that, given $\eps\in (0,1/2)$, a graph $G=(V,E,w)$, and a terminal set $U\subseteq V$,
returns a $U$-Steiner-subgraph packing $\pset$
of value $\val(\pset)\ge\optval(U)/(2+\eps)$ 
in $\Otil(m^{2}/\epsilon^2)$ time,
and in the case of unweighted $G$ in $m^{1+o(1)}/\epsilon^2$ time.
\end{lem}

We prove Lemmas~\ref{lem:packing vs mincut} and \ref{lem:packing MWU} in \Cref{sec:packing_mincut,sec:mincost}, respectively.
Assuming these lemmas, we immediately obtain the following.

\begin{corollary}\label{cor:packing close to mincut}
There is a deterministic algorithm that, given $\eps\in (0,1/2)$
and a graph $G=(V,E,w)$ with $m$ edges and terminal set $U\subseteq V$,
returns a $U$-Steiner subgraph packing $\pset$
of value $\val(\pset)\ge\lambda(U)/(4+\eps)$
in $\Otil(m^{2}/\epsilon^2)$ time, 
and in the case of unweighted $G$ in $m^{1+o(1)}/\epsilon^2$ time.
\end{corollary}

\paragraph{Algorithm for constructing guide trees.}
Given \Cref{cor:packing close to mincut}, we can now prove \Cref{lem:guide}. 

\begin{proof}[Proof of \Cref{lem:guide}]
Fix $\eps_0=1/60$ (or another sufficiently small $\eps_0>0$). 
The construction of guide trees is described in Algorithm~\ref{algorithm:guide_trees}.
To analyze this algorithm,
let $\pset$ be the packing computed in line~\ref{step:cor}.
Consider $t\in U\setminus\{s\}$,
and let $(S_t,V\setminus S_t)$ be a minimum $s,t$-cut in $G$.
Denote by $w(S_t,V\setminus S_t)$ the total edge-weight of this cut in $G$,
and by $w'(S_t,V\setminus S_t)$ the total edge-weight of the cut in $G'$ between these same vertices.

\IncMargin{1em}
\begin{algorithm}
  \caption{An algorithm for constructing guide trees}\label{algorithm:guide_trees}
  \SetKwFor{RepTimes}{repeat}{times}{end}
  \SetKwInOut{Input}{input}
  \SetKwInOut{Output}{output}
  \SetKwFunction{ExtractMax}{ExtractMax}
  \SetKw{KwCompute}{compute}
  \SetKw{KwSample}{sample}
  \SetKw{KwLet}{let}
  \Indm
  \Input{Undirected graph $G=(V,E,w)$ (weighted or unweighted) and terminal set $U\subseteq V$}
  \Output{A collection of guide trees for $U$}
  \Indp
  \BlankLine
  \If{$G$ is weighted then}{
      \KwCompute for it a $(1\pm\eps_0)$-cut-sparsifier $G'$ using~\cite{BeK15} and $\eps_0=1/60$,
    thus $G'$ has $m=\tO(n/\eps_0^2)$ edges \label{step:spars}
    }
  \Else{
\KwLet $G' \gets G$
  }
  \KwCompute a packing $\pset$ for $G'$
  by applying Corollary~\ref{cor:packing close to mincut} \label{step:cor}\\
  \KwSample $300\ln n$ subgraphs from $\pset$,
  each drawn independently from the distribution $\{\val(H)/\val(\pset)\}_{H\in\pset}$\label{step:sample}\\
  \KwCompute any Steiner tree of each sampled subgraph, and report these trees\label{step:Steiner} \\
\end{algorithm}
\DecMargin{1em}

Consider first an unweighted input $G$.
Then, the computation in line \ref{step:cor} is applied to $G'=G$.
By combining \Cref{cor:packing close to mincut}
and the promise in the single source terminal mincuts problem (\Cref{problem:ssmc}) that $\lambda_G(s,t) \le 1.1\lambda(U)$,
\footnote{For a graph $G'$, $\lambda_{G'}(s, t)$ denotes the value of an $(s,t)$-mincut in $G'$, and recall that $\lambda(U)$ is the value of a $U$-Steiner mincut in $G$.}
we get that
\begin{align}\label{eq:valp1}
  \val(\pset)
  \geq \frac{\lambda(U)}{4+\eps_0}
  \geq \frac{\lambda_{G}(s,t)}{1.1 (4+\eps_0)}
  = \frac{w'(S_t,V\setminus S_t)}{1.1 (4+\eps_0)}\ .
\end{align}

If the input graph $G$ is weighted,
then the bound in~\eqref{eq:valp1} applies to the cut-sparsifier $G'$ of $G$,
and we get that 
\begin{align} \label{eq:valp2}
  \val(\pset)
  \geq \frac{\lambda_{G'}(s,t)}{1.1 (4+\eps_0)}
  \geq \frac{(1-\eps_0)\cdot w(S_t,V\setminus S_t)}{1.1 (4+\eps_0)} 
  \geq \frac{(1-\eps_0)\cdot w'(S_t,V\setminus S_t)}{1.1 (4+\eps_0)(1+\eps_0)} 
  \geq \frac{w'(S_t,V\setminus S_t)}{1.1 (4+30\eps_0)}\ .
\end{align}
We remark that now the packing $\pset$ contains subgraphs of the sparsifier $G'$ and not of $G$,
but it will not pose any issue.

In both cases we have the weaker inequality
\begin{align}\label{eq:main}
  \val(\pset)\geq \frac{w'(S_t,V\setminus S_t)}{1.1 (4+30\eps_0)}\ .
\end{align}
Let $E'_t$ be the set of edges in the cut $(S_t,V\setminus S_t)$ in $G'$
(depending on the case, $G'$ is either $G$ or the sparsifier).
Let $\pset_{\leq 4}\subseteq \pset$ be the subset of all Steiner subgraphs
$H\in\pset$ whose intersection with $E'_t$ is at most $4$ edges,
and let $F_t$ be the event that no subgraph from $\pset_{\leq 4}$ is sampled
in line~\ref{step:sample}.
Then
\begin{align}
  \Pr[F_t]
  &=\left(1-\val(\pset_{\leq 4})/\val(\pset)\right)^{300\ln n} \nonumber\\
  &\leq n^{-300\cdot \val(\pset_{\leq 4})/\val(\pset)} .  \label{eq:weighted}
\end{align}

Similarly to Karger's paper~\cite{Karger00},
define $x_H$ to be one less than the number of
edges of $H$ that crosses the cut $E'_t$,
and observe that $x_H$ is always a non-negative integer
(because $U$ is connected in $H$). 
Since $\pset$ is a packing,
every edge of $E'_t$ appears in at most one subgraph of $\pset$,
and consequently,
\begin{align*}
  \sum_{H\in \pset} \val(H)(x_H+1) &\leq \sum_{e\in E'_t} w(e) = w'(S_t,V\setminus S_t)
  \\
  \Longrightarrow 
  \sum_{H\in \pset} \val(H)x_H &\leq w'(S_t,V\setminus S_t)-\sum_{H\in \pset} \val(H)= w'(S_t,V\setminus S_t)-\val(\pset).
\end{align*}
Observe that for a random $\bar{H}\in \pset$ drawn as in line~\ref{step:sample},
\[
  \EX_{\bar{H}}[x_{\bar{H}}]
  = \sum_{H\in \pset} x_H\cdot \frac{\val(H)}{\val(\pset)}
  \leq \frac{w'(S_t,V\setminus S_t)}{\val(\pset)} - 1
  \leq 1.1(4+30\eps_0)-1,
\]
where the last inequality is by~\eqref{eq:main}.
By Markov's inequality, 
$$
  \Pr_{\bar{H}} [x_{\bar{H}}\ge 4]
  \leq \frac{1.1(4+30\eps_0)-1}{4}
  \leq 0.99.
$$
Observe that $\val(\pset_{\leq 4})/\val(\pset) = \Pr_{\bar{H}} [x_{\bar{H}}< 4] \geq 0.01$,
and so by plugging into $(\ref{eq:weighted})$ we get that $\Pr[F_t]\leq 1/n^3$.
Finally, by a union bound we have that with high probability,
for every $t\in U\setminus\{s\}$,
at least one of the subgraphs that are sampled in line~\ref{step:sample}
of the algorithm $4$-respects the cut $E'_t$,%
\footnote{Strictly speaking, \Cref{defn:k-respect} defines $4$-respecting
  only relative to a tree $T$,
  but the same wording extends immediately to any graph $T$
  (not necessarily a tree).
}
and thus at least one of the trees reported by the algorithm $4$-respects $E'_t$.
Furthermore, since the cut $E'_t$ in $G'$
has the exact same bipartition of $V$ as the $(s,t)$-mincut in $G$,
the reported tree mentioned above $4$-respects also the $(s,t)$-mincut in $G$
(recall that Definition~\ref{defn:k-respect} refers to a cut as a bipartition of $V$). 

Finally, computing a Steiner tree of a Steiner subgraph in Line~\ref{step:Steiner} only takes linear time, and so the running time is dominated by line~\ref{step:cor} and thus it is bounded by $\tO(n^2/\epsilon^2)$ for weighted graphs and by $m^{1+o(1)}/\epsilon^2$ for unweighted graphs,
and by fixing a small $\eps>0$, we can write these as $\tO(n^2)$ and $m^{1+o(1)}$, respectively. 
This concludes the proof of~\Cref{lem:guide}.
\end{proof}


\subsection{Steiner-Subgraph Packing vs Steiner Mincut}
\label{sec:packing_mincut}

In this section, we prove \Cref{lem:packing vs mincut}, 
i.e., that $\lambda(U)/2\le\optval(U)\le\lambda(U)$.

We start with the second inequality $\optval(U)\le\lambda(U)$,
which is easier.
Let $\pset$ be a $U$-Steiner-subgraph packing,
and let $S\subset V$ be a Steiner min-cut of $U$.
Then for every Steiner subgraph $H\in\pset$,
by definition $|E_H(S,V\setminus S)| \ge 1$.\footnote{For any graph $G'$, the set of edges in the cut $(S, V\setminus S)$ is denoted $E_{G'}(S,V\setminus S)$.}
By the feasibility of $\pset$,
for every $e\in E$ we have $\sum_{H\in {\pset}: e\in H}  \val(H) \le w(e)$. 
Putting these together, we conclude that
\begin{align*}
\val({\cal P})=\sum_{H\in{\cal P}}\val(H) & \le\sum_{H\in{\cal P}}\val(H)\cdot|E_{H}(S,V\setminus S)|\\
 & =\sum_{e\in E_{G}(S,V\setminus S)}\sum_{H\in{\cal P}: e\in H}\val(H)\\
 & \le \sum_{e\in E_{G}(S,V\setminus S)}w(e) 
 = \lambda(U) .
\end{align*}

To establish the first inequality $\lambda(U)/2\le\optval(U)$,
we need to show that one can always
pack into $G$ Steiner subgraphs of total value (at least) $\lambda(U)/2$.
This packing bound actually follows from more general theorems of Bang-Jensen~{\em et al.}~\cite{BFJ95} and Bhalgat~{\em et al.}~\cite{BHKP07,BCHKP08},  but we give a simple self-contained proof here. 

First note that we can assume, only for sake of analysis, 
that the graph is an unweighted multi-graph
(by replacing each edge by parallel edges of unit weight)
and that every vertex has an even degree (by appropriate scaling). (In what follows, we will obtain an integral packing but undoing the scaling step possibly converts it into a fractional packing.)

We use the following classic result due to Mader~\cite{Mader78}.

\begin{theorem}
\label{thm:localsplitting}
Let $G = (V\cup\{s\}, E)$ be an undirected unweighted multi-graph where every vertex has an even degree.%
\footnote{In the original theorem of Mader, the only restriction is that the degree of $s$ cannot be equal to 3. But, for our purposes, it suffices to assume that every vertex has even degree.}
Then, there exists a pair of edges $(u, s)$ and $(v, s)$ incident on $s$ that can be {\em split off}, i.e., replaced by their {\em shortcut} edge $(u, v)$, while preserving the $x-y$ edge connectivity for all pairs of vertices $x, y\in V$.
\end{theorem}

By applying Theorem~\ref{thm:localsplitting} repeatedly on edges incident on $s$, we can isolate the vertex $s$. Iterating further, we can isolate any subset of vertices while preserving the pairwise edge connectivities of all the remaining vertices.
We thus isolate all the non-terminal vertices,
and then delete these isolated vertices
to obtain a graph $G(U)$ only on the terminal vertices $U$.
Observe that edges in $G(U)$ represent edge-disjoint paths in $G$.
We also claim that the global edge connectivity of $G(U)$ is at least $\lambda(U)$.
Indeed, the global edge connectivity of $G(U)$
equals the minimum of all its pairwise edge connectivities,
which are all preserved from $G$,
and it is clear that in $G$ all pairwise edge connectivities are at least $\lambda(U)$. 

Next, we shall use the following classic theorem of Nash-Williams~\cite{NW61} and Tutte~\cite{Tut61} to pack edge-disjoint spanning trees in $G(U)$.

\begin{theorem}
\label{thm:spanning-pack}
Let $G = (V, E)$ be an undirected unweighted multi-graph with global edge connectivity $\lambda_0$.
Then, $G$ contains at least $\lambda_0/2$ edge-disjoint spanning trees.
\end{theorem}

Now we apply this theorem to our graph $G(U)$, 
to pack in it (at least) $\lambda(U)/2$ edge-disjoint spanning trees,
and then replace each edge in these spanning trees of $G(U)$ by its corresponding path in $G$.
These paths are edge-disjoint as well,
hence the corresponding subgraphs are edge-disjoint in $G$.
This yields a set of (at least) $\lambda(U)/2$
edge-disjoint $U$-spanning subgraphs in $G$,
which completes the proof of \Cref{lem:packing vs mincut}.


\subsection{$(2+\protect\eps)$-approximate Steiner-Subgraph Packing}
\label{sec:mincost}

In this section, we provide a $(2+\eps)$-approximation algorithm for
fractionally packing Steiner subgraphs, proving \Cref{lem:packing MWU}.
The technique is a standard application of the width-independent multiplicative
weight update (MWU) framework \cite{GargK07,Fleischer00,AroraHK12}. We provide
the proofs for completeness, following closely the presentation from \cite{ChuzhoyS20}.

We start by describing, in \Cref{sec:mincost_MWU},
an algorithm based on the multiplicative weight update framework,
and we bound its number of iterations, also called ``augmentations''. 
Then, in \Cref{sec:mincost_dynamic},
we show how to implement these augmentations efficiently by using
either a static $(2+\eps)$-approximation algorithm by Mehlhorn \cite{Mehlhorn88},
or in unweighted graphs, by using a decremental
$(2+\eps)$-approximation algorithm from \Cref{thm:dec_MinSteiner}.

\subsubsection{Algorithm Based on Multiplicative Weight Update}
\label{sec:mincost_MWU}

Let the input be an edge-weighted graph $G=(V,E,w)$ with
$n$ vertices, $m$ edges, and terminal set $U\subseteq V$.
Denote by $\HH$ the set of all $U$-Steiner subgraphs. 
The algorithm maintains $\ell\in\mathbb{R}_{> 0}^{E}$,
which we refer to as \emph{edge lengths}.%
\footnote{Multiplicative weight update algorithms usually maintain ``weights'' on edges of the input graph. We use here instead the terminology of \emph{edge lengths}, because $G$ already has edge weights (that can be viewed as capacities). 
}
The \emph{total length} of a subgraph $H$ is defined as
$\ell(H)=\sum_{e\in E(H)}\ell(e)$.
A $U$-Steiner subgraph $H$ is said to have
\emph{$\gamma$-approximate minimum $\ell$-length}
if its length satisfies $\ell(H) \leq \gamma\cdot\min_{H'\in\HH}\ell(H')$.
Our algorithm below for packing Steiner subgraphs assumes access to a procedure
that computes a $\gamma$-approximate minimum $\ell$-length Steiner subgraph. 

\IncMargin{1em}
\begin{algorithm}
  \caption{A $(\gamma+O(\protect\eps))$-approximation algorithm for Steiner subgraph
packing}\label{algorithm:MWU}
  \SetKwFor{RepTimes}{repeat}{times}{end}
  \SetKwInOut{Input}{input}
  \SetKwInOut{Output}{output}
  \SetKwFor{While}{while}{do}{end}
  \SetKwFunction{ExtractMax}{ExtractMax}
  \SetKw{KwLet}{let}
  \Indm  
  \Input{Undirected edge-weighted graph $G=(V,E,w)$, terminal set $U\subseteq V$, and
accuracy parameter $0<\epsilon<1$}
  \Output{Feasible $U$-Steiner subgraph packing ${\cal P}$}
  \Indp
  \BlankLine
   \KwLet $\pset \gets \emptyset$ and $\delta\gets (2m)^{-1/\epsilon}$ \\
   \KwLet $\ell(e) \gets \delta/w(e)$ for all $e\in E$ \\
   \While{$\sum_{e\in E}w(e)\ell(e)<1$}{
     \KwLet $H\gets$ a $\gamma$-approximate minimum $\ell$-length $U$-Steiner
subgraph\label{enu:oracle} \\
     	\KwLet $v\gets\min\{w(e):\ e\in E(H)\}$ \\
     	add $H$ with $\val(H)\gets v$ into the packing $\pset$ \label{enu:pack} \\        
    \KwLet $\ell(e)\gets\ell(e)(1+\frac{\epsilon v}{w(e)})$ for all $e\in E(H)$ \label{enu:addlength}
    }
  \textbf{return} a scaled-down packing $\pset$, where $\val(H)\gets\val(H)/\log_{1+\epsilon}(\frac{1+\epsilon}{\delta})$ for all $H\in\P$\label{enu:scaling} 
\end{algorithm}
\DecMargin{1em}

\begin{lem}
The scaled-down packing ${\cal P}$ computed by Algorithm~\ref{algorithm:MWU}
is feasible.\label{lem:feasible} 
\end{lem}

\begin{proof}
Given a packing $\P$, if we define the ``flow'' (or load) on an edge $e$
as $f(e)=\sum_{H\in{\cal P}, e\in H}\val(H)$,
then we need to show that $f(e)\le w(e)$ for all $e\in E$.
(This terminology highlights the analogy with packing flow paths
and thinking of $w(e)$ as the capacity of $e$.)
Consider the computation of $\P$
\emph{before} scaling it down in line~\ref{enu:scaling}.
Whenever an iteration adds some $H$ into the packing $\P$ (in line~\ref{enu:pack}),
it effectively increases the flow $f(e)$ additively
by $v=a\cdot w(e)$ for some $0\le a\le1$,
and the corresponding length $\ell(e)$ is increased multiplicatively
by $1+a\epsilon\ge(1+\epsilon)^{a}$ (in line~\ref{enu:addlength}).
Observe that initially $\ell(e)=\delta/w(e)$ and
at the end $\ell(e)<(1+\epsilon)/w(e)$ 
(because prior to the last iteration $w(e)\ell(e)<1$),
thus over the course of the execution, $\ell(e)$ grows multiplicatively by
at most $(1+\eps)/\delta=(1+\eps)^{\log_{1+\eps}((1+\eps)/\delta)}$,
implying that at the end of the execution, 
$f(e)\le w(e)\cdot\log_{1+\epsilon}(\frac{1+\epsilon}{\delta})$.
Hence, scaling down $\P$ by factor $\log_{1+\epsilon}(\frac{1+\epsilon}{\delta})$
in line~\ref{enu:scaling} makes this packing feasible. 
\end{proof}

We call each iteration in the while loop of Algorithm~\ref{algorithm:MWU} an \emph{augmentation}
and say that an edge $e\in E$ \emph{participates} in the augmentation
if $e$ is contained in the Steiner subgraph $H$ of that iteration (in line~\ref{enu:oracle}). 
The next lemma is used to bound the total running time. 

\begin{lem}
\label{lem:aug bound} \ 
\begin{enumerate}
\renewcommand{\theenumi}{(\alph{enumi})}
\item \label{it:AugBound1} 
If $G$ is unweighted, then
each edge $e\in E$ participates in at most $\Otil(1/\eps^{2})$
augmentations.
\item \label{it:AugBound2}
There are at most $\tilde{O}(m/\epsilon^{2})$ augmentations. 
\end{enumerate}
\end{lem}

\begin{proof}
\ref{it:AugBound1} 
Fix $e\in E$. In every augmentation, $v\ge 1$ because $G$ is unweighted.
So whenever $e$ participates in an augmentation,
$\ell(e)$ is increased by factor $1+\eps$.
Since $\ell(e)=\delta/w(e)$ initially
and $\ell(e)<(1+\epsilon)/w(e)$ at the end,
$e$ can participate in at most
$\log_{1+\epsilon}(\frac{1+\epsilon}{\delta})
  = O(\frac{1}{\epsilon} \log\frac{1}{\delta}) 
  =\tilde{O}(\frac{1}{\epsilon^{2}})$
augmentations.

\ref{it:AugBound2} 
By the choice of $v$ in each augmentation,
at least one edge $e$ has its length $\ell(e)$ increased by factor $1+\epsilon$.
For every edge $e\in E$,
initially $\ell(e)=\delta/w(e)$ and at the end $\ell(e)<(1+\epsilon)/w(e)$,
and thus the total number of augmentations is bounded by
$m\log_{1+\epsilon}(\frac{1+\epsilon}{\delta})=\tilde{O}(m/\epsilon^{2})$. 
\end{proof}

\begin{lem}
\label{lem:approx bound}The scaled-down packing ${\cal P}$ is a
$(\gamma+O(\epsilon))$-approximate $U$-Steiner subgraph packing. 
\end{lem}

\begin{proof}
We first write an LP for the $U$-Steiner subgraph packing problem and its dual LP, using our notation $\ell(H)=\sum_{e\in E(H)}\ell(e)$. 
\begin{center}
\begin{tabular}{|c|c|}
\hline 
$\begin{array}{crccc}
(\boldsymbol{P}_{LP})\\
\max & \sum_{H\in\HH}\val(H)\\
\text{s.t.} & \sum_{H\in\HH : e\in H}\val(H) & \le & w(e) & \forall e\in E\\
 & \val(H) & \ge & 0 & \forall  H\in\HH
\end{array}$  & $\begin{array}{crccc}
(\boldsymbol{D}_{LP})\\
\min & \sum_{e\in E}w(e)\ell(e)\\
\text{s.t.} & \ell(H) & \ge & 1 & \forall H\in\HH\\
 & \ell(e) & \ge & 0 &  \forall e\in E 
\end{array}$\tabularnewline
\hline 
\end{tabular}
\par\end{center}

Denote $D(\ell)=\sum_{e\in E}w(e)\ell(e)$, and let $\alpha(\ell)=\min_{H\in\HH}\ell(H)$
be the length of the minimum $\ell$-length $U$-Steiner subgraph.
Let $\ell_{i}$ be the edge-length function $\ell$ after $i$ executions
of the while loop
(and by convention $i=0$ refers to just before entering the loop).
For brevity, denote $D(i)=D(\ell_{i})$ and $\alpha(i)=\alpha(\ell_{i})$.
We also denote by $H_{i}$ the Steiner subgraph found
in the $i$-th iteration and by $v_{i}$ the value of $H_{i}\in\P$
(before scaling down). Observe that
\begin{align*}
D(i) & =\sum_{e\in E}w(e)\ell_{i-1}(e)+\sum_{e\in E(H_{i})}w(e) \Big(\frac{\epsilon v_{i}}{w(e)}\cdot\ell_{i-1}(e)\Big)\\
 & =D(i-1)+\epsilon v_{i}\cdot\ell_{i-1}(H_{i}).
\end{align*}
Since $H_{i}$ is a $\gamma$-approximate minimum $\ell_{i-1}$-length $U$-Steiner subgraph, 
\begin{align*}
D(i) & \le D(i-1)+\epsilon v_{i}\gamma\cdot\alpha(i-1).
\end{align*}
Observe that the optimal value of the dual LP $\boldsymbol{D}_{LP}$
can be written as $\opt=\min_{\ell}D(\ell)/\alpha(\ell)$,
and thus
\begin{align*}
D(i) & \le D(i-1)+\epsilon\gamma v_{i}D(i-1)/\opt\\
 & \le D(i-1)\cdot e^{\eps\gamma v_{i}/\opt}.
\end{align*}
Let $t$ be the index of the last iteration, then $D(t)\ge1$.
Since $D(0)\le2\delta m$, we get
\[
1\le D(t)\le2\delta m\cdot e^{\epsilon\gamma\sum_{i=1}^{t}v_{i}/\opt}.
\]
Taking a logarithm from both sides, and denoting $\val=\sum_{i=1}^{t}v_{i}$,
we get
\begin{equation}
  \frac{\val\cdot\eps\gamma}{\opt}\geq\ln(1/(2\delta m)).
  \label{eq: calc}
\end{equation}

Notice that $\val$ is exactly the value of $\P$ before scaling down in line~\ref{enu:scaling}.
By \Cref{lem:feasible},
the scaled-down packing is a feasible solution for the primal $\boldsymbol{P}_{LP}$,
and thus by weak duality
$\val/\log_{1+\epsilon}(\frac{1+\epsilon}{\delta})\le\opt$.
Thus, to prove that it achieves $(\gamma+O(\epsilon))$-approximation, 
it suffices to show that $\val/\log_{1+\epsilon}(\frac{1+\epsilon}{\delta}) \ge \frac{1-O(\epsilon)}{\gamma}\cdot \opt$.
And indeed, using~\eqref{eq: calc}
and the fact that $\delta=(2m)^{-1/\epsilon}$, i.e., $2m=(1/\delta)^{\eps}$,
we have 
\begin{align*}
\frac{\val/\log_{1+\epsilon}(\frac{1+\epsilon}{\delta})}{\opt} & \ge\frac{\ln(1/(2\delta m))}{\epsilon\gamma}\cdot\frac{1}{\log_{1+\epsilon}(\frac{1+\epsilon}{\delta})}\\
 & =\frac{\ln(1/\delta)-\ln(2m)}{\epsilon\gamma}\cdot\frac{\ln(1+\epsilon)}{\ln(\frac{1+\epsilon}{\delta})}\\
 & \ge\frac{(1-\epsilon)\ln(1/\delta)}{\epsilon\gamma}\cdot\frac{\ln(1+\epsilon)}{\ln(\frac{1+\epsilon}{\delta})}\\
 & \geq\frac{1-O(\epsilon)}{\gamma},
\end{align*}
where the last inequality holds because
$\ln(1+\epsilon)\ge\epsilon-\epsilon^{2}/2$
and $\ln(\frac{1+\epsilon}{\delta})\le(1+\eps)\ln(1/\delta)$.
\end{proof}


\subsubsection{Efficient Implementation\label{sec:mincost_dynamic}}

In this section, we complete the proofs of \Cref{lem:packing MWU}
by providing an efficient implementation of Algorithm~\ref{algorithm:MWU} from
\Cref{sec:mincost_MWU}.

\paragraph{$\protect\Otil(m^{2}/\epsilon^2)$-time Algorithms for General Graphs.}

In line~\ref{enu:oracle} of Algorithm~\ref{algorithm:MWU},
we invoke the algorithm of Mehlhorn~\cite{Mehlhorn88}
for $2$-approximation of the minimum-length Steiner tree
(it clearly approximates also the minimum-length Steiner subgraph,
because an optimal Steiner subgraph is always a tree).
Each invocation runs in time $O(m+n\log n)$,
which subsumes the time complexity of other steps in a single iteration of the while loop.
There are $\Otil(m/\eps^{2})$ iterations by \Cref{lem:aug bound}\ref{it:AugBound1},
hence the total running time is $\Otil(m^{2}/\eps^{2})$.
The packing $\P$ produced as output is a $(2+O(\eps))$-approximate solution
by \Cref{lem:approx bound}.

\paragraph{$\protect m^{1+o(1)}/\epsilon^2$-time Algorithms for Unweighted Graphs.}

In unweighted graphs, we obtain an almost-linear time
by invoking the data structure ${\cal D}$ from \Cref{thm:dec_MinSteiner},
which is a dynamic algorithm that maintains a graph $G$ under edge-weight increases
and can report a $(2+\eps)$-approximate minimum-length Steiner subgraph. 
Then in line~\ref{enu:oracle} of Algorithm~\ref{algorithm:MWU},
we query ${\cal D}$ to obtain a graph $H$,
and in line~\ref{enu:addlength} we instruct ${\cal D}$ to update the edge lengths.
The packing $\P$ produced as output is again
an $(2+O(\eps))$-approximate solution by \Cref{lem:approx bound},
and it only remains to analyze the total running time. 

A small technical issue is that ${\cal D}$ maintains a graph
whose edge lengths are integers in the range $\{1,\ldots,W\}$.
Our edge lengths $\ell(e)$ lie in the range $[(2m)^{-1/\epsilon},(1+\eps)]$,
which can be scaled to the range $[\frac{1}{\eps}, \frac{1}{\eps}(1+\eps)(2m)^{1/\epsilon}]$.
Each update to ${\cal D}$ is further rounded upwards to the next integer,
and thus we can use $W=\lceil \frac{1}{\eps}(1+\eps)(2m)^{1/\epsilon}\rceil$. 
The rounding increases edge lengths multiplicatively by at most $1+\eps$,
and thus each reported subgraph achieves $(2+O(\eps))$-approximation
with respect to the unrounded lengths.
We can prevent the rounding error from accumulating,
by maintaining a table of the exact lengths of each edge (separately from $\cal D$),
and using this non-rounded value when computing the edge's new length
in line~\ref{enu:addlength}.

The overall running time is dominated by the time spent by updates and queries
to the data structure ${\cal D}$.
By \Cref{lem:aug bound}\ref{it:AugBound1},
each edge participates in at most $\Otil(1/\eps^{2})$ augmentations
and thus appears in at most $\Otil(1/\eps^{2})$ different Steiner subgraphs $H$.
As each iteration updates only the lengths of edges in its subgraph $H$,
the total number of edge-length updates, over the entire execution,
is at most $\sum_{H\in\P}|E(H)| = \Otil(m/\eps^{2})$.
By \Cref{thm:dec_MinSteiner}, the total update time of ${\cal D}$
is $mn^{o(1)}\log^{O(1)}W  + O(\sum_{H\in\P}|E(H)|) =mn^{o(1)}/\epsilon^2$.
Lastly, ${\cal D}$
returns each Steiner subgraph $H$ in time $|E(H)|\cdot n^{o(1)}\log^{O(1)}W$,
and therefore the total time of all queries is
$\sum_{H\in\P}|E(H)|\cdot n^{o(1)}\log^{O(1)}W = mn^{o(1)}/\epsilon^2$.
Thus, the overall running time is $mn^{o(1)}/\epsilon^2$.


\section{Decremental $(2+\eps)$-approximate Minimum Steiner Subgraphs}
\label{sec:unweighted}


The goal of this section is to establish an algorithm for maintaining a $(2+\eps)$-approximate minimum Steiner subgraph with $n^{o(1)}$ amortized time per weight-increase update, and $n^{o(1)}$ output-sensitive query time.
In this paper, the main application of this algorithm is in proving Theorem~\ref{lem:guide}: fractionally packing $\geq\lambda(U)/(4+\eps)$ $U$-Steiner-subgraphs in an unweighted graph with $U$-Steiner mincut $\lambda(U)$ in \emph{almost-linear time}.
More specifically, this result gets used inside the MWU framework (Lemma~\ref{lem:packing MWU}).
However, it can also be viewed as a result of independent interest: obtaining an almost-optimal algorithm for a classic problem in the decremental (weight-increase) setting.
To be consistent with Section~\ref{sec:mincost_MWU} we refer to the edge-weights as \emph{lengths}.
A key feature of this new algorithm is that it is deterministic, and therefore works in an adaptive-adversary setting (such as MWU).

\begin{theorem}[Decremental Approximate Minimum Steiner Subgraph]
\label{thm:dec_MinSteiner}
For any constant $\eps>0$, there is a deterministic data-structure that maintains a graph $G=(V,E,\ell:E \to [L])$ with terminals $U \subseteq V$ under a sequence of updates (edge-length-increase only) in total $m\cdot n^{o(1)}\cdot \log^{O(1)}{L}+O(y)$ time, and can, at any time, produce a $(2+\eps)$-approximate minimum $\ell$-length $U$-Steiner subgraph $H$ of $G$ in time $m'\cdot n^{o(1)}\cdot \log^{O(1)}{L}$ where $m'$ is the number of edges in the subgraph it outputs.
\end{theorem}

At a high-level, there are two distinct steps for establishing this result:

\begin{enumerate}
\item \textbf{A reduction to decremental Single-Source Shortest Paths (SSSP).} It is well-known that a $2$-approximation for Steiner Tree can be obtained by computing the Minimum Spanning Tree (MST) of a helper graph $H$. This helper graph is simply the complete distance graph of the terminals, and (naively) requires $\Omega(mn)$ time to compute.
The MST of $H$ then gets expanded into a Steiner subgraph of $G$.
Mehlhorn's algorithm achieves a $2$-approximation in $\tilde{O}(m)$ time by further reducing the problem to an SSSP computation, by seeking the MST of a slightly different helper graph $H$.
At a high-level, our plan is to design a dynamic version of Mehlhorn's algorithm. 

\item \textbf{An almost-optimal decremental algorithm for $(1+\eps)$-SSSP.}
The challenge with the above reduction is in the difficult task of maintaining the SSSP information that is required to construct the helper graph $H$ while the input graph $G$ is changing.
Roughly speaking, we succeed in doing that, while paying an additional $(1+\eps)$ factor, due to a recent decremental $(1+\eps)$-SSSP \emph{deterministic} algorithm by Bernstein, Gutenberg, and Saranurak \cite{BGS21_arxiv} with $n^{o(1)}\cdot \log^{O(1)}{L}$ time per update.
\end{enumerate}

Thus, the plan is as follows.
Let $G^{(i)}$ be the graph $G$ after the $i^{th}$ update.
Using a decremental $(1+\eps)$-SSSP algorithm we maintain a helper graph $H^{(i)}$. 
Then, using a fully-dynamic MST algorithm we maintain the MST of $H^{(i)}$.
Whenever there is a query, we output a Steiner subgraph $S^{(i)}$ of $G^{(i)}$ by querying for the MST of $H^{(i)}$ and expanding it into $S^{(i)}$.
Notably, whereas $G^{(i)}$ is obtained from $G^{(i-1)}$ in a decremental fashion, i.e. only via increasing edge-lengths, the updates to the helper graph can be both decremental and incremental. As a result, we need to maintain the MST in a fully-dynamic setting, which can be done in $\polylog(n)$ update time with a deterministic algorithm \cite{holm2001poly}.
This is schematically described in Figure~\ref{Figs:Diag_UW}.

\begin{figure}[ht]
  \begin{center}
    \ifarxiv
    \includegraphics[width=5.2in]{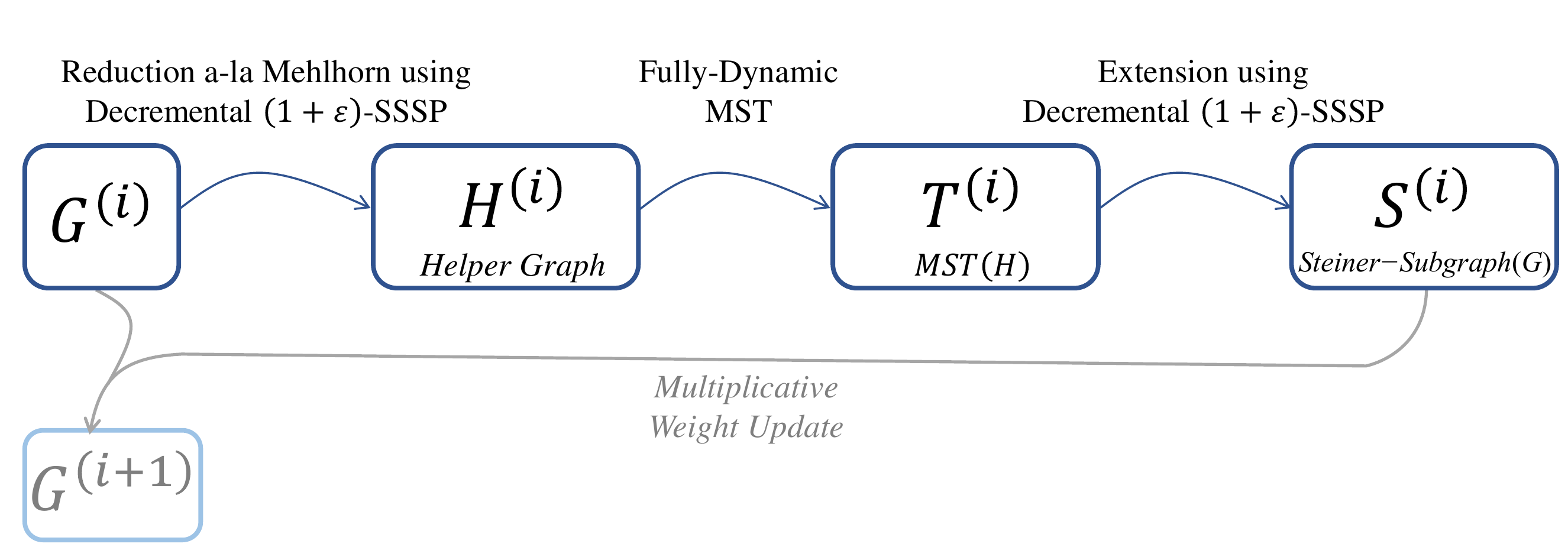}
    \else
    \includegraphics[width=5.2in]{Diag_UW}
    \fi
    \end{center}
    \caption{A diagram representing how an approximate minimum Steiner subrgaph is produced by our algorithm (inside an iteration of the MWU framework for packing Steiner subgraphs).}
    \label{Figs:Diag_UW}
\end{figure}

Realizing this plan, however, involves several lower-level challenges.
First of all, the construction of the helper graph from Mehlhorn's algorithm needs to be modified in order to prevent the number of changes to $H$ from being much larger than the number of changes to $G$. (Otherwise, each edge update in $G$ could result in $\Omega(n)$ edge changes in $H$.)
This is presented in Section~\ref{sec:construction}.
Second, the new construction of a helper graph has slightly different properties than Mehlhorn's which we state and analyze next in Section~\ref{sec:properties}.
And finally, the reduction requires additional features from the  $(1+\eps)$-SSSP data-structure that were not explicitly stated in previous work. 
We explain how they can be achieved (with minor modifications) in the discussion below  
\Cref{lem:DStoHelper}.

\subsection{The Helper Graph}
\label{sec:properties}

We begin by describing what properties the helper graph $H$ should have, and by proving that they suffice for the MST of $H$ to yield a $(2+\eps)$-approximate $U$-Steiner tree for $G$.
The algorithm for dynamically maintaining a helper graph with these properties is discussed later on. 
Note that we will have a helper graph $H^{(i)}$ of $G^{(i)}$ in each step $i$, but the properties we describe in this section are indifferent to the step and to the fact that we are in a dynamic environment; therefore we will drop the $(i)$ superscript in this subsection.

The starting point of both our helper graph and Mehlhorn's original construction is the well-known fact that an MST of the complete distance graph of the terminals is a $2$-approximation to the Steiner tree.

\begin{lemma}[\cite{HR92,Tak80,KMB81,Ple81}]
\label{lem:2apx}
Let $G=(V,E,\ell:E \to [L])$ be a graph with terminals $U \subseteq V$, and let $H^*=(U,E_H,{\ell}_H:E_H \to [2nL])$ be the distance graph of the terminals, i.e. $\forall t_1,t_2 \in U: \ {\ell}_{H^*}(t_1,t_2) = d_G (t_1,t_2)$.
Then, the MST of $H$ has length at most $2$ times the minimum $U$-Steiner tree of $G$.
\end{lemma}

The issue with this helper graph $H^*$ is in the complexity of computing (and maintaining) it.
Mehlhorn's fast algorithm is based on the observation that a different helper graph, that can be viewed as a proxy towards $H^*$, suffices to get a $2$-approximation. 
The exact property of a helper graph that Mehlhorn's algorithm is based on is stated in the next lemma. (The analysis of our helper graph will not use this Lemma; we only include it for context.) 

\begin{lemma}[\cite{Mehlhorn88}]
Let $G=(V,E,\ell:E \to [L])$ be a graph with terminals $U \subseteq V$, and let $H=(U,E_H,{\ell}_H:E_H \to [2nL])$ be a weighted graph such that for each edge $(x,y) \in E$ there exists an edge $(t_x,t_y) \in E_H$, where $t_x$ and $t_y$ are the closest terminals to $x$ and $y$ (respectively), of length $d_G(t_x,x)+{\ell}_G(x,y)+d_G(y,t_y)$.
Then, the MST of $H$ has length at most $2$ times the minimum $U$-Steiner tree of $G$.
\end{lemma}

Notice that a pair of vertices in $H$ may have many parallel edges between them with different lengths. (See Figure~\ref{Figs:Hs} for an illustration.) 

\begin{figure}[ht]
  \begin{center}
    \ifarxiv
    \includegraphics[width=6.6in]{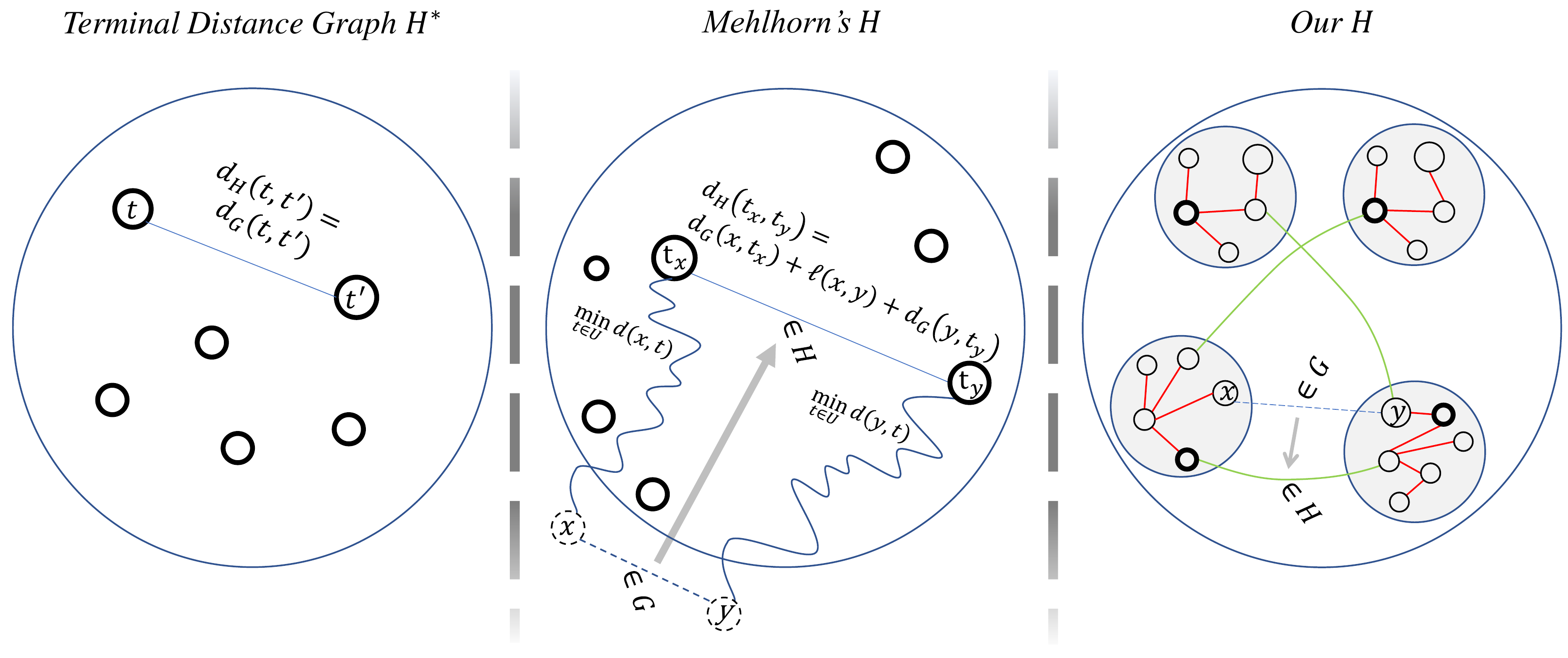}
    \else
    \includegraphics[width=6.6in]{Hs}
    \fi
    \end{center}
    \caption{An illustration of the different kinds of helper graphs discussed in this paper. Dashed lines represent edges in $G$, and curves represent paths in $G$.}
    \label{Figs:Hs}
\end{figure}

As an intermediate step towards defining the properties of our actual helper graph, we consider the following approximate version of Mehlhorn's definition.
We say that $t_x$ is a \emph{$(1+\eps)$-closest terminal to $x$} iff $d_G(x,t_x) \leq (1+\eps)\cdot  d_G(x,t)$ for all $t \in U$.

\begin{lemma}
\label{lem:apxMehlhorn}
Let $G=(V,E,\ell:E\to [L])$ be a graph with terminals $U \subseteq V$, and let $H=(U,E_H,{\ell}_H:E_H \to [2nL])$ be a graph such that: (1) for each vertex $x \in V$ there is a unique terminal $t_x\in U$ assigned to $x$, where $t_x$ is a $(1+\eps)$-closest terminal to $x$, and (2) for each edge $(x,y) \in E$ there exists an edge $(t_x,t_y) \in E_H$ of length $\leq (1+\eps)\cdot(d_G(t_x,x)+{\ell}_G(x,y)+d_G(y,t_y))$.
Then, the MST of $H$ has length at most $2(1+\eps)^2$ times the minimum $U$-Steiner tree of $G$.
\end{lemma}

\begin{proof}
Let $T^*$ be an MST of the ideal helper graph $H^*$ (the complete distance graph of the terminals).
We will prove that there exists an MST $T$ of our helper graph $H$ with length ${\ell}_H(T) \leq (1+\eps)^2\cdot {\ell}_{H^*}(T^*)$.
Then by Lemma~\ref{lem:2apx} we get that the MST of $H$ has length at most $2(1+\eps)^2$ times the minimum $U$-Steiner tree of $G$.

We describe a process that modifies $T^*$ by replacing its edges with edges from $H$ while keeping it a terminal spanning tree and while only increasing its length by $(1+\eps)^2$ (see Figure~\ref{Figs:Proof_54}).
\begin{figure}[!htbp]
  \begin{center}
    \ifarxiv
    \includegraphics[width=3.0in]{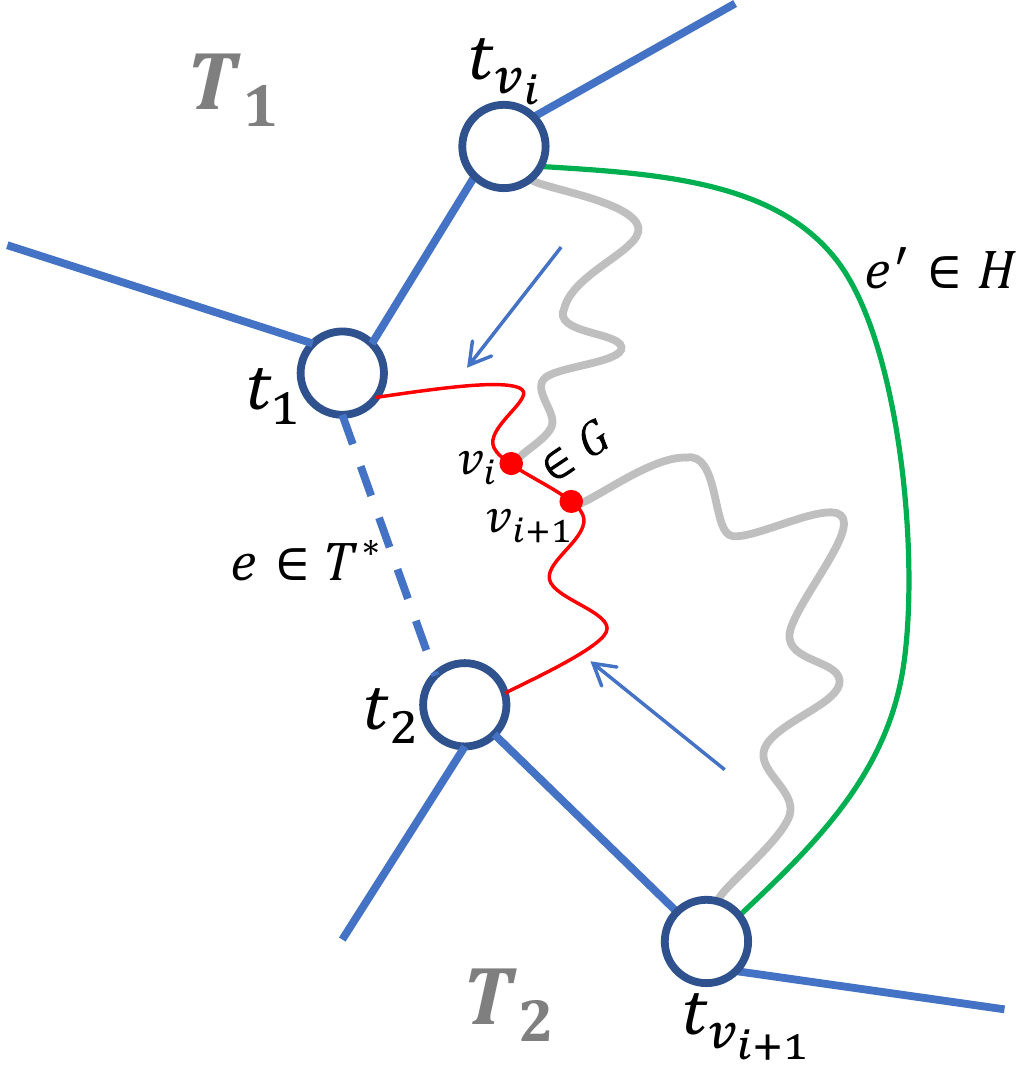}
    \else
    \includegraphics[width=3.0in]{Proof_54}
    \fi
    \end{center}
    \caption{An illustration representing the process of replacing edges from $T^*$ by edges from $H$. Here, $e\in T^*$ is replaced with $e'\in H$. Gray curves represent paths, the red curve represents the shortest $t_1t_2$-path in $G$, and blue arrows represent the length charging.}
    \label{Figs:Proof_54}
\end{figure}
Specifically, each edge $e \in T^*$ gets replaced by an edge $e' \in H$ of length ${\ell}_H(e') \leq (1+\eps)^2\cdot {\ell}_{H^*}(e)$.
We will do this replacement as long as there is an edge in $T^*$ that is not in $H$, or is in $H$ but has a different length than it does in $H^*$.

Let $e=(t_1,t_2) \in T^*$ be such an edge, and let $v_1,\ldots,v_{p}$ be a shortest path from $t_1$ to $t_2$ in $G$ with total length $d_G(t_1,t_2)={\ell}_{H^*}(t_1,t_2)$.
The removal of $e$ from $T^*$ disconnects the tree into two components, call them $t_1\in T_1$ and $t_2 \in T_2$. 
Let $i \in [p]$ be such that $t_{v_i} \in T_1$ but $t_{v_{i+1}} \in T_2$.
Such an $i$ must exist because $t_{v_1} = t_1 \in T_1$ and $t_{v_{p}} = t_2 \in T_2$.

Consider the edge $e' = (t_{v_i}, t_{v_{i+1}}) \in H$ that follows from the edge $(v_i,v_{i+1})$ by the assumptions on the helper graph $H$.
Replacing $e$ with $e'$ keeps $T^*$ a terminal spanning tree.
It only remains to bound the difference in lengths:
$$
{\ell}_H(e') \leq (1+\eps)\cdot(d_G(t_{v_i},v_i)+{\ell}_G(v_i,v_{i+1})+d_G(v_{i+1},t_{v_{i+1}})).
$$
Since $t_{v_i}$ is a $(1+\eps)$-closest terminal to $v_i$ we know that $d_G(t_{v_i},v_i) \leq (1+\eps) \cdot d_G(t_1,v_i)$ and similarly $d_G(v_{i+1},t_{v_{i+1}}) \leq (1+\eps)\cdot d_G(v_{i+1},t_2)$.
Therefore,
$$
{\ell}_H(e') \leq (1+\eps)^2\cdot(d_G(t_1,v_i)+{\ell}_G(v_i,v_{i+1})+d_G(v_{i+1},t_2))
$$
$$
= (1+\eps)^2\cdot d_G(t_1,t_2) = (1+\eps)^2\cdot {\ell}_{H^*}(e).
$$

\end{proof}

Our actual helper graph, described next, differs from Mehlhorn's not only in the fact that it is approximate, but more importantly, it is a graph on the entire vertex set $V$ rather than just the terminals.
This is done for efficiency; the (admittedly vague) intuition is as follows. Updating each edge of $H$ after each update to $G$ is too costly because the closest (or even the $(1+\eps)$-closest) terminal to a vertex might change very frequently. 
But one can observe that most of these frequent changes are not important for the MST; there is usually a parallel edge that was not affected and whose length is within $(1+\eps)$ of the new edges that we would want to add. 
Figuring out which changes in $H$ are important and which are not does not seem tractable.
Instead, our new idea is to make the update only in a helper graph that is, on the one hand, much closer to $G$ than to Mehlhorn's helper graph (and therefore doesn't change too frequently), and on the other hand, has the same MST up to $(1+\eps)$ as Mehlhorn's. 
In other words, the idea is to let the MST algorithm do the work of figuring out which edges are relevant.

The properties that our actual helper graph will satisfy are described in the following lemma.
In words, each terminal $t$ has a disjoint component $C_t$ in $H$ of vertices that can be reached from $t$ with distance $0$ (due to edges of length zero in $H$).
The components morally represent the approximate Voronoi cells of the terminals.
And instead of demanding that each edge $(x,y)$ of the original graph be represented by an edge $(t_x,t_y)$ between the two corresponding terminals in the helper graph, we only demand that there be \emph{some} edge $(x',y')$ between the two corresponding components (where $x',y'$ need not be terminals).

\begin{lemma}
\label{lem:helperproperties}
Let $G=(V,E,\ell:E\to [L])$ be a graph with terminals $U \subseteq V$, and let $H=(V_H,E_H, \ell_H:E_H \to [2nL] \cup \{0\})$ be a weighted graph satisfying the following properties:
\begin{itemize}
    \item for each vertex $x \in V_H$ there is a unique terminal $t_x \in U$ such that $d_H(x,t_x)=0$,
    \item for each edge $(x,y) \in E$ there exists an edge $(x',y') \in E_H$ with $d_H(x',t_x)=0$ and $d_H(y',t_y)=0$, where $t_x$ and $t_y$ are $(1+\eps)$-closest terminals to $x$ and $y$ in $G$ (respectively), of length $\ell_H(x',y')\leq (1+\eps)\cdot(d_G(t_x,x)+\ell_G(x,y)+d_G(y,t_y))$.
\end{itemize}
Then, the MST of $H$ has length at most $2(1+\eps)^2$ times the minimum $U$-Steiner tree of $G$.
\end{lemma}

\begin{proof}
For each terminal $t\in U$ let $C_t = \{ x \in V_H: d_H(t,x)=0 \}$ be the component of vertices at distance zero to $t$.
By the assumption on $H$, the sets $\{C_t\}_{t \in U}$ form a partition of $V_H$.
Let $\bar{H}$ be the graph obtained from $H$ by contracting each $C_t$ into a single vertex.
It immediately follows that $\bar{H}$ is a helper graph that satisfies requirements of Lemma~\ref{lem:apxMehlhorn} above, and therefore the MST of $\bar{H}$ has length at most $2(1+\eps)^2$ times the minimum $U$-Steiner tree of $G$.
But any MST of $\bar{H}$ can be expanded into an MST of $H$ with the same length, by taking a zero-length spanning tree for each component $C_t$; the shortest path tree rooted at $t$ is such a tree.
\end{proof}

\subsection{Maintaining the Helper Graph Dynamically: A Reduction to SSSP}
\label{sec:construction}

The main technical result towards proving Theorem~\ref{thm:dec_MinSteiner} is an algorithm for maintaining a helper graph with the properties required by Lemma~\ref{lem:helperproperties}.

\begin{lemma}
\label{lem:helpermaintenance}
Given a graph $G=(V,E,\ell:E \to [L])$ with terminal set $U \subseteq V$ and a sequence of length-increase updates, 
it is possible to maintain a helper graph $H=(V_H,E_H, \ell_H:E_H \to [2nL] \cup \{0\})$ satisfying the properties in Lemma~\ref{lem:helperproperties} with a deterministic algorithm such that:
\begin{itemize}
    \item The total time is $m\cdot n^{o(1)}\cdot \log^{O(1)}{L}$, and
    \item The recourse, i.e. number of edge-length changes to $H$ throughout the sequence, is $m\cdot n^{o(1)}\cdot \log^{O(1)}{L}$.
    \item Given an MST of $H$ of total length $mst(H)$ it is possible to produce a $U$-Steiner subgraph of $G$ of total length $\leq (1+\eps) \cdot mst(H)$, in time $m'\cdot n^{o(1)}\cdot \log^{O(1)}{L}$ where $m'$ is the number of edges in the subgraph. 
\end{itemize}
\end{lemma}

This lemma is proved in the next sections.
But first, let us explain how Theorem~\ref{thm:dec_MinSteiner} follows from it.

\begin{proof}[Proof of Theorem~\ref{thm:dec_MinSteiner}]
Let $G$ be the input graph that is undergoing a sequence of length-increase updates; let $G^{(i)}$ be the graph after the $i^{th}$ update. 
We use the algorithm in Lemma~\ref{lem:helpermaintenance} to maintain a helper graph $H$ as $G$ changes; let $H^{(i)}$ be the helper graph after the $i^{th}$ update. 
We use the fully-dynamic MST algorithm by Holm et al.~\cite{holm2001poly}, with $H$ as the input, to (explicitly) maintain an MST of $H$ as $G$ changes.
If after the $i^{th}$ update, there is a query asking for a $U$-Steiner subgraph for $G^{(i)}$, we first inspect the MST data-structure to obtain an MST $T^{(i)}$ of $H^{(i)}$.
We use the algorithm in Lemma~\ref{lem:helpermaintenance} again to expand the MST $T^{(i)}$ into a $U$-Steiner subgraph $S^{(i)}$ of $G^{(i)}$.
If the total length of the MST is $x$ then the total length of the Steiner subgraph is $\leq (1+\eps)\cdot x$, and therefore, since the helper graph satisfies the properties of Lemma~\ref{lem:helperproperties}, we get that $S^{(i)}$ is a $2(1+\eps)^3$-minimum $U$-Steiner subgraph for $G^{(i)}$.

The running time for maintaining the helper graph $H$ is $m\cdot n^{o(1)}\cdot \log^{O(1)}{L}$.
To bound the running time of the MST algorithm we first need to bound the number of updates we make to $H$ (not $G$) because $H$ is the input graph to the MST algorithm.
This number of updates is exactly the recourse and it is bounded by $m\cdot n^{o(1)}\cdot \log^{O(1)}{L}$.
The MST algorithm \cite{holm2001poly} has $\polylog(n)$ amortized update time (and is deterministic) and therefore it can support this number of updates in a total of $m\cdot n^{o(1)}\cdot \log^{O(1)}{L}$ time as well.
Finally, the time to produce the $U$-Steiner subgraph is $m'\cdot n^{o(1)}\cdot \log^{O(1)}{L}$ where $m'$ is the number of edges in the subgraph, as required by Theorem~\ref{thm:dec_MinSteiner}.
\end{proof}

\subsubsection{The Basic Construction}

In order to understand our final construction of the helper graph $H$, it is helpful to first see a simplified version that works assuming an SSSP data-structure with (slightly) stronger properties than currently achievable.

Let us first briefly overview how a decremental SSSP data-structure works, to clarify the context of some of the notation below. 
Naturally, the idea is to maintain a shortest-path tree rooted at the source $s$. 
For efficiency, it is desirable that the tree's depth be small.
For this reason, the data-structures (e.g.~\cite{henzinger2014decremental,BGS21_arxiv}) typically use a \emph{hopset}: by adding a set $E_{hopset}$ of $m \cdot n^{o(1)}$ edges to the graph, it is guaranteed that there exist a path from $s$ to any $v$ that is $(1+\eps)$-shortest-path and that only uses $n^{o(1)}$ edges; namely, it has only $n^{o(1)}$ hops.
A \emph{hopset-edge} $e \in E_{hopset}$ typically has a length that corresponds to the shortest path distance between its endpoints, and therefore the distance in the new graph is never shorter than the original distance.
The data-structure then maintains an approximate shortest-path tree $T_{hopset}$ on the graph with the hopset, of depth $n^{o(1)}$, which can be done efficiently.
The distances that it reports are exactly the distances in this tree.

\begin{proposition}
\label{prop:simpler}
Suppose that there is a deterministic data-structure that given a graph $G=(V,E,\ell:E \to [L])$ and a single-source $s \in V$, that undergoes a sequence of edge-length-increase updates, and supports the following operations in $m \cdot n^{o(1)}\cdot \log^{O(1)}{L}$ total time:
\begin{itemize}
    \item Explicitly maintains an estimate $\tilde{d}(s,v)$ for each $v\in V$, such that $d(s,v) \leq \tilde{d}(s,v) \leq (1+\eps) \cdot d(s,v)$. In particular, can report after each update, all vertices $v \in V$ such that $\tilde{d}(s,v)$ increased by a $\geq (1+\eps)$ factor (since the last time $v$ was reported).
    \item Explicitly maintains a tree $T_{hopset}$ rooted at $s$, such that $\tilde{d}(s,v)=d_{T_{hopset}}(s,v)$ for all $v \in V$. An edge of the tree is either from $E$ or from a set of \emph{hopset-edges} $E_{hopset} \subseteq V \times V$.
    In particular, can report after each update, all edges that are added or removed from $T_{hopset}$.
\end{itemize}
Moreover, it holds that each hopset-edge $e=(u,v) \in E_{hopset}$ of length $w$ is associated with a $u,v$-path $\pi_e$ in $G$ of length $w$. And that the number of times an edge $e \in E$ may appear on a single root-to-leaf path in $T_{hopset}$, either as an edge $e \in E$ or as an edge $e \in \pi_{e'}$ on the path of a hopset-edge $e' \in E_{hopset}$, is $n^{o(1)}$.  
Additionally, the data-structure supports the following query:
\begin{itemize}
    \item Given a hopset-edge $e$ and an integer $k\geq 1$, can return the first (or last) $k$ edges of $\pi_e$ in time $k \cdot n^{o(1)}\cdot \log^{O(1)}{L}$.
\end{itemize}
Then Lemma~\ref{lem:helpermaintenance} holds.
\end{proposition}

The main difference between the requirements of Proposition~\ref{prop:simpler} and the features of the existing $(1+\eps)$-SSSP algorithm is simply that the data-structure explicitly maintains $O(\log{n})$ trees instead of one, and the $i^{th}$ tree is only guaranteed to approximate to within $(1+\eps)$ the distances $d(s,v)$ when they are in the range $[\gamma^{i-1},\gamma^i]$ (for some $\gamma\geq 2$).
Another minor gap is that the hopsets use auxiliary vertices, in addition to the shortcut edges.
These gaps will be handled in the next subsection.
For now, let us see how to obtain Proposition~\ref{prop:simpler}.

\begin{proof}[Proof of Proposition~\ref{prop:simpler}]
We start off with a trick similar to the way Mehlhorn's algorithm reduces the problem of computing the closest terminal to each vertex (i.e. the Voronoi cell of each terminal) to an SSSP problem.
Let $G'$ be the graph obtained from $G=(V,E_G,\ell_G:E_G \to [L] )$ by adding a super-source $s^*$ that is connected with edges of length $0$ to all terminals $U \subseteq V$ (see Figure~\ref{Figs:Tree_GR}).
\begin{figure}[ht]
  \begin{center}
    \ifarxiv
    \includegraphics[width=4.8in]{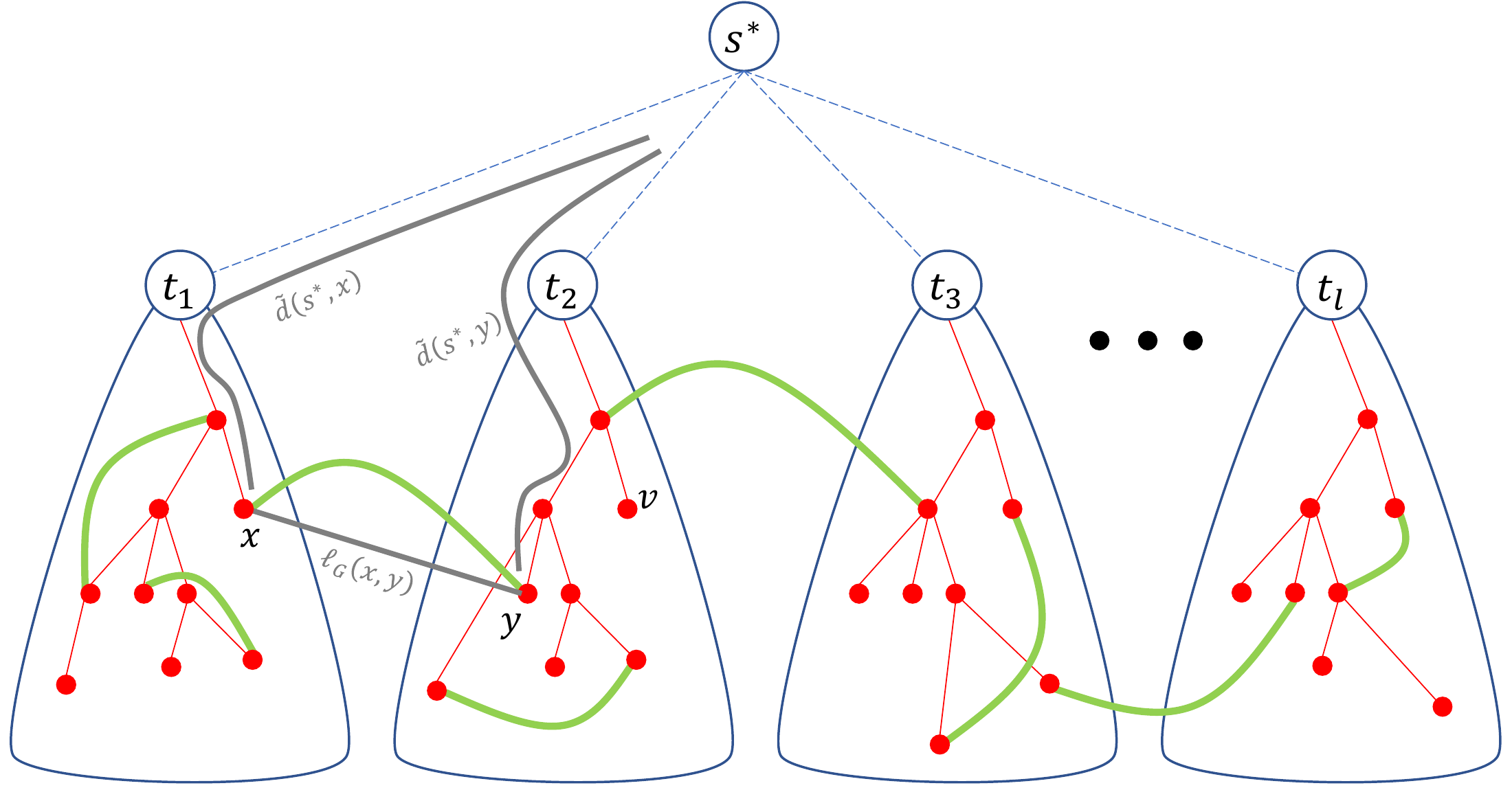}
    \else
    \includegraphics[width=4.8in]{Tree_GR}
    \fi
    \end{center}
    \caption{An illustration representing $H$ in the basic construction. Here, $d_{G'}(s^*,v)=\min_{t\in U}d_{G}(t,v)=d_G(t_2,v)$. Furthermore, the green edge $(x,y)\in H$ has length that is the sum of the gray curves.}
    \label{Figs:Tree_GR}
\end{figure}
Consequently, for any vertex $v \in V$, $d_{G'}(s^*,v) = \min_{t \in U} d_G(t,v)$, and moreover $v$ belongs to the subtree rooted at $\arg\min_{t \in U} d_G(t,v)$ in the SSSP-tree rooted at $s^*$.
Note that this can be achieved without really using edges of length $0$ by simply contracting all terminals into one vertex and calling it $s^*$.\footnote{The description below assumes we do add the zero-length edges. In the full construction in the next subsection we do in fact need to use the contractions, and so the arguments differ in minor ways.}
We will run the $(1+\eps)$-SSSP data-structure in the statement on $G'$ with $s^*$ as the source. Note that it is trivial to transform the length-increase sequence of updates to $G$ into a sequence of length-increase updates on $G'$.
Let $T_{hopset}$ and $\tilde{d}(s^*,v)$ be the tree and estimates maintained by the data-structure.

\paragraph{Maintaining a helper graph}
The helper graph $H=(V,E_H,\ell_H:E_H \to [2nL]\cup \{0\})$ is defined based on $G$, $T_{hopset}$ and $\tilde{d}(s^*,v)$ as follows.
Note that the vertex set of $H$ is the same as $G$'s; in our full construction (in the next subsection) there will be additional vertices.
There are two kinds of edges in $H$:

\begin{itemize}
    \item \emph{Green edges}: for each edge $(x,y) \in E_G$ we add an edge $(x,y)$ to $E_H$ with length 
    $$\ell_H(x,y) := \tilde{d}(s^*,x) + \ell_G(x,y) + \tilde{d}(s^*,y).$$ 
    A green edge gets updated whenever: (1) the length $\ell_G(x,y)$ gets updated and is more than $(1+\eps)$ times the value used in the current $\ell_H(x,y)$, or (2) the data-structure reports that $\tilde{d}(s^*,x)$ or $\tilde{d}(s^*,y)$ is increased by a $\geq (1+\eps)$ factor.
    \item \emph{Red edges}: for each edge $(x,y) \in E(T_{hopset})$ we add an edge $(x,y)$ of length $0$ to $E_H$.\footnote{The zero-length edges between $s^*$ and each terminal in $T$ are exempted, since $s^*$ does not exist in $H$.}
    A red edge gets updated whenever the tree $T_{hopset}$ changes. 
    If an edge is removed from the tree, the red edge is removed (i.e. the length is increased from $0$ to infinity),\footnote{Notably, these are the only incremental updates that we do.} and if a new edge is added to the tree we add an edge of length zero to $H$.
\end{itemize}

\paragraph{The recourse of the helper graph}
The data-structure will not only maintain the helper graph $H$ described above, but it will also feed it into a dynamic MST algorithm and explicitly maintain an MST of $H$.
For this to be efficient, we must argue that the total number of updates we make to $H$ is small.
The running time of the data-structure is upper bounded by $m \cdot n^{o(1)}\cdot \log^{O(1)}{L}$, and since it stores the tree $T_{hopset}$ explicitly, this is also an upper bound on the total number of changes to $T_{hopset}$.
Whenever a red edge is updated it is because the corresponding edge was updated in $T_{hopset}$; this can happen $m \cdot n^{o(1)}\cdot \log^{O(1)}{L}$ times.
A green edge $(x,y)$ is updated either when the corresponding edge in $G$ has its length increase by at least a $(1+\eps)$ factor, or when $\tilde{d}(s^*,x)$ or $\tilde{d}(s^*,y)$ increase by $\geq (1+\eps)$.
Since the lengths and estimates are upper bounded by $nL$, they can only increase by a $(1+\eps)$ factor $O(\log{nL})$ times. Each time an estimate  $\tilde{d}(s^*,x)$  increases we would update $\deg_G(x)$ edges in $H$.
Thus, the total number of updates to the green edges is $O(m \log{n}) +\sum_{x \in V} \deg_G(x) \cdot O(\log {nL}) = O(m \log {nL})$.
It follows that the running time and total number of updates to $H$ is $m\cdot n^{o(1)}\cdot \log^{O(1)}{L}$.

\paragraph{The properties of the helper graph}
Next, let us prove that $H$ satisfies the two properties in Lemma~\ref{lem:helperproperties}.
For that, it is helpful to establish the following claims.

\begin{claim}
In the tree $T_{hopset}$ the root $s^*$ has all of the terminals in $U$ as children.
\end{claim}
\begin{proof}
Suppose for contradiction that $t_1 \in U$ is in the subtree of terminal $t_2 \in U$.
Then it must follow that $\tilde{d}(s^*,t_1)=d_{T_{hopset}}(s^*,t_1) = 0+ d_{T_{hopset}}(t_2,t_1)>0$ because all edges that are not adjacent to $s^*$ in $G'$ have non-zero length. On the other hand, $d_{G'}(s^*,t_1)=0$, contradicting the fact that $\tilde{d}(s^*,t_1) \leq (1+\eps)\cdot d_{G'}(s^*,t_1)$.
\end{proof}

The first of the two properties in Lemma~\ref{lem:helperproperties} now follows.
In the tree $T_{hopset}$, each vertex $x \in V \setminus U$ appears in the subtree of exactly one terminal $t_x \in U$ and, due to the red edges, has distance $d_H(x,t_x)=0$. For all other terminals $t' \neq t_x \in U$ the path from $t'$ to $x$ must use a green edge of length $>0$.

\begin{claim}
\label{cl:closest_terminal}
For all $x \in V$, $\tilde{d}(s^*,x) \leq (1+\eps)\cdot \min_{t \in U} d_G(t,x)$.
Moreover, if $x$ is in the subtree of terminal $t_x \in U$ then $t_x$ is a $(1+\eps)$-closest terminal to $x$.
\end{claim}
\begin{proof}
The first part follows from the fact that $d_{G'}(s^*,x)=\min_{t \in U} d_G(t,x)$, and the guarantee on the estimates.
For the second part note that $d_G(t_x,x) \leq d_{T_{hopset}}(t_x,x) = \tilde{d}(s^*,x) \leq (1+\eps)\cdot \min_{t \in U} d_G(t,x)$; thus, $t_x$ is a $(1+\eps)$-closest terminal to $x$.
\end{proof}

To prove the second property, let $(x,y) \in E$ be an edge in $G$; we will show that the green edge $(x,y) \in E_H$ satisfies the requirement.\footnote{That is, we prove the property with $x=x'$ and $y=y'$. This will not be the case in the full construction in the next subsection.}
Let $t_x$ and $t_y$ be such that $d_H(x,t_x)=0$ and $d_H(y,t_y)=0$.
By the above claim, we know that $t_x$ and $t_y$ are $(1+\eps)$-closest terminals to $x$ and $y$ (respectively) in $G$.
Finally, we can upper bound the length of the green edge $(x,y)$ in $H$ by 
$$\ell_H(x,y) = \tilde{d}(s^*,x) + \ell_G(x,y) + \tilde{d}(s^*,y) = \tilde{d}(t_x,x) + \ell_G(x,y) + \tilde{d}(t_y,y)$$
$$ \leq (1+\eps)\cdot(d_G(t_x,x)+\ell_G(x,y)+d_G(y,t_y)).$$

\paragraph{Expanding an MST of $H$ into a $U$-Steiner subgraph of $G$}
By the above we know that the length of an MST of $H$ is a $(2+\eps)$-approximation to the minimum length of a Steiner subgraph of $G$.
So if our goal was only to maintain an approximation to the \emph{length} of a Steiner-Tree, we would just use the length of a given MST of $H$ and be done.
Alas, for our applications the data-structure must return the subgraph itself, which is more complicated. The process of (efficiently) expanding an MST into a Steiner subgraph is as follows.

%
For a terminal $t \in U$ we define $t$'s component (or subtree or Voronoi cell) to be the set of vertices $\{x \in V: d_H(t,x)=0\}$ that are connected by red edges to $t$.
Pick an arbitrary terminal $t_1 \in U$ as a starting point.
We will construct a Steiner Subgraph $S'$ by starting from $S'=\{t_1\}$ and iteratively adding paths (in $G$) to it.
Let $U' \subseteq U$ be the set of terminals that are spanned by $S'$; initially $U'=\{t_1\}$.
In each step, we pick an arbitrary \emph{green} edge of the MST that has an endpoint in one of the components of the terminals in $U'$. 
Let $e=(x,y)$ be such an edge, and let $t_x,t_y$ be the two terminals such that $x,y$ are in their components (respectively). 
First, note that such an edge must exist as long as $U' \neq U$ (or else the MST does not span the entire graph).
Second, note that $t_x \in U'$ but $t_y \notin U'$ (otherwise there would be a cycle in the MST).
We would like to add to $S'$: (1) a path from $t_x$ to $x$, (2) the edge $(x,y)$, and (3) a path from $y$ to $t_y$.
This would connect $t_y$ to $t_x$ and therefore add $t_y$ to $U'$.
And we would like to have that the total length of (1) + (2) + (3) is at most $\ell_{H}(x,y)$.
Ignoring certain subtle details, this can be done by observing that $\ell_{H}(x,y) = d_{T_{hopset}}(t_x,x) + \ell_G(x,y) + d_{T_{hopset}}(t_y,y)$ and the distances in $T_{hopset}$ come from paths in $G$.
Namely, we can scan the path from $t_x$ to $x$ in $T_{hopset}$ and, whenever we encounter an edge from $E$ we simply add it to $S'$, while if we encounter an edge $e \in E_{hopset}$ we ask the data-structure to expand it into the path $\pi_e$ and add $\pi_e$ to $S'$.
If we do this, the total length of the final Steiner Subgraph $S'$ will be exactly $mst(H)$.

The subtle issue with the above is that the total running time for outputting a Steiner Subgraph with $m'$ edges may not be $m' \cdot n^{o(1)}\cdot \log^{O(1)}{L}$ even though every time the data-structure reports an edge it only costs $n^{o(1)}\cdot \log^{O(1)}{L}$ time.
Indeed, if $m''$ edges are reported throughout the above process, we only spend $m'' \cdot n^{o(1)}\cdot \log^{O(1)}{L}$ time, but due to edges being repeated many times, the number of edges that we finally output $m'$ could be much smaller than $m''$.\footnote{While the total length from the repetitions is still within our budget of $mst(H)$, the MWU framework does not allow us to spend more than $n^{o(1)}\cdot \log^{O(1)}{L}$ time on an edge we output, even if it is being repeated many times in the output.}
To resolve this issue we use a more careful process that uses the ability of the data-structure 
to list only the first or last $k$ edges on a path $\pi_e$ at cost only $k \cdot n^{o(1)}\cdot \log^{O(1)}{L}$; it turns out that if we see an edge that we have already seen, we can stop expanding $\pi_e$.

In more details, the step for expanding a green edge $(x,y)$ is as follows. We start from $t_y$ (rather than $t_x$) and scan the tree down to $y$. When we see an edge $e \in E$ we check if it is already in $S'$: 
\begin{itemize}
\item If it is not in $S'$ we add it to a set $S_{(x,y)}$ of edges that will be added to $S'$ at the end of the step. 
\item If $e$ happens to already be in $S'$, then we stop, add all edges in $S_{(x,y)}$ to $S'$, and finish the step of edge $(x,y)$.
This makes sense because the edges in $S_{(x,y)}$ form a path from $t_y$ to $e\in S'$ and all of $S'$ is connected to $t_x \in U'$, so we have accomplished the goal of connecting $t_y$ to $t_x$ (at a potentially cheaper length than $\ell_H(x,y)$).
\end{itemize}
And when we see an edge $e \in E_{hopset}$ we do the following for each $k=2^i, i=1,\ldots,\log{|\pi_e|}$: ask the data-structure to report the first (closest to $t_y$) $k$ edges of $\pi_{e}$, call them $e_1,\ldots,e_k \in E$, and check if they are already in $S'$:
\begin{itemize}
\item If none of them are in $S'$ we add them to the set $S_{(x,y)}$ of edges that will be added to $S'$ at the end of the step. 
\item Otherwise, let $e_i$ be the first (closest to $t_y$) reported edge that is already in $S'$. We stop the step, add all edges in $S_{(x,y)} \cup \{e_1,\ldots,e_{i-1} \}$ to $S'$, and finish the step of edge $(x,y)$.
The edges in $S_{(x,y)} \cup \{e_1,\ldots,e_{i-1} \}$ form a path from $t_y$ to $e_i\in S'$ and therefore to $t_x$, so we are done.
\end{itemize}
If we reach $y$ without stopping, we continue on to the edge $(x,y)$. Again, we check if it is already in $S'$; if not we add it to $S_{(x,y)}$, if it is, we stop and add all of $S_{(x,y)}$ into $S'$.
Similarly, after reaching $x$ we continue by scanning the tree up to $t_x$. When encountering an edge from $E$ or $E_{hopset}$ we do the exact steps above to decide if we should add it to $S_{(x,y)}$ or stop.
If we reach the $t_x$ without stopping, we add all of $S_{(x,y)}$ into $S'$ and finish the step.

We claim that the total time for the above process is $m' \cdot n^{o(1)}\cdot \log^{O(1)}{L}$ where $m'$ is the number of distinct edges that end up in $S'$. 
Consider the step corresponding to an edge $(x,y)$, and let $e^*$ be the first edge on the corresponding path from $t_y$ to $t_x$ (the one closest to $t_y$) that is already in $S'$ before the step. 
Let $e_1,\ldots,e_{i}$ be the $i$ edges that come before $e^*$ on the path, and let $m^*\leq i$ be the number of distinct edges among them; note that the step ends up adding $m^*$ new edges to $S'$ and our goal is to bound its running time by $m^* \cdot n^{o(1)}\cdot \log^{O(1)}{L}$.
By the assumption on the data-structure, each edge can appear only $n^{o(1)}$ times on a root-to-leaf path in $T_{hopset}$, and therefore $ m^* \geq i /n^{o(1)}$.
Due to the repeated doubling of $k$ in the above process, each edge might get reported an additional $O(\log{n})$ times, and more importantly, the number of edges $e'$ that come after $e^*$ on the path and that the data-structure reports (but do not end up in $S'$) is $\leq 2i$.
Therefore, the total time can be upper bounded by $O(i \log n) \cdot n^{o(1)}\cdot \log^{O(1)}{L}= m^* \cdot n^{o(1)}\cdot \log^{O(1)}{L}$.

\end{proof}

\subsubsection{The Full Construction}

We are now ready to present our actual construction. The proof is heavily based on the proof of the simpler construction above.
\begin{figure}[htbp]
  \begin{center}
    \ifarxiv
    \includegraphics[width=6.0in]{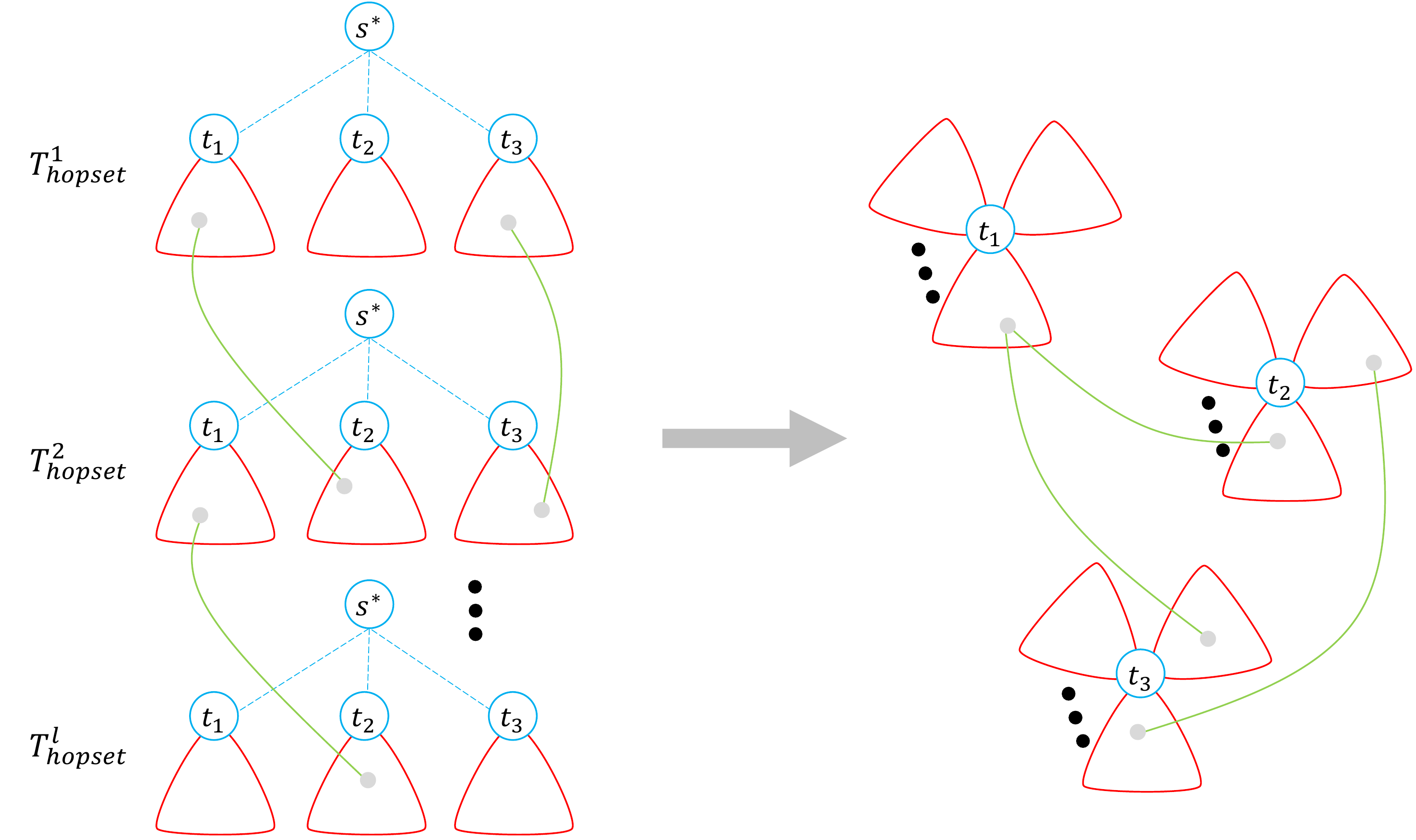}
    \else
    \includegraphics[width=6.0in]{Flowers}
    \fi
    \end{center}
    \caption{An illustration representing how the $O(\log{n})$ approximate SSSP trees maintained by the dynamic algorithm are combined in our helper graph.}
    \label{Figs:Flowers}
\end{figure}

\begin{lemma}
\label{lem:DStoHelper}
Suppose that for any constant $\eps>0$, there is a deterministic data-structure that given a graph $G=(V,E,\ell:E\to[L])$
and a source $s\in V$, that undergoes a sequence of edge-length-increase
updates of length $U$, and supports the following operations in $m\cdot n^{o(1)}\cdot\log^{O(1)}L+O(U)$
total time. There is a parameter $2\le\gamma\le n^{o(1)}$ such that,
for each $i\in[\log_{\gamma}nL]$:
\begin{itemize}
\item Explicitly maintains an estimate $\tilde{d}^{i}(s,v)$ for each $v\in V$,
such that $d(s,v)\leq\tilde{d}^{i}(s,v)$ and if $d(s,v)\in[\gamma^{i-1},\gamma^{i}]$
then $\tilde{d}^{i}(s,v)\leq(1+\eps)\cdot d(s,v)$. In particular,
can report after each update, all vertices $v\in V$ such that $\tilde{d}^{i}(s,v)$
increased. 
\item Explicitly maintains a graph $\Htil^{i}$ where $\Htil^{i}$ is not
a subgraph of $G$ and $V(\Htil^{i})$ contains all vertices of $G$ but
also a set $C^{i}$ of auxiliary vertices, which forms an independent
set in $\Htil^{i}$. 
For any pair of regular vertices $u,v\in V$, if $(u,v)\in E(\Htil^{i})$ has length $w$
or there is an auxiliary vertex $c\in C^{i}$ such that $(u,c),(c,v)\in E(\Htil^{i})$ is a $2$-path of length $w$, then there
is an associated $u,v$-path $\pi_{u,v}$ in $G$ of length between $w$ and $(1+\eps)\cdot w$.
Although $\pi_{u,v}$ is not a simple path, each
edge $e$ in $G$ may appear at most $\beta=n^{o(1)}$ times in $\pi_{u,v}$.
\item Explicitly maintains a tree $T_{hopset}^{i}$ rooted at $s$ in $\Htil^{i}$
of depth $n^{o(1)}$, such that $\tilde{d}^{i}(s,v)=d_{T_{hopset}^{i}}(s,v)$
for all $v\in V$. In particular, can report after each update, all
edges that are added or removed from $T_{hopset}$. 
\end{itemize}
Additionally, the data-structure supports the following query: 
\begin{itemize}
\item For each path $\pi_{u,v}$ in $\Htil^{i}$, given an integer $k\geq1$,
can return the first (or last) $k$ edges of $\pi_{u,v}$ in time $k\cdot n^{o(1)}\cdot\log^{O(1)}L$.
\end{itemize}

Then Lemma~\ref{lem:helpermaintenance} holds.
\end{lemma}

A data-structure with the above requirements indeed exists by the recent work of Bernstein~\emph{et al.} \cite{BGS21_arxiv}.
Since this exact formulation was not explicitly stated in \cite{BGS21_arxiv}, let us point out how to obtain it. 
In Section III.5 of \cite{BGS21_arxiv}, there are $\log_{\gamma}nL$ \emph{distance scales} of the data structures where $\gamma \in (2, n^{o(1)})$. 
For each scale $i$, the objects in \Cref{lem:helpermaintenance}, such as the estimates $\tilde{d}^i(v)$ for all $v$, are maintained by the $\textsc{ApxBall}^\pi$ data structure described in Sections II.5 and III.3 of \cite{BGS21_arxiv}.  $\Htil^i$ is called an \emph{emulator} in Definition II.5.2 in \cite{BGS21_arxiv}. A vertex $c\in C^{i}$ corresponds to a \emph{core} (see Definition II.2.8 of \cite{BGS21_arxiv}). The tree $T^i_{hopset}$ is the \emph{Monotone Even-Shiloach} tree maintained by $\textsc{ApxBall}^\pi$ as described in Section II.5.2.

The only guarantee which was implicit in \cite{BGS21_arxiv} is the last bullet point about returning the first/last $k$ edges in the path $\pi_{u,v}$ in $k \cdot n^{o(1)}\cdot\log^{O(1)}L$ time. This guarantee follows by inspecting Algorithm 7 of \cite{BGS21_arxiv}. When the algorithm needs to report the whole path $\pi_{u,v}$, it does so by recursively querying the data structure starting from the largest distance scale to the smallest distance scale. More precisely, the path $\pi_{u,v}$ is represented by a hierarchical tree of depth $\log_{\gamma}nL$ whose leaves correspond to edges of $\pi_{u,v}$ and internal vertices correspond to subpaths of $\pi_{u,v}$. Each internal vertex has $n^{o(1)}$ children ordered according to the ordering of the edges in the path. 
Now, when we query only for the first $k$ edges in the path, the algorithm would resemble the situation when we have a balanced binary search tree and we want to list the first $k$ leaves of this tree. For our setting, the tree has $\log_{\gamma}nL$ depth and each vertex has $n^{o(1)}$ children. 
By recursively querying the multi-level data structures in this way, we will obtain the first $k$ edges of the path in $k\cdot n^{o(1)}\cdot\log^{O(1)}L$ time.

\begin{proof}
The proof follows along the lines of the proof of Proposition~\ref{prop:simpler} but with a few modifications.
As before, we apply the data-structure on the graph $G'$ that is obtained from $G$ by adding a super-source $s^*$ that represents all terminals in $T$. Rather than adding zero-length edges between each $t\in U$ and $s^*$ we contract all terminals into a single vertex and call it $s^*$.
Note that each edge $(s^*,x)$ in $G'$ can be traced back to an edge $(t,x)$ for some $t \in U$.
Let $T^1_{hopset},\ldots,T^{a}_{hopset}$ and $\tilde{d}^1(s^*,v),\ldots,\tilde{d}^{a}(s^*,v)$ be the trees and estimates maintained by the data-structure, where $a=\log_{\gamma}nL=n^{o(1)}\log^{O(1)}L$.

\paragraph{Maintaining a helper graph}
The helper graph $H=(V_H,E_H,w_H:E_H \to [2nL]\cup \{0\})$ is defined based on $G$, the trees $T^i_{hopset}$ and estimates $\tilde{d}^i(s^*,v)$ as follows.
For each vertex $x \in V \setminus U$ in $G$ we create $a$ copies $ \{ x^1,\ldots,x^{a}\}$ in $V_H$.
A terminal vertex $t \in U \subseteq V$ however only has one copy $t \in V_H$.
Moreover, we add all auxiliary vertices in $C^i$ for all $i$ to $V_H$.
There are two kinds of edges in $H$:

\begin{itemize}
    \item \emph{Green edges}: for each edge $(x,y) \in E_G$ we add $a^2$ edges $(x^i,y^j)$ to $E_H$, for all $i,j \in [a]$, and set their length to
    $$\ell_H(x^i,y^j) := \tilde{d}^i(s^*,x) + \ell_G(x,y) + \tilde{d}^j(s^*,y).$$ 
           A green edge gets updated whenever: (1) the length $\ell_G(x,y)$ gets updated and is more than $(1+\eps)$ times the value used in the current $\ell_H(x^i,y^j)$, or (2) the data-structure reports that $\tilde{d}^i(s^*,x)$ or $\tilde{d}^j(s^*,y)$ increased by a $\geq (1+\eps)$ factor.

    \item \emph{Red edges}: for each edge $(x,y) \in E(T^i_{hopset})$ for some $i \in [a]$ we add an edge $(x^i,y^i)$ of length $0$ to $E_H$. If one of $x$ or $y$ is $s^*$, suppose \emph{w.l.o.g.} that it is $x$, then we trace the edge $(s^*,y)$ back to an edge $(t,x) \in E$ and we add an edge $(t,y^i)$ of length $0$ to $E_H$.
    A red edge gets updated whenever the tree $T^i_{hopset}$ changes. 
    If an edge is removed from the tree, the red edge is removed (i.e. the length is increased from $0$ to infinity),\footnote{As before, these are the only incremental update that we do.} and if a new edge is added to the tree we add an edge of length zero to $H$.
\end{itemize}

(See Figure~\ref{Figs:Flowers} for an illustration.)

\paragraph{The recourse of the helper graph}
The $m\cdot n^{o(1)}\cdot \log^{O(1)}{L}$ upper bound on the total running time and recourse of maintaining $H$ follows exactly as in the simpler construction above, with the exception that there is an additional $a^2=n^{o(1)}\log^{O(1)}L$ factor due to the number of times an edge is copied.

\paragraph{The properties of the helper graph}
Next, let us prove that $H$ satisfies the two properties in Lemma~\ref{lem:helperproperties}.

For the first property, consider any vertex $x^i \in V_H$, for some $x \in V \setminus U$, and our goal is to prove that there is exactly one terminal $t_x \in U$ such that $d_H(x^i,t_x)=0$.
First, observe that the only red edges that are adjacent to a vertex $x^i$ come from edges in the tree $T^i_{hopset}$.
Let $e=(s^*,u)$ be the first edge on the path from $s^*$ to $x^i$ in $T^i_{hopset}$. This edge can be traced back to an edge $(t_x,u) \in E$ where $t_x \in U$. 
Due to the red edges, we will have $d_H(x^i,t_x)=0$.
On the other hand, since $T^i_{hopset}$ is a tree, there cannot be any other path from $x^i$ to $s^*$ and consequently there will not be any other terminal $t' \in U$ with a red path to $x^i$ (via edges from $T^i_{hopset})$.
The only other way that $t'$ could potentially reach $x^i$ with a red path is if $t'$ could reach $t_x$ via edges from some other $T^j_{hopset}$ and then use the red path from $t_x$ to $x^i$. 
However, by construction, any red path between two terminals implies a loop in one of the trees $T^j_{hopset}$ -- a contradiction.

To prove the second property, let $(x,y) \in E$ be an edge in $G$, and let $i,j \in [a]$ be such that $d_{G'}(s^*,x) \in [\gamma^{i-1},\gamma^i]$ and $d_{G'}(s^*,y) \in [\gamma^{j-1},\gamma^j]$.
Consequently, $\tilde{d}^i(s^*,x) \leq (1+\eps)\cdot d_{G'}(s^*,x)$ and $\tilde{d}^j(s^*,y) \leq (1+\eps)\cdot d_{G'}(s^*,y)$.
We will show that the green edge $(x^i,y^j) \in E_H$ satisfies the requirement.
Let $t_x$ and $t_y$ be such that $d_H(x^i,t_x)=0$ and $d_H(y^j,t_y)=0$.
By an argument similar to Claim~\ref{cl:closest_terminal} from before, we know that $t_x$ and $t_y$ are $(1+\eps)$-closest terminals to $x$ and $y$ (respectively) in $G$.
This is because $d_G(t_x,x) \leq d_{T^i_{hopset}}(t_x,x) = \tilde{d}^i(s^*,x) \leq (1+\eps)\cdot d_{G'}(s^*,x) = (1+\eps)\cdot \min_{t \in U} d_G(t,x)$, and similarly $d_G(t_y,y) \leq \tilde{d}^j(s^*,y) \leq (1+\eps)\cdot \min_{t \in U} d_G(t,y)$.
Finally, we can upper bound the length of the green edge $(x^i,y^j)$ in $H$ by 
$$\ell_H(x^i,y^j) = \tilde{d}^i(s^*,x) + \ell_G(x,y) + \tilde{d}^j(s^*,y) \leq (1+\eps)\cdot (d_{G'}(s^*,x) + \ell_G(x,y) + d_{G'}(s^*,y))$$ 
$$\leq (1+\eps)\cdot(d_G(t_x,x)+\ell_G(x,y)+d_G(y,t_y)).$$

\paragraph{Expanding an MST of $H$ into a $U$-Steiner subgraph of $G$}
To conclude the proof we point out that the process of expanding an MST of $H$ into a Steiner Subgraph of $G$ can be carried out in exactly the same way as in the previous subsection, with the minor detail that we have to get rid of the auxiliary vertices $C^i$ if they are in the MST.
By assumption, there are no edges between them, so we can simply make the expansion on two-paths in the MST (that skip over auxiliary vertices) rather than edges.
This expansion of an edge or two-path of length $w$ into a path $\pi_{u,x}$ in $G$ of length up to $(1+\eps)\cdot w$ introduces a $(1+\eps)$ factor to the total length of the Steiner subgraph (that we did not have in the simple construction).
\end{proof}

\section{Conclusion}\label{sec:conclusion}

In this paper, we broke the longstanding cubic barrier for the \ghtree and \apmf problems in general, weighted graphs. All previous improvements since 1961, were either corollaries of speed-ups in (single-pair) max-flow algorithms, or were limited to special graph classes.
Assuming the \apsp Conjecture, a cornerstone of fine-grained complexity, our result disproves the belief that computing max-flows is at least as hard as computing shortest-paths.

Our algorithm has a running time of $\tO(n^2)$ which is nearly-optimal for \apmf if all ${n \choose 2}$ max-flow values must be returned. 
For \ghtree in unweighted graphs, our techniques yield an improved running time of $m^{1+o(1)}$. In fact, a succinct representation of all-pairs max-flows can be produced in the time of $\tilde{O}(1)$ calls to a (single-pair) max-flow algorithm.

It is an interesting open question as to whether our techniques can be extended to obtain an $m^{1+o(1)}$-time \ghtree algorithm for weighted graphs as well. In particular, the most significant challenge is to design a new dynamic Steiner tree subroutine (analogous to the one in \Cref{sec:unweighted}) that can be used for packing Steiner trees in weighted graphs in almost-linear time. 

\appendix

\bibliographystyle{alphaurlinit}
\bibliography{reference}

\newcommand{\etalchar}[1]{$^{#1}$}
\begin{thebibliography}{CMWW19}

\bibitem[ACK20]{AbboudCK20}
A.~Abboud, V.~Cohen{-}Addad, and P.~N. Klein.
\newblock New hardness results for planar graph problems in p and an algorithm
  for sparsest cut.
\newblock In K.~Makarychev, Y.~Makarychev, M.~Tulsiani, G.~Kamath, and
  J.~Chuzhoy, editors, {\em Proccedings of the 52nd Annual {ACM} {SIGACT}
  Symposium on Theory of Computing, {STOC} 2020, Chicago, IL, USA, June 22-26,
  2020}, pages 996--1009. {ACM}, 2020.
\newblock Available from: \url{https://doi.org/10.1145/3357713.3384310}, \href
  {http://dx.doi.org/10.1145/3357713.3384310}
  {\path{doi:10.1145/3357713.3384310}}.

\bibitem[ACZ98]{ACZ98}
S.~R. Arikati, S.~Chaudhuri, and C.~D. Zaroliagis.
\newblock All-pairs min-cut in sparse networks.
\newblock {\em J. Algorithms}, 29(1):82--110, 1998.
\newblock \href {http://dx.doi.org/10.1006/jagm.1998.0961}
  {\path{doi:10.1006/jagm.1998.0961}}.

\bibitem[AD16]{AbboudD16}
A.~Abboud and S.~Dahlgaard.
\newblock Popular conjectures as a barrier for dynamic planar graph algorithms.
\newblock In I.~Dinur, editor, {\em {IEEE} 57th Annual Symposium on Foundations
  of Computer Science, {FOCS} 2016, 9-11 October 2016, Hyatt Regency, New
  Brunswick, New Jersey, {USA}}, pages 477--486. {IEEE} Computer Society, 2016.
\newblock Available from: \url{https://doi.org/10.1109/FOCS.2016.58}, \href
  {http://dx.doi.org/10.1109/FOCS.2016.58} {\path{doi:10.1109/FOCS.2016.58}}.

\bibitem[AGI{\etalchar{+}}19]{A+18}
A.~Abboud, L.~Georgiadis, G.~F. Italiano, R.~Krauthgamer, N.~Parotsidis,
  O.~Trabelsi, P.~Uznanski, and D.~Wolleb-Graf.
\newblock {Faster Algorithms for All-Pairs Bounded Min-Cuts}.
\newblock In {\em 46th International Colloquium on Automata, Languages, and
  Programming (ICALP 2019)}, volume 132, pages 7:1--7:15, 2019.
\newblock \href {http://dx.doi.org/10.4230/LIPIcs.ICALP.2019.7}
  {\path{doi:10.4230/LIPIcs.ICALP.2019.7}}.

\bibitem[AGW15]{AbboudGW15}
A.~Abboud, F.~Grandoni, and V.~V. Williams.
\newblock Subcubic equivalences between graph centrality problems, {APSP} and
  diameter.
\newblock In P.~Indyk, editor, {\em Proceedings of the Twenty-Sixth Annual
  {ACM-SIAM} Symposium on Discrete Algorithms, {SODA} 2015, San Diego, CA, USA,
  January 4-6, 2015}, pages 1681--1697. {SIAM}, 2015.
\newblock Available from: \url{https://doi.org/10.1137/1.9781611973730.112},
  \href {http://dx.doi.org/10.1137/1.9781611973730.112}
  {\path{doi:10.1137/1.9781611973730.112}}.

\bibitem[AHK12]{AroraHK12}
S.~Arora, E.~Hazan, and S.~Kale.
\newblock The multiplicative weights update method: a meta-algorithm and
  applications.
\newblock {\em Theory of Computing}, 8(1):121--164, 2012.

\bibitem[AIS{\etalchar{+}}16]{akiba2016cut}
T.~Akiba, Y.~Iwata, Y.~Sameshima, N.~Mizuno, and Y.~Yano.
\newblock Cut tree construction from massive graphs.
\newblock In {\em 2016 IEEE 16th International Conference on Data Mining
  (ICDM)}, pages 775--780. IEEE, 2016.

\bibitem[AKT20]{AKT20_b}
A.~Abboud, R.~Krauthgamer, and O.~Trabelsi.
\newblock Cut-equivalent trees are optimal for min-cut queries.
\newblock In {\em 61st {IEEE} Annual Symposium on Foundations of Computer
  Science, {FOCS} 2020}, pages 105--118, 2020.
\newblock \href {http://dx.doi.org/10.1109/FOCS46700.2020.00019}
  {\path{doi:10.1109/FOCS46700.2020.00019}}.

\bibitem[AKT21a]{AKT21_focs}
A.~Abboud, R.~Krauthgamer, and O.~Trabelsi.
\newblock {APMF} $<$ {APSP}? {Gomory-Hu} tree for unweighted graphs in
  almost-quadratic time.
\newblock {\em Accepted to FOCS'21}, 2021.
\newblock arXiv:2106.02981.

\bibitem[AKT21b]{AKT20}
A.~Abboud, R.~Krauthgamer, and O.~Trabelsi.
\newblock New algorithms and lower bounds for all-pairs max-flow in undirected
  graphs.
\newblock {\em Theory of Computing}, 17(5):1--27, 2021.
\newblock \href {http://dx.doi.org/10.4086/toc.2021.v017a005}
  {\path{doi:10.4086/toc.2021.v017a005}}.

\bibitem[AKT21c]{AKT21_stoc}
A.~Abboud, R.~Krauthgamer, and O.~Trabelsi.
\newblock Subcubic algorithms for gomory--hu tree in unweighted graphs.
\newblock In {\em Proceedings of the 53rd Annual ACM SIGACT Symposium on Theory
  of Computing}, pages 1725--1737, 2021.
\newblock \href {http://dx.doi.org/10.1145/3406325.3451073}
  {\path{doi:10.1145/3406325.3451073}}.

\bibitem[AKT22]{AKT22_soda}
A.~Abboud, R.~Krauthgamer, and O.~Trabelsi.
\newblock Friendly cut sparsifiers and faster {Gomory-Hu} trees.
\newblock {\em Accepted to SODA'22}, 2022.
\newblock arXiv:2110.15891.

\bibitem[AMO93]{AhujaMO93}
R.~Ahuja, T.~Magnanti, and J.~Orlin.
\newblock {\em Network Flows}.
\newblock Prentice Hall, 1993.

\bibitem[AV20]{AnariV20}
N.~Anari and V.~V. Vazirani.
\newblock Planar graph perfect matching is in {NC}.
\newblock {\em Journal of the ACM}, 67(4):1--34, 2020.

\bibitem[AVY15]{AVY15}
A.~Abboud, V.~{Vassilevska Williams}, and H.~Yu.
\newblock Matching triangles and basing hardness on an extremely popular
  conjecture.
\newblock In {\em Proc.\ of 47th STOC}, pages 41--50, 2015.

\bibitem[AW14]{AbboudW14}
A.~Abboud and V.~V. Williams.
\newblock Popular conjectures imply strong lower bounds for dynamic problems.
\newblock In {\em 55th {IEEE} Annual Symposium on Foundations of Computer
  Science, {FOCS} 2014}, pages 434--443, 2014.
\newblock \href {http://dx.doi.org/10.1109/FOCS.2014.53}
  {\path{doi:10.1109/FOCS.2014.53}}.

\bibitem[AW21]{AlmanW20}
J.~Alman and V.~V. Williams.
\newblock A refined laser method and faster matrix multiplication.
\newblock In {\em Proceedings of the 2021 {ACM-SIAM} Symposium on Discrete
  Algorithms, {SODA} 2021}, pages 522--539, 2021.
\newblock \href {http://dx.doi.org/10.1137/1.9781611976465.32}
  {\path{doi:10.1137/1.9781611976465.32}}.

\bibitem[BBDF06]{barth2006revisiting}
D.~Barth, P.~Berthom{\'e}, M.~Diallo, and A.~Ferreira.
\newblock Revisiting parametric multi-terminal problems: Maximum flows, minimum
  cuts and cut-tree computations.
\newblock {\em Discrete Optimization}, 3(3):195--205, 2006.

\bibitem[BCH{\etalchar{+}}08]{BCHKP08}
A.~Bhalgat, R.~Cole, R.~Hariharan, T.~Kavitha, and D.~Panigrahi.
\newblock Efficient algorithms for {S}teiner edge connectivity computationand
  {G}omory-{H}u tree construction for unweighted graphs.
\newblock Unpublished full version of \cite{BHKP07}, 2008.
\newblock Available from: \url{http://hariharan-ramesh.com/papers/gohu.pdf}.

\bibitem[BDT16]{BackursDT16}
A.~Backurs, N.~Dikkala, and C.~Tzamos.
\newblock Tight hardness results for maximum weight rectangles.
\newblock In I.~Chatzigiannakis, M.~Mitzenmacher, Y.~Rabani, and D.~Sangiorgi,
  editors, {\em 43rd International Colloquium on Automata, Languages, and
  Programming, {ICALP} 2016, July 11-15, 2016, Rome, Italy}, volume~55 of {\em
  LIPIcs}, pages 81:1--81:13. Schloss Dagstuhl - Leibniz-Zentrum f{\"{u}}r
  Informatik, 2016.
\newblock Available from: \url{https://doi.org/10.4230/LIPIcs.ICALP.2016.81},
  \href {http://dx.doi.org/10.4230/LIPIcs.ICALP.2016.81}
  {\path{doi:10.4230/LIPIcs.ICALP.2016.81}}.

\bibitem[BENW16]{BENW16}
G.~Borradaile, D.~Eppstein, A.~Nayyeri, and C.~{Wulff-Nilsen}.
\newblock All-pairs minimum cuts in near-linear time for surface-embedded
  graphs.
\newblock In {\em 32nd International Symposium on Computational Geometry},
  volume~51 of {\em SoCG '16}, pages 22:1--22:16, 2016.
\newblock \href {http://dx.doi.org/10.4230/LIPIcs.SoCG.2016.22}
  {\path{doi:10.4230/LIPIcs.SoCG.2016.22}}.

\bibitem[BFJ95]{BFJ95}
J.~Bang{-}Jensen, A.~Frank, and B.~Jackson.
\newblock Preserving and increasing local edge-connectivity in mixed graphs.
\newblock {\em {SIAM} J. Discret. Math.}, 8(2):155--178, 1995.
\newblock \href {http://dx.doi.org/10.1137/S0036142993226983}
  {\path{doi:10.1137/S0036142993226983}}.

\bibitem[BGK20]{BGK20esa}
S.~Baswana, S.~Gupta, and T.~Knollmann.
\newblock Mincut sensitivity data structures for the insertion of an edge.
\newblock In {\em 28th Annual European Symposium on Algorithms (ESA 2020)},
  2020.

\bibitem[BGMW20]{bringmann2020tree}
K.~Bringmann, P.~Gawrychowski, S.~Mozes, and O.~Weimann.
\newblock Tree edit distance cannot be computed in strongly subcubic time
  (unless apsp can).
\newblock {\em ACM Transactions on Algorithms (TALG)}, 16(4):1--22, 2020.

\bibitem[BGS22]{BGS21_arxiv}
A.~Bernstein, M.~P. Gutenberg, and T.~Saranurak.
\newblock Deterministic decremental sssp and approximate min-cost flow in
  almost-linear time.
\newblock In {\em 2021 IEEE 62nd Annual Symposium on Foundations of Computer
  Science (FOCS)}, pages 1000--1008. IEEE, 2022.

\bibitem[BHKP07]{BHKP07}
A.~Bhalgat, R.~Hariharan, T.~Kavitha, and D.~Panigrahi.
\newblock An {$\tilde{O}(mn)$} {G}omory-{H}u tree construction algorithm for
  unweighted graphs.
\newblock In {\em 39th Annual ACM Symposium on Theory of Computing}, STOC'07,
  pages 605--614, 2007.
\newblock \href {http://dx.doi.org/10.1145/1250790.1250879}
  {\path{doi:10.1145/1250790.1250879}}.

\bibitem[BK15]{BeK15}
A.~A. Bencz{\'{u}}r and D.~R. Karger.
\newblock Randomized approximation schemes for cuts and flows in capacitated
  graphs.
\newblock {\em {SIAM} J. Comput.}, 44(2):290--319, 2015.
\newblock \href {http://dx.doi.org/10.1137/070705970}
  {\path{doi:10.1137/070705970}}.

\bibitem[BSW15]{BSW15}
G.~Borradaile, P.~Sankowski, and C.~{Wulff-Nilsen}.
\newblock Min $st$-cut oracle for planar graphs with near-linear preprocessing
  time.
\newblock {\em ACM Trans. Algorithms}, 11(3), 2015.
\newblock \href {http://dx.doi.org/10.1145/2684068}
  {\path{doi:10.1145/2684068}}.

\bibitem[BT17]{backurs2017improving}
A.~Backurs and C.~Tzamos.
\newblock Improving viterbi is hard: Better runtimes imply faster clique
  algorithms.
\newblock In {\em International Conference on Machine Learning}, pages
  311--321. PMLR, 2017.

\bibitem[CCPS97]{CookCPS97}
W.~Cook, W.~Cunningham, W.~Pulleybank, and A.~Schrijver.
\newblock {\em Combinatorial Optimization}.
\newblock Wiley, 1997.

\bibitem[CKL{\etalchar{+}}22]{ChenKLPGS22}
L.~Chen, R.~Kyng, Y.~P. Liu, R.~Peng, M.~P. Gutenberg, and S.~Sachdeva.
\newblock Maximum flow and minimum-cost flow in almost-linear time.
\newblock {\em CoRR}, abs/2203.00671, 2022.
\newblock Available from: \url{https://doi.org/10.48550/arXiv.2203.00671},
  \href {http://arxiv.org/abs/2203.00671} {\path{arXiv:2203.00671}}, \href
  {http://dx.doi.org/10.48550/arXiv.2203.00671}
  {\path{doi:10.48550/arXiv.2203.00671}}.

\bibitem[CLP22]{CenLP22}
R.~Cen, J.~Li, and D.~Panigrahi.
\newblock Augmenting edge connectivity via isolating cuts.
\newblock In {\em Proceedings of the 2021 {ACM-SIAM} Symposium on Discrete
  Algorithms, {SODA} 2022}, 2022.

\bibitem[CMWW19]{cygan2019problems}
M.~Cygan, M.~Mucha, K.~Wegrzycki, and M.~W{\l}odarczyk.
\newblock On problems equivalent to (min,+)-convolution.
\newblock {\em ACM Transactions on Algorithms (TALG)}, 15(1):1--25, 2019.

\bibitem[CQ21]{CQ21}
C.~Chekuri and K.~Quanrud.
\newblock Isolating cuts,(bi-) submodularity, and faster algorithms for
  connectivity.
\newblock In {\em 48th International Colloquium on Automata, Languages, and
  Programming (ICALP 2021)}. Schloss Dagstuhl-Leibniz-Zentrum f{\"u}r
  Informatik, 2021.

\bibitem[CS21]{ChuzhoyS20}
J.~Chuzhoy and T.~Saranurak.
\newblock Deterministic algorithms for decremental shortest paths via layered
  core decomposition.
\newblock In {\em Proceedings of the 2021 ACM-SIAM Symposium on Discrete
  Algorithms (SODA)}, pages 2478--2496. SIAM, 2021.

\bibitem[Dij59]{dijkstra1959note}
E.~W. Dijkstra.
\newblock A note on two problems in connexion with graphs.
\newblock {\em Numerische mathematik}, 1(1):269--271, 1959.

\bibitem[Elm64]{elmaghraby1964sensitivity}
S.~E. Elmaghraby.
\newblock Sensitivity analysis of multiterminal flow networks.
\newblock {\em Operations Research}, 12(5):680--688, 1964.

\bibitem[Fle00]{Fleischer00}
L.~K. Fleischer.
\newblock Approximating fractional multicommodity flow independent of the
  number of commodities.
\newblock {\em SIAM Journal on Discrete Mathematics}, 13(4):505--520, 2000.

\bibitem[GH61]{GH61}
R.~E. Gomory and T.~C. Hu.
\newblock Multi-terminal network flows.
\newblock {\em Journal of the Society for Industrial and Applied Mathematics},
  9:551--570, 1961.
\newblock Available from: \url{http://www.jstor.org/stable/2098881}.

\bibitem[GK07]{GargK07}
N.~Garg and J.~K{\"o}nemann.
\newblock Faster and simpler algorithms for multicommodity flow and other
  fractional packing problems.
\newblock {\em SIAM Journal on Computing}, 37(2):630--652, 2007.

\bibitem[GMW21]{GawrychowskiMW21}
P.~Gawrychowski, S.~Mozes, and O.~Weimann.
\newblock Planar negative \emph{k}-cycle.
\newblock In D.~Marx, editor, {\em Proceedings of the 2021 {ACM-SIAM} Symposium
  on Discrete Algorithms, {SODA} 2021, Virtual Conference, January 10 - 13,
  2021}, pages 2717--2724. {SIAM}, 2021.
\newblock Available from: \url{https://doi.org/10.1137/1.9781611976465.161},
  \href {http://dx.doi.org/10.1137/1.9781611976465.161}
  {\path{doi:10.1137/1.9781611976465.161}}.

\bibitem[GR98]{GR98}
A.~V. Goldberg and S.~Rao.
\newblock Beyond the flow decomposition barrier.
\newblock {\em J. ACM}, 45(5):783–797, 1998.
\newblock \href {http://dx.doi.org/10.1145/290179.290181}
  {\path{doi:10.1145/290179.290181}}.

\bibitem[GT01]{GT01}
A.~V. Goldberg and K.~Tsioutsiouliklis.
\newblock Cut tree algorithms: an experimental study.
\newblock {\em Journal of Algorithms}, 38(1):51--83, 2001.

\bibitem[Gus90]{Gusfield90}
D.~Gusfield.
\newblock Very simple methods for all pairs network flow analysis.
\newblock {\em SIAM Journal on Computing}, 19(1):143--155, 1990.

\bibitem[HdLT01]{holm2001poly}
J.~Holm, K.~de~Lichtenberg, and M.~Thorup.
\newblock Poly-logarithmic deterministic fully-dynamic algorithms for
  connectivity, minimum spanning tree, 2-edge, and biconnectivity.
\newblock {\em Journal of the ACM (JACM)}, 48(4):723--760, 2001.
\newblock \href {http://dx.doi.org/10.1145/502090.502095}
  {\path{doi:10.1145/502090.502095}}.

\bibitem[HKN14]{henzinger2014decremental}
M.~Henzinger, S.~Krinninger, and D.~Nanongkai.
\newblock Decremental single-source shortest paths on undirected graphs in
  near-linear total update time.
\newblock In {\em 2014 IEEE 55th Annual Symposium on Foundations of Computer
  Science}, pages 146--155. IEEE, 2014.

\bibitem[HL07]{HL07}
R.~Hassin and A.~Levin.
\newblock Flow trees for vertex-capacitated networks.
\newblock {\em Discrete Appl. Math.}, 155(4):572--578, 2007.
\newblock \href {http://dx.doi.org/10.1016/j.dam.2006.08.012}
  {\path{doi:10.1016/j.dam.2006.08.012}}.

\bibitem[HR92]{HR92}
F.~K. Hwang and D.~S. Richards.
\newblock Steiner tree problems.
\newblock {\em Networks}, 22(1):55--89, 1992.
\newblock \href {http://dx.doi.org/10.1002/net.3230220105}
  {\path{doi:10.1002/net.3230220105}}.

\bibitem[Hu74]{Hu74}
T.~C. Hu.
\newblock Optimum communication spanning trees.
\newblock {\em SIAM Journal on Computing}, 3(3):188--195, 1974.
\newblock \href {http://dx.doi.org/10.1137/0203015}
  {\path{doi:10.1137/0203015}}.

\bibitem[HW13]{hartmann2013dynamic}
T.~Hartmann and D.~Wagner.
\newblock Dynamic {Gomory--Hu} tree construction--fast and simple.
\newblock {\em arXiv preprint arXiv:1310.0178}, 2013.

\bibitem[Jel63]{Jel63}
F.~Jelinek.
\newblock On the maximum number of different entries in the terminal capacity
  matrix of oriented communication nets.
\newblock {\em IEEE Transactions on Circuit Theory}, 10(2):307--308, 1963.
\newblock \href {http://dx.doi.org/10.1109/TCT.1963.1082149}
  {\path{doi:10.1109/TCT.1963.1082149}}.

\bibitem[Kar00]{Karger00}
D.~R. Karger.
\newblock Minimum cuts in near-linear time.
\newblock {\em Journal of the ACM}, 47(1):46--76, 2000.

\bibitem[KL15]{KL15}
D.~R. Karger and M.~S. Levine.
\newblock Fast augmenting paths by random sampling from residual graphs.
\newblock {\em {SIAM} J. Comput.}, 44(2):320--339, 2015.
\newblock \href {http://dx.doi.org/10.1137/070705994}
  {\path{doi:10.1137/070705994}}.

\bibitem[KMB81]{KMB81}
L.~T. Kou, G.~Markowsky, and L.~Berman.
\newblock A fast algorithm for steiner trees.
\newblock {\em Acta Informatica}, 15:141--145, 1981.
\newblock Available from: \url{https://doi.org/10.1007/BF00288961}, \href
  {http://dx.doi.org/10.1007/BF00288961} {\path{doi:10.1007/BF00288961}}.

\bibitem[KT18]{KT18}
R.~Krauthgamer and O.~Trabelsi.
\newblock Conditional lower bounds for all-pairs max-flow.
\newblock {\em {ACM} Trans. Algorithms}, 14(4):42:1--42:15, 2018.
\newblock \href {http://dx.doi.org/10.1145/3212510}
  {\path{doi:10.1145/3212510}}.

\bibitem[KV12]{korte2012combinatorial}
B.~Korte and J.~Vygen.
\newblock {\em Combinatorial optimization}, volume~2.
\newblock Springer, 2012.

\bibitem[Li21]{Li21thesis}
J.~Li.
\newblock {\em Preconditioning and Locality in Algorithm Design}.
\newblock PhD thesis, Carnegie Mellon University, 2021.

\bibitem[LNP{\etalchar{+}}21]{LNPSS21}
J.~Li, D.~Nanongkai, D.~Panigrahi, T.~Saranurak, and S.~Yingchareonthawornchai.
\newblock Vertex connectivity in poly-logarithmic max-flows.
\newblock In {\em Proceedings of the 53rd Annual ACM SIGACT Symposium on Theory
  of Computing}, pages 317--329, 2021.

\bibitem[LP20]{LP20}
J.~Li and D.~Panigrahi.
\newblock Deterministic min-cut in poly-logarithmic max-flows.
\newblock In {\em 61st {IEEE} Annual Symposium on Foundations of Computer
  Science, {FOCS} 2020}, pages 85--92, 2020.
\newblock \href {http://dx.doi.org/10.1109/FOCS46700.2020.00017}
  {\path{doi:10.1109/FOCS46700.2020.00017}}.

\bibitem[LP21]{LP21}
J.~Li and D.~Panigrahi.
\newblock Approximate {G}omory-{H}u tree is faster than $n-1$ max-flows.
\newblock In {\em {STOC} '21: 53rd Annual {ACM} {SIGACT} Symposium on Theory of
  Computing}, pages 1738--1748. {ACM}, 2021.
\newblock \href {http://dx.doi.org/10.1145/3406325.3451112}
  {\path{doi:10.1145/3406325.3451112}}.

\bibitem[LPS21]{LPS21}
J.~Li, D.~Panigrahi, and T.~Saranurak.
\newblock A nearly optimal all-pairs min-cuts algorithm in simple graphs.
\newblock {\em Accepted to FOCS'21}, 2021.
\newblock arXiv:2106.02233.

\bibitem[Mad78]{Mader78}
W.~Mader.
\newblock A reduction method for edge-connectivity in graphs.
\newblock In B.~Bollobás, editor, {\em Advances in Graph Theory}, volume~3 of
  {\em Annals of Discrete Mathematics}, pages 145--164. Elsevier, 1978.
\newblock \href
  {http://dx.doi.org/https://doi.org/10.1016/S0167-5060(08)70504-1}
  {\path{doi:https://doi.org/10.1016/S0167-5060(08)70504-1}}.

\bibitem[May62]{May62}
W.~Mayeda.
\newblock On oriented communication nets.
\newblock {\em IRE Transactions on Circuit Theory}, 9(3):261--267, 1962.
\newblock \href {http://dx.doi.org/10.1109/TCT.1962.1086912}
  {\path{doi:10.1109/TCT.1962.1086912}}.

\bibitem[Meh88]{Mehlhorn88}
K.~Mehlhorn.
\newblock A faster approximation algorithm for the steiner problem in graphs.
\newblock {\em Information Processing Letters}, 27(3):125--128, 1988.

\bibitem[MN21]{MN21}
S.~Mukhopadhyay and D.~Nanongkai.
\newblock A note on isolating cut lemma for submodular function minimization.
\newblock {\em arXiv preprint arXiv:2103.15724}, 2021.

\bibitem[NS18]{Naves18}
G.~Naves and F.~B. Shepherd.
\newblock When do {G}omory-{H}u subtrees exist?
\newblock {\em CoRR}, 2018.
\newblock Available from: \url{http://arxiv.org/abs/1807.07331}.

\bibitem[NW61]{NW61}
C.~S.~A. Nash-Williams.
\newblock {Edge-Disjoint Spanning Trees of Finite Graphs}.
\newblock {\em Journal of the London Mathematical Society}, s1-36(1):445--450,
  01 1961.
\newblock \href {http://dx.doi.org/10.1112/jlms/s1-36.1.445}
  {\path{doi:10.1112/jlms/s1-36.1.445}}.

\bibitem[Pan16]{Panigrahi16}
D.~Panigrahi.
\newblock {G}omory-{H}u trees.
\newblock In M.-Y. Kao, editor, {\em Encyclopedia of Algorithms}, pages
  858--861. Springer New York, 2016.
\newblock \href {http://dx.doi.org/10.1007/978-1-4939-2864-4_168}
  {\path{doi:10.1007/978-1-4939-2864-4_168}}.

\bibitem[Ple81]{Ple81}
J.~Plesn{\'{\i}}k.
\newblock The complexity of designing a network with minimum diameter.
\newblock {\em Networks}, 11(1):77--85, 1981.
\newblock Available from: \url{https://doi.org/10.1002/net.3230110110}, \href
  {http://dx.doi.org/10.1002/net.3230110110}
  {\path{doi:10.1002/net.3230110110}}.

\bibitem[PQ80]{picard1980structure}
J.-C. Picard and M.~Queyranne.
\newblock On the structure of all minimum cuts in a network and applications.
\newblock In {\em Combinatorial Optimization II}, pages 8--16. Springer, 1980.

\bibitem[PR82]{PR82}
M.~W. Padberg and M.~R. Rao.
\newblock Odd minimum cut-sets and $b$-matchings.
\newblock {\em Mathematics of Operations Research}, 7(1):67--80, 1982.

\bibitem[RZ04]{RZ04}
L.~Roditty and U.~Zwick.
\newblock On dynamic shortest paths problems.
\newblock In {\em European Symposium on Algorithms}, pages 580--591. Springer,
  2004.

\bibitem[Sah15]{Saha15}
B.~Saha.
\newblock Language edit distance and maximum likelihood parsing of stochastic
  grammars: Faster algorithms and connection to fundamental graph problems.
\newblock In V.~Guruswami, editor, {\em {IEEE} 56th Annual Symposium on
  Foundations of Computer Science, {FOCS} 2015, Berkeley, CA, USA, 17-20
  October, 2015}, pages 118--135. {IEEE} Computer Society, 2015.
\newblock Available from: \url{https://doi.org/10.1109/FOCS.2015.17}, \href
  {http://dx.doi.org/10.1109/FOCS.2015.17} {\path{doi:10.1109/FOCS.2015.17}}.

\bibitem[Sch03]{Schrijver03}
A.~Schrijver.
\newblock {\em Combinatorial Optimization}.
\newblock Springer, 2003.

\bibitem[Sei95]{seidel1995all}
R.~Seidel.
\newblock On the all-pairs-shortest-path problem in unweighted undirected
  graphs.
\newblock {\em Journal of computer and system sciences}, 51(3):400--403, 1995.

\bibitem[TM80]{Tak80}
H.~Takahashi and A.~Matsuyama.
\newblock An approximate solution for the {S}teiner problem in graphs.
\newblock {\em Math. Japon.}, 24(6):573--577, 1979/80.

\bibitem[Tut61]{Tut61}
W.~T. Tutte.
\newblock {On the Problem of Decomposing a Graph into n Connected Factors}.
\newblock {\em Journal of the London Mathematical Society}, s1-36(1):221--230,
  01 1961.
\newblock \href {http://dx.doi.org/10.1112/jlms/s1-36.1.221}
  {\path{doi:10.1112/jlms/s1-36.1.221}}.

\bibitem[WL93]{WL93}
Z.~Wu and R.~Leahy.
\newblock An optimal graph theoretic approach to data clustering: Theory and
  its application to image segmentation.
\newblock {\em IEEE transactions on pattern analysis and machine intelligence},
  15(11):1101--1113, 1993.

\bibitem[WW18]{VW18}
V.~V. Williams and R.~R. Williams.
\newblock Subcubic equivalences between path, matrix, and triangle problems.
\newblock {\em J. {ACM}}, 65(5):27:1--27:38, 2018.
\newblock \href {http://dx.doi.org/10.1145/3186893}
  {\path{doi:10.1145/3186893}}.

\bibitem[Zha21]{Zhang21b}
T.~Zhang.
\newblock {Gomory-Hu} trees in quadratic time.
\newblock {\em arXiv preprint arXiv:2112.01042}, 2021.

\bibitem[Zha22]{Zhang22}
T.~Zhang.
\newblock {Faster Cut-Equivalent Trees in Simple Graphs}.
\newblock In {\em 49th International Colloquium on Automata, Languages, and
  Programming (ICALP 2022)}, volume 229 of {\em Leibniz International
  Proceedings in Informatics (LIPIcs)}, pages 109:1--109:18, Dagstuhl, Germany,
  2022.
\newblock \href {http://arxiv.org/abs/2106.03305} {\path{arXiv:2106.03305}},
  \href {http://dx.doi.org/10.4230/LIPIcs.ICALP.2022.109}
  {\path{doi:10.4230/LIPIcs.ICALP.2022.109}}.

\end{thebibliography}

\section{Reduction to Single-Source Terminal Mincuts Given a Promise}\label{appendix:reduction}

In this section, we present the reduction of \Cref{lem:reduction2}, restated below.
\Reduction*

Before giving the formal reduction, let us give a high-level description.
The original Gomory-Hu algorithm is a reduction to $n-1$ max-flow calls.
At a high-level, it starts with all vertices in one super-vertex $V'=V$ 
and recursively bi-partitions $V'\subseteq V$, until $V'$ contains only one vertex.
In any recursive step, the input to the algorithm is the set $V'\subseteq V$ (called terminals)
and a graph $G'$ that is formed from contracting sets of vertices in the input graph 
$G$ in previous recursive steps. In the current step, the algorithm
picks an arbitrary pair of vertices $s,t \in V'$ and computes an $(s,t)$-mincut in $G'$.
Then, it creates two recursive subproblems, where in each subproblem,
one side of the $(s,t)$-mincut cut is retained uncontracted and the other side is contracted 
into a single vertex. The new $V'$ in a subproblem comprises the vertices in 
$V'$ on the uncontracted side.

The contractions serve to enforce consistency between the cuts, so that we end up with 
a tree structure of cuts.
But, recent work has shown that these contractions also help with efficiency. 
Suppose that instead of an arbitrary $s,t$ pair (which can lead to unbalanced cuts
and large recursive depth), we split using \emph{single-source mincuts} 
from a \emph{random pivot} $s\in V'$ to all other terminals $V'\setminus\{s\}$.
It turns out that with sufficiently high probability, the cuts from the random pivot 
to many other terminals only contain 
at most (say) $0.9$-fraction of the vertices in $V'$, 
and splitting according to all these cuts (instead of into two) 
leads to recursion whose depth is bounded by $O(\log n)$.

Effectively, this gives a reduction from \ghtree to the single-source mincuts problem, 
with only a multiplicative $\polylog(n)$ overhead in the running time.
The random-pivot idea was first used by Bhalgat {\em et al.}~\cite{BHKP07} but the 
first general reduction to single-source mincuts is by Abboud {\em et al.} \cite{AKT20_b}. 
We use a refined reduction that has two additional features:
(1) the algorithm (for the single-source problem) only needs to return the values of the 
cuts rather than the cuts themselves (and in fact, it is given estimates and only needs 
to decide if they are correct), 
and (2) the mincut values between any pair of terminals in the graph are guaranteed 
to be within a (say) $1.1$-factor from each other (the so called ``promise''). 
The first restriction only serves to simplify our new algorithm, while the second 
is more important for our new ideas to work. 
This is encapsulated in \Cref{lem:reduction2} which we prove below.

As discussed in the overview, we first remove the ``promise'' condition from the 
single source terminal mincuts problem (\Cref{problem:ssmc}). To be precise, we first formally 
define the single source terminal mincuts problem without the promise.

\begin{problem}[Single-Source Terminal Mincuts (without Promise)]\label{problem:ssmc-no-promise}
The input is a graph $G=(V,E,w)$, a terminal set $U \subseteq V$ and a source terminal $s\in U$. The goal is to determine the value of $\lambda(s,t)$ for each terminal $t\in U\setminus\{s\}$.
\end{problem}

We also require the approximate single-source mincuts algorithm of~\cite{LP21}, stated below, which we use as a black box.
\begin{theorem}[Theorem~1.7 of~\cite{LP21}]\label{lem:approx-ssmc}
Let $G=(V,E,w)$ be a graph, let $s\in V$, and let $\epsilon>0$ be a constant. There is an algorithm that outputs, for each vertex $v\in V\setminus\{s\}$, a $(1+\epsilon)$-approximation of $\lambda(s,v)$, and runs in $\tilde{O}(m\log\Delta)$ time plus $\textup{polylog}(n)\cdot\log\Delta$ calls to max-flow on $O(n)$-vertex, $O(m)$-edge graphs, where $\Delta$ is the ratio of maximum to minimum edge weights.
\end{theorem}

We now present the reduction that removes the promise.
\begin{lemma}
There is a randomized algorithm for the unconditional Single-Source Terminal Mincuts problem (\Cref{problem:ssmc-no-promise}) that makes calls to the promise version of the Single-Source Terminal Mincuts problem (\Cref{problem:ssmc}) on graphs with a total of $O(n\log(nW))$ vertices and $O(m\log(nW))$ edges. Outside of these calls, the algorithm takes $\tilde{O}(m\log W)$ time plus $\textup{polylog}(n)\cdot\log W$ calls to max-flow on $O(n)$-vertex, $O(m)$-edge graphs.
\end{lemma}
\begin{proof}
Call \Cref{lem:approx-ssmc} with $\epsilon=0.01$; note that $\Delta\le W$ since the input graph is integer-weighted with maximum weight $W$. Let $\tilde\lambda(s,v)$ be the computed approximations, which satisfy $\frac1{1+\epsilon}\lambda(s,v)\le\tilde\lambda(s,v)\le(1+\epsilon)\lambda(s,v)$. For each integer $i$, let $V_i\subseteq V\setminus\{s\}$ be all vertices $v\in V\setminus\{s\}$ with $(1+\epsilon)^i\le\tilde\lambda(s,v)<(1+\epsilon)^{i+1}$. Note that $V_i$ is nonempty for only $O(\log(nW))$ many integers $i$, since $\tilde\lambda(s,v)$ must be in the range $[\frac1{1+\epsilon},(1+\epsilon)mW]$.

We first claim that for each $i$, the terminal set $U_i=\{s\}\cup V_i$ satisfies the promise of \Cref{problem:ssmc}, namely that for all $t\in V_i$, we have $\lambda(U_i)\le\lambda(s,t)\le1.1\lambda(U_i)$. To prove this claim, note that for each $v\in V_i$, 
\[ (1+\epsilon)^{i-1} = \frac1{1+\epsilon}(1+\epsilon)^i \le \frac1{1+\epsilon}\tilde\lambda(s,v) \le \lambda(s,v) \le (1+\epsilon)\tilde\lambda(s,v) < (1+\epsilon)(1+\epsilon)^{i+1} = (1+\epsilon)^{i+2} ,\]
so all $\lambda(s,v)$ values are in the range $[(1+\epsilon)^{i-1},(1+\epsilon)^{i+3}]$. Moreover, we must have $\lambda(U_i)\ge(1+\epsilon)^{i-1}$ since the Steiner mincut of $U_i$ is an $(s,v)$-mincut for some $v\in U_i$. It follows that $\lambda(s,v)\le(1+\epsilon)^3\lambda(U_i)\le1.1\lambda(U_i)$ by the choice $\epsilon=0.01$, concluding the claim.

Since each terminal set $U_i$ satisfies the promise condition of \Cref{problem:ssmc}, we can run an algorithm that solves \Cref{problem:ssmc} to determine the values $\lambda(s,v)$ for all $v\in V_i$. Since the sets $V_i$ partition $V\setminus\{s\}$, we have correctly computed all values of $\lambda(s,v)$.
\end{proof}

For the remainder of this section, we reduce Gomory-Hu tree to the unconditional Single-Source Terminal Mincuts (\Cref{problem:ssmc-no-promise}). The reduction is virtually identical to Section~4.5 of~\cite{Li21thesis}, and we include it for the sake of completeness.

\begin{theorem}\label{thm:ghtree:reduction}
There is a randomized algorithm that outputs a Gomory-Hu tree of a weighted, undirected graph w.h.p. It makes calls to the unconditional Single-Source Terminal Mincuts problem (\Cref{problem:ssmc-no-promise}) on graphs with a total of $\tilde{O}(n)$ vertices and $\tilde{O}(m)$ edges. It also calls max-flow on graphs with a total of $\tilde{O}(n)$ vertices and $\tilde{O}(m)$ edges, and runs for $\tilde{O}(m)$ time outside of these calls. 
\end{theorem}

The precise Gomory-Hu tree algorithm that proves \Cref{thm:ghtree:reduction} is described in Algorithm~\ref{ghtree} a few pages down.

\subsection{Additional Preliminaries}

For this section, we need to formally define a generalization of Gomory-Hu trees called Gomory-Hu \emph{Steiner} trees that are more amenable to our recursive structure.
\begin{definition}[Gomory-Hu Steiner tree]
Given a graph $G=(V,E,w)$ and a set of terminals $U\subseteq V$, the Gomory-Hu Steiner tree is a weighted tree $T$ on the vertices $U$, together with a function $f:V\to U$, such that
 \begin{itemize}
 \item For all $s,t\in U$, consider the minimum-weight edge $(u,v)$ on the unique $s$--$t$ path in $T$. Let $U'$ be the vertices of the connected component of $T-(u,v)$ containing $s$.
Then, the set $f^{-1}(U')\subseteq V$ is an $(s,t)$-mincut, and its value is $w_T(u,v)$.
 \end{itemize}
\end{definition}

In our analysis, we use the notion of a {\em minimal} mincut and a {\em rooted minimal} Gomory-Hu tree. We define these next.

\begin{definition}[Minimal $(s,t)$-mincut]
A \emph{minimal} $(s,t)$-mincut is an $(s,t)$-mincut whose side $S\subseteq V$ containing $s$ is vertex-minimal.
\end{definition}

\begin{definition}[Rooted minimal Gomory-Hu Steiner tree]
Given a graph $G=(V,E,w)$ and a set of terminals $U\subseteq V$, a rooted minimal Gomory-Hu Steiner tree is a Gomory-Hu Steiner tree on $U$, rooted at some vertex $r\in U$, with the following additional property:
 \begin{itemize}
 \item[$(*)$] For all $t\in U\setminus\{r\}$, consider the minimum-weight edge $(u,v)$ on the unique $r-t$ path in $T$; if there are multiple minimum weight edges, let $(u, v)$ denote the one that is {\em closest to $t$}. Let $U'$ be the vertices of the connected component of $T-(u,v)$ containing $r$.
Then, $f^{-1}(U')\subseteq V$ is a \emph{minimal} $(r,t)$-mincut, and its value is $w_T(u,v)$.
 \end{itemize}
\end{definition}

The following theorem
establishes the existence of a rooted minimal Gomory-Hu Steiner tree rooted at any given vertex.

\begin{theorem}\label{thm:ghtree:rooted}
For any graph $G=(V,E,w)$, terminals $U\subseteq V$, and root $r\in U$, there exists a rooted minimal Gomory-Hu Steiner tree rooted at $r$.
\end{theorem}
\begin{proof}
Let $\epsilon>0$ be a small enough weight, and let $G'$ be the graph $G$ with an additional edge $(r,v)$ of weight $\epsilon$ added for each $v\in V\setminus\{r\}$. (If the edge $(r,v)$ already exists in $G$, then increase its weight by $\epsilon$ instead.) If $\epsilon>0$ is small enough, then for all $t\in V\setminus\{r\}$ and $S\subseteq V$, if $S$ is an $(r,t)$-mincut in $G'$, then $S$ is an $(r,t)$-mincut in $G$.

Let $(T',f)$ be a Gomory-Hu Steiner tree for $G'$. We claim that it is essentially a minimal Gomory-Hu Steiner tree for $G$, except that its edge weights need to be recomputed as mincuts in $G$ and not $G'$. More formally, let $T$ be the tree $T'$ with the following edge re-weighting: for each edge $(u,v)$ in $T$, take a connected component $U'$ of $T-(u,v)$ and reset the edge weight of $(u,v)$ to be $\delta_G(f^{-1}(U'))$ and not $\delta_{G'}(f^{-1}(U'))$. We now claim that $(T,f)$ is a minimal Steiner Gomory-Hu tree for $G$.

We first show that $(T,f)$ is a Gomory-Hu Steiner tree for $G$. Fix $s,t\in U$, let $(u,v)$ be the minimum-weight edge on the $s$--$t$ path in $T'$, and let $U'$ be the vertices of the connected component of $T'-(u,v)$ containing $s$. Since $(T',f)$ is a Gomory-Hu Steiner tree for $G'$, we have that $f^{-1}(U')$ is an $(s,t)$-mincut in $G'$. If $\epsilon>0$ is small enough, then by our argument from before, $f^{-1}(U')$ is also an $(s,t)$-mincut in $G$. By our edge re-weighting of $T$, the edge $(u,v)$ has the correct weight. Moreover, $(u,v)$ is the minimum-weight edge on the $s$--$t$ path in $T$, since a smaller weight edge would contradict the fact that $f^{-1}(U')$ is an $(s,t)$-mincut.

We now show the additional property $(*)$ that makes $(T,f)$ a minimal Gomory-Hu Steiner tree. Fix $t\in U\setminus\{r\}$, and let $(u,v)$ and $U'$ be defined as in $(*)$, i.e., $(u,v)$ is the minimum-weight edge $(u,v)$ on the $r-t$ path that is closest to $t$, and $U'$ is the vertices of the connected component of $T-(u,v)$ containing $r$. Since $(T,f)$ is a Gomory-Hu Steiner tree for $G$, we have that $f^{-1}(U')$ is an $(r,t)$-mincut of value $w_T(u,v)$. Suppose for contradiction that $f^{-1}(U')$ is not a \emph{minimal} $(r,t)$-mincut. Then, there exists $S\subsetneq f^{-1}(U')$ such that $S$ is also an $(r,t)$-mincut. By construction of $G'$, $\delta_{G'}S=\delta_GS+|S|\epsilon$ and $\delta_{G'}(f^{-1}(U')) = \delta_G(f^{-1}(U'))+|f^{-1}(U')|\epsilon$. We have $\delta_GS=\delta_G(f^{-1}(U'))$ and $|S|<|f^{-1}(U')|$, so $\delta_{G'}S<\delta_{G'}(f^{-1}(U'))$. In other words, $f^{-1}(U')$ is not an $(r,t)$-mincut in $G'$, contradicting the fact that $(T',f)$ is a Gomory-Hu Steiner tree for $G'$. Therefore, property $(*)$ is satisfied, concluding the proof.
\end{proof}

\subsection{A Single Recursive Step}\label{sec:ghtree:step}
Before we present Algorithm~\ref{ghtree}, we first consider the subprocedure \ref{ghtreestep} that it uses on each recursive step.

\begin{algorithm}
\mylabel{ghtreestep}{\textsc{GHTreeStep}}\caption{\ref{ghtreestep}$(G=(V,E,w),s,U)$} 
\begin{enumerate}
\item Initialize $R^0\gets U$ and $D\gets\emptyset$
\item For all $i$ from $0$ to $\lfloor\lg|U|\rfloor$ do:
 \begin{enumerate}
 \item Call the Isolating Cuts Lemma (\Cref{lem:iso cut}) on the singleton sets $\{v\}$ for $v\in R^i$, obtaining disjoint sets $S^i_v$ (the minimal $(v,R^i\setminus v)$-mincut) for each $v\in R^i$. \label{line:Sv}
 \item Call \Cref{problem:ssmc-no-promise} on graph $G$ and source $s$.
 \item Let $D^i\subseteq R^i$ be the union of $S^i_v\cap U$ over all $v\in R^i\setminus\{s\}$ satisfying $\delta S^i_v=\lambda(s,v)$ and $|S^i_v\cap U|\le|U|/2$\label{line:D}
 \item $R^{i+1}\gets$ subsample of $R^i$ where each vertex in $R^i\setminus \{s\}$ is sampled independently with probability $1/2$, and $s$ is sampled with probability $1$
 \end{enumerate}
\item Return the largest set $D^i$ and the corresponding sets $S^i_v$ over all $v\in R^i\setminus\{s\}$ satisfying the conditions in line~\ref{line:D} 
\end{enumerate}
\end{algorithm}

Let $D=D^0\cup D^1\cup \cdots\cup D^{\lfloor\lg|U|\rfloor}$ be the union of the sets $D^i$ as defined in Algorithm~\ref{ghtreestep}. Let $D^*$ be all vertices $v\in U\setminus \{s\}$ for which there exists an $(s,v)$-mincut whose $v$ side has at most $|U|/2$ vertices in $U$. We now claim that $D$ covers a large fraction of vertices in $D^*$ in expectation.

\begin{lemma}\label{lem:ghtree:step}
 $\mathbb E[|D\cap D^*|] = \Omega(|D^*|/\log|U|)$. 
\end{lemma}
\begin{proof}

Consider a rooted minimal Steiner Gomory-Hu tree $T$ of $G$ on terminals $U$ rooted at $s$, which exists by \Cref{thm:ghtree:rooted}. For each vertex $v\in U\setminus \{s\}$, let $r(v)$ be defined as the child vertex of the lowest weight edge on the path from $v$ to $s$ in $T$. If there are multiple lowest weight edges, choose the one with the maximum depth. 


For each vertex $v\in D^*$, consider the subtree rooted at $v$, define $U_v$ to be the vertices in the subtree, and define $n_v$ as the number of vertices in the subtree. We say that a vertex $v\in D^*$ is \emph{active} if $v\in R^i$ for $i=\lfloor\lg n_{r(v)}\rfloor$. In addition, if $U_{r(v)}\cap R^i=\{v\}$, then we say that $v$ \emph{hits} all of the vertices in $U_{r(v)}$ (including itself); see Figure~\ref{fig:hits}. In particular, in order for $v$ to hit any other vertex, it must be active. For completeness, we say that any vertex in $U\setminus D^*$ is not active and does not hit any vertex.

\begin{figure}\centering
\includegraphics[scale=1]{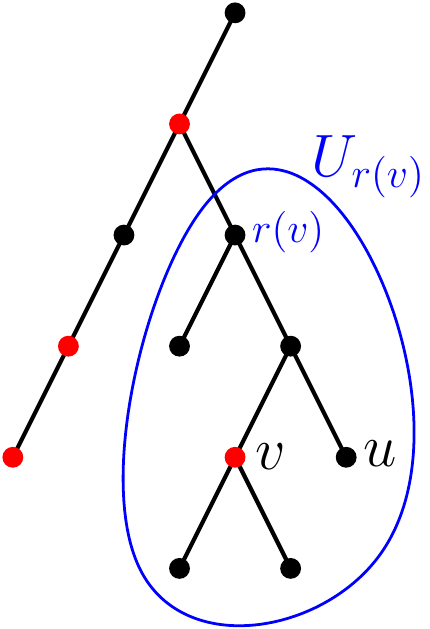}
\caption{Let $i=\lfloor\lg n_{r(v)}\rfloor=\lfloor\lg 7\rfloor=2$, and let the red vertices be those sampled in $R^2$. Vertex $v$ is active and hits $u$ because $v$ is the only vertex in $U_{r(v)}$ that is red.}\label{fig:hits}
\end{figure}

To prove that $\mathbb E[|D|] \ge \Omega(|D^*|/\log|U|)$, we will show that
 \begin{enumerate}
 \item[(a)] each vertex $u$ that is hit is in $D$, 
 \item[(b)] the total number of pairs $(u,v)$ for which $v\in D^*$ hits $u$ is at least $c |D^*|$ in expectation for some small enough constant $c>0$, and
 \item[(c)] each vertex $u$ is hit by at most $\lfloor\lg|U|\rfloor+1$ vertices. 
 \end{enumerate}

For (a), consider the vertex $v$ that hits $u$. By definition, for $i=\lfloor\lg n_{r(v)}\rfloor$, we have $U_{r(v)}\cap R^i=\{v\}$, so $f^{-1}(U_{r(v)})$ is a $(v,R^i\setminus \{v\})$-cut.  By the definition of $r(v)$, we have that $f^{-1}(U_{r(v)})$ is a $(v,s)$-mincut. On the other hand, we have that $S^i_v$ is a $(v,R^i\setminus\{v\})$-mincut, so in particular, it is a $(v,s)$-cut. It follows that $f^{-1}(U_{r(v)})$ and $S^i_v$ are both $(v,s)$-mincuts and $(v,R^i\setminus v)$-mincuts, and $\delta S^i_v=\lambda(s,v)\le W$. Since $T$ is a minimal Gomory-Hu Steiner tree, we must have $f^{-1}(U_{r(v)}) \subseteq S^i_v$. Since $S^i_v$ is the minimal $(v,R^i\setminus\{v\})$-mincut, it is also the minimal $(v,s)$-mincut, so $S^i_v\subseteq f^{-1}(U_{r(v)}) $. It follows that $f^{-1}(U_{r(v)})=S^i_v$. Since $f^{-1}(U_{r(v)})$ is the minimal $(v,s)$-mincut and $v\in D^*$, we must have $|f^{-1}(U_{r(v)})\cap U|\le z$, so in particular, $|S^i_v\cap U|=|f^{-1}(U_{r(v)})\cap U|\le z$. Therefore, the vertex $v$ satisfies all the conditions of line~\ref{line:D}. Moreover, since $u\in U_{r(v)}\subseteq f^{-1}(U_{r(v)})= S^i_v$, vertex $u$ is added to $D$ in the set $S^i_v\cap U$. 

For (b), for $i=\lfloor\lg n_{r(v)}\rfloor$, we have $v\in R^i$ with probability exactly $1/2^i = \Theta(1/n_{r(v)})$, and with probability $\Omega(1)$, no other vertex in $U_{r(v)}$ joins $R^i$. Therefore, $v$ is active with probability $\Omega(1/n_{r(v)})$. Conditioned on $v$ being active, it hits exactly $n_{r(v)}$ many vertices. It follows that $v$ hits $\Omega(1)$ vertices in expectation.

For (c), since the isolating cuts $S^i_v$ over $v\in R^i$ are disjoint for each $i$, each vertex is hit at most once on each iteration $i$. Since there are $\lfloor\lg|U|\rfloor+1$ many iterations, the property follows. 

Finally, we show why properties (a) to (c) imply $\mathbb E[|D\cap D^*|] \ge \Omega(|D^*|/\log|U|)$. By property~(b), the number of times some vertex hits another vertex is $\Omega(|D^*|)$ in expectation. Since each vertex is hit at most $O(\log|U|)$ times by property~(c), there are at least $\Omega(|D^*|/\log|U|)$ vertices hit in expectation, all of which are included in $D$ by property~(a).
\end{proof}

The corollary below immediately follows. Note that the sets $S^i_v$ output by the algorithm are disjoint, which we require in the recursive GH tree algorithm. 
\begin{corollary}\label{cor:step}
The largest set $D^i$ returned by \ref{ghtreestep} satisfies $\mathbb E[|D^i\cap D^*|] = \Omega(|D^*|/\log^2|U|)$.
\end{corollary}

\subsection{The Gomory-Hu Tree Algorithm}
The Gomory-Hu tree algorithm is presented in  \ref{ghtree}, which uses \ref{ghtreestep} as a subprocedure on each recursive step.

\paragraph{Correctness.}
Algorithm~\ref{ghtree} has the same recursive structure as Gomory and Hu's original algorithm, except that it computes multiple mincuts on each step. Therefore, correctness of the algorithm follows similarly to their analysis. For completeness, we include it below.

\begin{algorithm}
\mylabel{ghtree}{\textsc{GHTree}}\caption{\ref{ghtree}$(G=(V,E,w),U)$} 
\begin{enumerate}
\item $s\gets$ uniformly random vertex in $U$
\item Call $\ref{ghtreestep}(G,s,U)$ to obtain $D^i$ and the sets $S^i_v$ (so that $D^i=\bigcup S^i_v\cap U$)

\ \label{line:max}
\item For each set $S^i_v$ do: \Comment{Construct recursive graphs and apply recursion}
 \begin{enumerate}
 \item Let $G_v$ be the graph $G$ with vertices $V\setminus S^i_v$ contracted to a single vertex $x_v$ \Comment{$S^i_v$ are disjoint}
 \item Let $U_v\gets S^i_v\cap U$
 \item If $|U_v|>1$, then recursively set $(T_v,f_v)\gets\ref{ghtree}(G_v,U_v)$
 \end{enumerate}
\item Let $G_\lar$ be the graph $G$ with (disjoint) vertex sets $S^i_v$ contracted to single vertices $y_v$ for all $v\in D^i$
\item Let $U_\lar\gets U\setminus D^i$
\item If $|U_v|>1$, then recursively set $(T_\lar,f_\lar)\gets\ref{ghtree}(G_\lar,U_\lar)$

\item Combine $(T_\lar,f_\lar)$ and $\{(T_v,f_v):v\in D^i\}$ into $(T,f)$ according to \ref{combine}

\item Return $(T,f)$

\end{enumerate}
\end{algorithm}

\begin{algorithm}
\mylabel{combine}{\textsc{Combine}}\caption{\ref{combine}$((T_\lar,f_\lar),\{(T_v,f_v): v\in R^i\} )$} 
\begin{algorithmic}[1]
\State Construct $T$ by starting with the disjoint union $T_\lar\cup\bigcup_{v\in R^i}T_v$ and, for each $v\in R^i$, adding an edge between $f_v(x_v)\in U_v$ and $f_\lar(y_v)\in U_\lar$ of weight $\delta_GS^i_v$\label{line:combine-T}
\State Construct $f:V\to U$ by $f(v')=f_\lar(v')$ if $v'\in U_\lar$ and $f(v')=f_v(v')$ if $v'\in U_v$ for some $v\in R^i$\label{line:combine-f}
\State\Return $(T,f)$
\end{algorithmic}
\end{algorithm}

\begin{lemma}\label{lem:ghtree:correctness}
Algorithm~\ref{ghtree}$(G=(V,E,w),U)$ outputs a Gomory-Hu Steiner tree.
\end{lemma}

To prove \Cref{lem:ghtree:correctness}, we first introduce a helper lemma.

\begin{lemma}\label{lem:ghtree:exact}
For any distinct vertices $p,q\in U_\lar$, we have $\lambda_{G_\lar}(p,q) = \lambda_G(p,q)$. The same holds with $U_\lar$ and $G_\lar$ replaced by $U_v$ and $G_v$ for any $v\in D^i$.
\end{lemma}
\begin{proof}
Since $G_\lar$ is a contraction of $G$, we have $\lambda_{G_\lar}(p,q) \ge \lambda_G(p,q)$. To show the reverse inequality, fix any $(p,q)$-mincut in $G$, and let $S$ be one side of the mincut. We show that for each $v\in  R^i$, either $S^i_v\subseteq S$ or $S^i_v\subseteq V\setminus S$. Assuming this, the cut $E(S,V\setminus S)$ stays intact when the sets $S^i_v$ are contracted to form $G_\lar$, so $\lambda_{G_\lar}(p,q) \le \delta S = \lambda_G(p,q)$.

Consider any $v\in R^i$, and suppose first that $v\in S$. Then, $S^i_v\cap S$ is still a $(v,R^i\setminus v)$-cut, and $S^i_v\cup S$ is still a $(p,q)$-cut. By the submodularity of cuts,
\[ \delta_GS^i_v + \delta_GS \ge \delta_G(S^i_v\cup S) + \delta_G(S^i_v\cap S). \]
In particular, $S^i_v\cap S$ must be a minimum $(v,R^i\setminus v)$-cut. Since $S^i_v$ is the minimal $(v,R^i\setminus v)$-mincut, it follows that $S^i_v\cap S = S^i_v$, or equivalently, $S^i_v\subseteq S$.

Suppose now that $v\notin S$. In this case, we can swap $p$ and $q$, and swap $S$ and $V\setminus S$, and repeat the above argument to get $S^i_v\subseteq V\setminus S$.

The argument for $U_v$ and $G_v$ is identical, and we skip the details.
\end{proof}

\begin{proof}[Proof (\Cref{lem:ghtree:correctness}).]
We apply induction on $|U|$. 
By induction, the recursive outputs $(T_\lar,f_\lar)$ and $(T_v,f_v)$ are Gomory-Hu Steiner trees. By definition, this means that for all $x,y\in U_\lar$ and the minimum-weight edge $(u,u')$ on the $x$--$y$ path in $T_\lar$, letting $U'_\lar\subseteq U_\lar$ be the vertices of the connected component of $T_\lar-(u,u')$ containing $x$, we have that $f^{-1}_\lar(U'_\lar)$ is an $(s,t)$-mincut in $G_\lar$ with value is $w_T(u,u')$. Define $U'\subseteq U$ as the vertices of the connected component of $T-(u,u')$ containing $x$. By construction of $(T,f)$ (lines~\ref{line:combine-T}~and~\ref{line:combine-f}), the set $f^{-1}(U')$ is simply $f^{-1}_\lar(U'_\lar)$ with the vertex $x_\lar$ replaced by $V\setminus S^i_\lar$ in the case that $x_\lar\in f^{-1}(U')$. Since $G_\lar$ is simply $G$ with all vertices $V\setminus S^i_\lar$ contracted to $x_\lar$, we conclude that $\delta_{G_\lar}(f^{-1}_\lar(U'_\lar)) = \delta_G( f^{-1}(U'))$. By \Cref{lem:ghtree:exact}, we have $\lambda_G(x,y)=\lambda_{G_\lar}(x,y)$ are equal, so $\delta_G(f^{-1}(U'))$ is an $(x,y)$-mincut in $G$. In other words, the Gomory-Hu Steiner tree condition for $(T,f)$ is satisfied for all $x,y\in U_\lar$. A similar argument handles the case $x,y\in U_v$ for some $v\in R^i$.

There are two remaining cases: $x\in U_v$ and $y\in U_\lar$, and $x\in U_v$ and $y\in U_{v'}$ for distinct $v,v'\in R^i$. Suppose first that $x\in U_v$ and $y\in U_\lar$. By considering which sides $v$ and $s$ lie on the $(x,y)$-mincut, we have
\[ \delta_GS=\lambda(x,y)\ge\min\{\lambda(x,v),\lambda(v,s),\lambda(s,y)\} .\]
We now case on which of the three mincut values $\lambda(x,y)$ is greater than or equal to. 

\begin{enumerate}
\item If $\lambda(x,y)\ge\lambda(v,s)$, then since $S^i_v$ is a $(v,s)$-mincut that is also an $(x,y)$-cut, we have $\lambda(x,y)=\lambda(v,s)$. By construction, the edge $(f_v(x_v),f_\lar(y_v))$ of weight $\delta_GS^i_v=\delta_GS$ is on the $x-y$ path in $T$. There cannot be edges on the $x-t$ path in $T$ of smaller weight, since each edge corresponds to a $(s,t)$-cut in $G$ of the same weight. Therefore, $(f_v(x_v),f_\lar(y_v))$ is the minimum-weight edge on the $s$--$t$ path in $T$.\label{ghtree:case1}
\item Suppose now that $\lambda(x,v)\le \lambda(x,y)<\lambda(v,s)$. The minimum-weight edge $e$ on the $x-v$ path in $T_v$ has weight $\lambda(x,v)$. This edge $e$ cannot be on the $v-f_v(x_v)$ path in $T_v$, since otherwise, we would obtain a $(v,x_v)$-cut of value $\lambda(x,v)$ in $G_v$, which becomes a $(v,s)$-cut in $G$ after expanding the contracted vertex $x_v$; this contradicts our assumption that $\lambda(x,v)<\lambda(v,s)$. It follows that $e$ is on the $x-f_v(x_v)$ path in $T_v$ which, by construction, is also on the $x-y$ path in $T$. Once again, the $x-y$ path cannot contain an edge of smaller weight. \label{ghtree:case2}
\item The final case $\lambda(s,y)\le\lambda(x,y)<\lambda(v,s)$ is symmetric to case~\ref{ghtree:case2}, except we argue on $T_\lar$ and $G_\lar$ instead of $T_v$ and $G_v$.
\end{enumerate}

Suppose now that $x\in U_v$ and $y\in U_{v'}$ for distinct $v,v'\in R^i$. By considering which sides $v,v',s$ lie on the $(x,y)$-mincut, we have
\[ \delta_GS=\lambda(x,y)\ge\min\{\lambda(x,v),\lambda(v,s),\lambda(s,v'),\lambda(v',y)\} .\]
We now case on which of the four mincut values $\lambda(x,y)$ is greater than or equal to. 
\begin{enumerate}
\item If $\lambda(x,y)\ge\lambda(v,s)$ or $\lambda(x,y)\ge\lambda(s,v')$, then the argument is the same as case~\ref{ghtree:case1} above.
\item If $\lambda(x,v)\le \lambda(x,y)<\lambda(v,s)$ or $\lambda(y,v')\le\lambda(x,y)<\lambda(v',s)$, then the argument is the same as case~\ref{ghtree:case2} above.
\end{enumerate}
This concludes all cases, and hence the proof.
\end{proof}

\paragraph{Running time.} We now bound the running time of \ref{ghtree}.

\begin{lemma}\label{lem:ghtree:depth}
W.h.p., the algorithm \ref{ghtree} has maximum recursion depth $O(\log^3n)$.
\end{lemma}
\begin{proof}
By construction, each recursive instance $(G_v,U_v)$ has $|U_v|\le|U|/2$.
We use the following lemma from~\cite{AKT20_b}.

\begin{lemma}\label{lem:ghtree:random-s}
Suppose the source vertex $s\in U$ is chosen uniformly at random. Then, $\mathbb E[|D^*|]=\Omega(|U|-1)$.
\end{lemma}

\noindent By \Cref{cor:step} and \Cref{lem:ghtree:random-s}, over the randomness of $s$ and \ref{ghtreestep}, we have
\[ \mathbb E[D^i]\ge \Omega(\mathbb E[|D^*|]/\log^2|U|) \ge \Omega((|U|-1)/\log^2|U|) ,\]
so the recursive instance $(G_\lar,U_\lar)$ satisfies $\mathbb E[|U_\lar|]\le(1-1/\log^2|U|)\cdot(|U|-1)$. Therefore, each recursive branch either has at most half the vertices in $U$, or has at most a $(1-1/\log^2|U|)$ fraction in expectation. It follows that w.h.p., all branches terminate by $O(\log^3n)$ recursive calls.
\end{proof}

\begin{lemma}\label{lem:run}
For a weighted graph $G=(V,E,w)$ and terminals $U\subseteq V$, $\ref{ghtree}(G,V)$ takes time $\tilde{O}(m)$ plus calls to max-flow on instances with a total of $\tilde{O}(n)$ vertices and $\tilde{O}(m)$ edges.
\end{lemma}
\begin{proof}
For a given recursion level, consider the instances $\{ (G_i,U_i,W_i)\}$ across that level. By construction, the terminals $U_i$ partition $U$. Moreover, the total number of vertices over all $G_i$ is at most $n+2(|U|-1)=O(n)$ since each branch creates $2$ new vertices and there are at most $|U|-1$ branches. 

To bound the total number of edges, define a \emph{parent} vertex in an instance as a vertex resulting from either (1)~contracting $V\setminus S^i_v$ in some previous recursive $G_v$ call, or (2)~contracting a component containing a parent vertex in some previous recursive call. There are at most $O(\log n)$ parent vertices: at most $O(\log n)$ can be created by~(1) since each $G_v$ call decreases $|U|$ by a constant factor, and (2)~cannot increase the number of parent vertices. Therefore, the total number of edges adjacent to parent vertices is at most $O(\log n)$ times the number of vertices. Since there are $O(n)$ vertices in a given recursion level, the total number of edges adjacent to parent vertices is $O(n\log n)$ in this level. Next, we bound the number of edges not adjacent to a parent vertex by $m$. To do so, we first show that on each instance, the total number of these edges over all recursive calls produced by this instance is at most the total number of such edges in this instance. Let $P\subseteq V$ be the parent vertices; then, each $G_v$ call has exactly $|E(G[S^i_v\setminus P])|$ edges not adjacent to parent vertices (in the recursive instance), and the $G_\lar$ call has at most $|E(G[V\setminus P]) \setminus \bigcup_vE(G[S^i_v\setminus P])|$, and these sum to $|E(G[V\setminus P])|$, as promised. This implies that the total number of edges not adjacent to a parent vertex at the next level is at most the total number at the previous level. Since the total number at the first level is $m$, the bound follows.

Therefore, there are $O(n)$ vertices and $\tilde{O}(m)$ edges in each recursion level. By \Cref{lem:ghtree:depth}, there are $O(\epsilon^{-1}\log^4n)$ levels, for a total of $\tilde{O}(n\epsilon^{-1})$ vertices and $\tilde{O}(m\epsilon^{-1})$ edges. In particular, the instances to the max-flow calls have $\tilde{O}(n\epsilon^{-1})$ vertices and $\tilde{O}(m\epsilon^{-1})$ edges in total.
\end{proof}

Together, \Cref{lem:ghtree:correctness} (correctness) and \Cref{lem:run} (running time) prove \Cref{thm:ghtree:reduction}.

\end{document}